\newif\ifnotes
\notesfalse
\newif\ifsubmission
\submissionfalse

\ifsubmission
\notesfalse
\else
\fi

\def\showtableofcontents{1}

\documentclass[11pt]{article}
\usepackage{fullpage}

\usepackage[utf8]{inputenc}
\usepackage{multirow}

\usepackage[shortlabels]{enumitem}
\usepackage{amsmath}
\usepackage{amssymb}
\usepackage{amsthm}
\usepackage[dvipsnames]{xcolor}
\usepackage{mathtools}
\usepackage[utf8]{inputenc}
\usepackage{amsfonts}
\usepackage{graphicx}
\usepackage{mathrsfs}
\usepackage{physics}
\usepackage{soul}
\usepackage{bm} % allows for writing mixed states \rho in bold
\usepackage[ruled,linesnumbered]{algorithm2e}
\usepackage[T1]{fontenc}
\usepackage{quantikz} % for making quantum circuits
\usepackage[splitrule]{footmisc}
\usepackage[normalem]{ulem} % for striking out text with \sout{text}

\usepackage{beton}

\definecolor{DarkBlue}{RGB}{0,0,150}
\definecolor{DarkRed}{RGB}{150,0,0}
\definecolor{DarkGreen}{RGB}{0,150,0}
\usepackage[colorlinks,linkcolor=Red,citecolor=DarkBlue]{hyperref}
\usepackage[capitalize]{cleveref}
  \crefname{step}{Step}{Steps}
\usepackage[hyperpageref]{backref}

\newcommand{\authnote}[3]{\textcolor{#3}{[{\footnotesize {\bf #1:} { {#2}}}]}}
\newcommand{\fermi}[1]{\ifnotes \authnote{Fermi}{#1}{BlueViolet} \fi}

\usepackage{tikz}
\usetikzlibrary{patterns}

\newtheorem{theorem}{Theorem}[section]

\newtheorem{claim}[theorem]{Claim}
\newtheorem{lemma}[theorem]{Lemma}
\newtheorem{conjecture}{Conjecture}
\newtheorem{corollary}[theorem]{Corollary}
\newtheorem{property}[theorem]{Property}

\newtheorem{remark}[theorem]{Remark}

\newtheorem{definition}[theorem]{Definition}

\newenvironment{theorem-restated}[1]{
  \renewcommand\thetheorem{\ref{#1}}\theorem}{\endtheorem\addtocounter{theorem}{-1}}
\newenvironment{definition-restated}[1]{
  \renewcommand\thetheorem{\ref{#1}}\definition}{\endtheorem\addtocounter{theorem}{-1}}
\newenvironment{corollary-restated}[1]{
  \renewcommand\thetheorem{\ref{#1}}\corollary}{\endtheorem\addtocounter{theorem}{-1}}
\newenvironment{lemma-restated}[1]{
  \renewcommand\thetheorem{\ref{#1}}\lemma}{\endtheorem\addtocounter{theorem}{-1}}

\Crefname{importedtheorem}{Imported Theorem}{Imported Theorems}
\Crefname{theorem}{Theorem}{Theorems}
\Crefname{proposition}{Proposition}{Propositions}
\Crefname{claim}{Claim}{Claims}
\Crefname{lemma}{Lemma}{Lemmas}
\Crefname{conjecture}{Conjecture}{Conjectures}
\Crefname{corollary}{Corollary}{Corollaries}
\Crefname{construction}{Construction}{Constructions}
\Crefname{property}{Property}{Properties}
\Crefname{game}{Game}{Games}
\Crefname{item}{Item}{Items}
\Crefname{algorithm}{Algorithm}{Algorithms}

\theoremstyle{definition}

\Crefname{definition}{Definition}{Definitions}
\Crefname{assumption}{Assumption}{Assumptions}
\Crefname{notation}{Notation}{Notations}

\theoremstyle{remark}

\Crefname{question}{Question}{Questions}
\Crefname{remark}{Remark}{Remarks}
\Crefname{comment}{Comment}{Comments}
\Crefname{fact}{Fact}{Facts}

\providecommand{\Dec}{\mathsf{Dec}}

\providecommand{\Ext}{\mathsf{Ext}}
\providecommand{\from}{\leftarrow}
\def\cA{{\cal A}}
\def\cB{{\cal B}}
\def\cC{{\cal C}}
\def\cD{{\cal D}}

\def\cH{{\cal H}}

\def\cM{{\cal M}}

\def\cO{{\cal O}}

\def\cX{{\cal X}}
\def\cY{{\cal Y}}

\providecommand{\bfD}{\mathbf{D}}

\providecommand{\bfS}{\mathbf{S}}

\providecommand{\RegA}{\mathsf{A}}
\providecommand{\RegB}{\mathsf{B}}
\providecommand{\RegC}{\mathsf{C}}
\providecommand{\RegD}{\mathsf{D}}
\providecommand{\RegE}{\mathsf{E}}
\providecommand{\RegF}{\mathsf{F}}

\providecommand{\RegH}{\mathsf{H}}

\providecommand{\RegK}{\mathsf{K}}

\providecommand{\RegM}{\mathsf{M}}

\providecommand{\RegP}{\mathsf{P}}
\providecommand{\RegQ}{\mathsf{Q}}
\providecommand{\RegR}{\mathsf{R}}
\providecommand{\RegS}{\mathsf{S}}

\providecommand{\RegV}{\mathsf{V}}
\providecommand{\RegW}{\mathsf{W}}
\providecommand{\RegX}{\mathsf{X}}
\providecommand{\RegY}{\mathsf{Y}}
\providecommand{\RegZ}{\mathsf{Z}}

\providecommand{\sfPi}{\mathsf{\Pi}}

\providecommand{\QMA}{\textup{QMA}}
\providecommand{\BQP}{\textup{BQP}}
\providecommand{\NP}{\mathsf{NP}}

\providecommand{\secp}{\lambda}
\providecommand{\poly}{\mathsf{poly}}
\providecommand{\polylog}{\mathsf{polylog}}
\providecommand{\negl}{\mathsf{negl}}

\providecommand{\fail}{\mathsf{fail}}

\providecommand{\Recover}{\mathsf{Recover}}

\providecommand{\SwapRecover}{\mathsf{SwapRecover}}
\providecommand{\SwapDiff}{\mathsf{SwapDiff}}
\providecommand{\sibling}{\mathsf{sib}}
\providecommand{\parent}{\mathsf{par}}

\providecommand{\Adv}{\mathsf{Adv}}

\providecommand{\Com}{\mathsf{Com}}
\providecommand{\QSC}{\mathsf{QSC}}
\providecommand{\SQSC}{\mathsf{sQSC}}
\providecommand{\QBC}{\mathsf{QBC}}

\providecommand{\PCP}{\mathsf{PCP}}
\providecommand{\ZKPCP}{\mathsf{zkPCP}}

\providecommand{\D}{\mathsf{D}}

\providecommand{\Extract}{\mathsf{Extract}}

\providecommand{\Reduction}{\mathsf{Reduction}}
\providecommand{\HybExtract}{\mathsf{HybExtract}}
\providecommand{\KnowledgeExt}{\mathsf{KnowledgeExt}}
\providecommand{\BitExtract}{\mathsf{BitExtract}}
\providecommand{\Repair}{\mathsf{Repair}}

\providecommand{\QEnc}{\mathsf{QEnc}}

\providecommand{\Path}{\mathsf{Path}}

\providecommand{\decom}{\mathsf{decom}}
\providecommand{\Sym}{\mathsf{Sym}}
\providecommand{\Antisym}{\mathsf{Antisym}}

\providecommand{\Prep}{\mathsf{Prep}}

\providecommand{\anc}{\mathsf{anc}}
\providecommand{\flag}{\mathsf{flag}}

\providecommand{\SQUARG}{\mathsf{SQUARG}}
\providecommand{\QSigma}{\mathsf{QSigma}}

\providecommand{\unif}{\mathrm{unif}}
\providecommand{\win}{\mathrm{win}}
\providecommand{\fail}{\mathrm{fail}}
\providecommand{\accept}{\mathrm{accept}}

\providecommand{\SWAP}{\mathsf{SWAP}}

\providecommand{\Id}{\mathbb{I}}

\providecommand{\Sch}{\mathsf{Sch}}
\providecommand{\USch}{\mathbf{U}_{\Sch}}
\providecommand{\Expand}{\mathsf{Expand}}

\providecommand{\ctl}{\mathsf{ctl}}
\providecommand{\MD}{\mathsf{MD}}
\providecommand{\brho}{\bm{\rho}}
\providecommand{\bpi}{\bm{\pi}}
\providecommand{\bsigma}{\bm{\sigma}}
\providecommand{\btau}{\bm{\tau}}

\providecommand{\tensor}{\otimes}

\DeclareMathSymbol{\mh}{\mathord}{operators}{`\-}

\providecommand{\Est}{\mathsf{Est}}

\DeclareMathOperator*{\image}{\mathsf{image}}

\title{Commitments to Quantum States}

\ifsubmission
\author{}
\else
\author{Sam Gunn\thanks{UC Berkeley. Email: \texttt{gunn@berkeley.edu}. Supported by a Google PhD Fellowship.} \and Nathan Ju\thanks{UC Berkeley. Email: \texttt{nju@berkeley.edu}.} \and Fermi Ma\thanks{Simons Institute $\&$ UC Berkeley. Email: \texttt{fermima@alum.mit.edu}.} \and Mark Zhandry\thanks{NTT Research $\&$ Princeton University. Email: \texttt{mzhandry@gmail.com}.}}
\fi
\date{\today}

\begin{document}

\maketitle

\begin{abstract} 

What does it mean to commit to a quantum state? In this work, we propose a simple answer: a commitment to quantum messages is \emph{binding} if, after the commit phase, the committed state is \emph{hidden} from the sender's view. We accompany this new definition with several instantiations. We build the first non-interactive \emph{succinct} quantum state commitments, which can be seen as an analogue of collision-resistant hashing for quantum messages. We also show that hiding quantum state commitments (QSCs) are implied by any commitment scheme for classical messages. All of our constructions can be based on quantum-cryptographic assumptions that are implied by \emph{but are potentially weaker than} one-way functions.

Commitments to quantum states open the door to many new cryptographic possibilities. Our flagship application of a succinct QSC is a quantum-communication version of Kilian's succinct arguments for any language that has quantum PCPs with constant error and polylogarithmic locality. Plugging in the PCP theorem, this yields succinct arguments for $\NP$ under significantly weaker assumptions than required classically; moreover, if the quantum PCP conjecture holds, this extends to $\QMA$. At the heart of our security proof is a new rewinding technique for extracting quantum information.

\end{abstract}

\ifnum\showtableofcontents=1
{
\thispagestyle{empty}
\newpage
\pagenumbering{roman}
\setcounter{tocdepth}{2}
{\small\tableofcontents}
\newpage
\pagenumbering{arabic}
}

\section{Introduction}
\label{sec:intro}

Is it possible to commit to a quantum state? There is a simple ``folklore'' proposal: apply a quantum one-time pad $X^rZ^s$ to the state $\brho$, and send the resulting state along with a commitment to the classical string $(r,s)$. While this scheme has appeared implicitly in several quantum cryptographic protocols~\cite{FOCS:BJSW16,C:ColVidZha20,FOCS:BroGri20}, our understanding of commitments to quantum states is extremely limited. For example, consider the following basic questions:

\begin{itemize}
    \item How should binding even be \emph{defined} for commitments to quantum messages? A proper definition of binding for quantum messages should capture the semantic meaning of \emph{binding}, while also being useful for cryptographic applications.
    
    \item The commitment in the folklore construction is always larger than the original message. Is it possible to \emph{succinctly} commit to a quantum state? (For example, is it possible to commit to a $2n$-qubit state with an $n$-qubit commitment?)
    
    \item What kinds of cryptographic protocols are enabled by commitments to quantum states?
\end{itemize}

More broadly, commitments are central to cryptography and appear in many different forms and contexts. Does an analogous picture exist in the quantum setting?

\paragraph{This work.} We initiate a formal study of quantum-state commitments (QSCs) and their applications. We develop general techniques to answer the above questions and position QSCs to play a vital role in quantum cryptography. Our contributions are the following:
\begin{itemize}
    \item \textbf{Definitions.} We provide a definition of binding for QSCs: informally, committing to a quantum state $\brho$ should \emph{erase} it from the sender's view. Our definition captures computational and statistical security, composes across multiple commitments, and we show that it naturally generalizes Unruh's collapse-binding definition for (quantum-secure) \emph{classical} commitments~\cite{EC:Unruh16}. Perhaps most surprisingly, we show that the notions of binding and hiding for QSCs satisfy a \emph{duality} that has no known analogue in classical cryptography. 
    \item \textbf{Constructions.} Our new definition directly enables new constructions. We build the first non-interactive \emph{succinct} QSCs, which can be seen as an analogue of collision-resistant hashing for quantum messages. We also formalize the ``folklore'' proposal to show that hiding and binding QSCs exist if and only if hiding and binding quantum bit commitments (QBCs) exist (i.e., commitments to \emph{classical} messages). All of our constructions can be based on assumptions that are implied by \emph{but are potentially weaker than one-way functions}.
    \item \textbf{Applications.} Finally, we use QSCs to build protocols. Our flagship application of a succinct QSC is a quantum version of Kilian's succinct arguments~\cite{STOC:Kilian92,ARXIV:CheMov21} for any language that has quantum PCPs with constant error and polylogarithmic locality. Plugging in the PCP theorem, this yields succinct arguments for $\NP$ from significantly weaker assumptions than in the classical setting; moreover, if the quantum PCP conjecture holds, this result extends to all of $\QMA$. Proving security is the most technical component of this work. Our proof develops a new rewinding technique for extracting quantum information, combining ideas from~\cite{FOCS:CMSZ21} (which handles quantum attacks on \emph{classical} protocols) with new tools and abstractions designed for quantum protocols.

\end{itemize}

We now give a more detailed description of our results.

\subsection{Our results}

\subsubsection{Definitions}
\label{subsubsec:intro-defs}

\paragraph{Syntax for quantum state commitments.} Before we present our binding definition, we establish some basic syntax for \emph{non-interactive} QSCs:

\begin{itemize}
    \item The sender commits to a state $\ket*{\psi}$ by applying a public unitary $\Com$ to $\ket{\psi} \ket*{0^\lambda}$, obtaining a state on registers $(\RegC, \RegD)$. The commitment is the state on register $\RegC$.
    \item The sender can later open the commitment by sending the register $\RegD$. The receiver verifies by applying $\Com^\dagger$ and checking that the last $\lambda$ qubits are $0$. If verification succeeds, the receiver recovers the committed state $\ket{\psi}$ from the unmeasured registers.\footnote{A similar \emph{syntax} appeared in prior work of Chen and Movassagh~\cite{ARXIV:CheMov21}, who proposed applying a Haar random unitary to $\ket{\psi}\ket{0^{\lambda}}$. However, since their goal was to build quantum tree commitments (see~\cref{subsec:related} for further discussion), they did not formalize this syntax for individual commitments.}
\end{itemize}

There is a simple transformation that turns any interactive QSC into a non-interactive one: define $\Com$ to run the honest interactive commit phase coherently (see~\cref{sec:non-interactive} for more details).\footnote{In the setting of quantum bit commitments, i.e., quantum commitments to \emph{one-bit classical messages}, a similar transformation and syntax for non-interactive commitments previously appeared in~\cite{AC:Yan22}.} We will therefore focus our attention on non-interactive QSCs.

\paragraph{Swap binding: a definition for quantum messages.} Our new binding definition --- which we call \emph{swap binding} --- requires that once the adversary has sent the commitment $\RegC$, it can no longer distinguish its original decommitment from one where the committed message has been swapped with junk. In other words, sending $\RegC$ erases the committed message from the sender's view. We make this concrete with the following security game:

\begin{enumerate}
    \item The adversary sends $(\RegC, \RegD)$ to the challenger.
    \item The challenger applies $\Com^\dagger$ and verifies that the last $\lambda$ qubits are $0$, and if not, aborts.\footnote{Since a successful verification measurement completely collapses the last $\lambda$ qubits, it would be equivalent to consider a security game where, instead of sending $(\RegC,\RegD)$, the adversary sends the challenger the committed message in the clear; indeed, this fact is crucial for our results on the ``hiding-binding duality.'' However, allowing the adversary to specify $(\RegC,\RegD)$ yields a more robust binding definition, since it would also capture schemes where verification does not completely collapse the last $\lambda$ qubits.} Next, it samples a random bit $b \gets \{0,1\}$ and does one of the following:
    \begin{itemize}
        \item If $b = 0$, it applies $\Com$ and sends $\RegD$ back to the adversary.
        \item If $b = 1$, it initializes an ancilla register to $\ket{0}$, and applies $\SWAP$ to replace the committed message with $\ket{0}$.\footnote{Recall that $\SWAP$ maps $\ket{\psi}\ket{\phi}$ to $\ket{\phi}\ket{\psi}$ for all pure states $\ket{\psi}$, $\ket{\phi}$.} Then it applies $\Com$ and sends $\RegD$ back to the adversary.
    \end{itemize}
    \item Finally, the adversary wins if it can guess $b$.
\end{enumerate}

\begin{definition}[Swap binding; see~\cref{def:swap-binding}]
A non-interactive QSC is computationally (resp. statistically) binding if no efficient (resp. inefficient) adversary can guess $b$ with advantage greater than $1/2 + \negl(\lambda)$.\footnote{Our swap-binding definition is stated for \emph{non-interactive} QSCs. Despite the fact that any interactive QSC can be made non-interactive, we believe it could still be useful to have a definition that directly handles interactive QSCs. In~\cref{subsec:pauli-binding}, we give an alternative definition of binding for QSCs that we call ``Pauli binding'' and show that (1) Pauli binding is equivalent to swap binding for non-interactive QSCs and (2) Pauli binding captures both interactive and non-interactive QSCs. However, it is unclear how to use interactive QSCs that satisfy Pauli binding to build secure interactive protocols. We leave this as an open question for future work.}
\end{definition}

Swap binding satisfies a number of desirable properties:
\begin{itemize}
    \item It handles both computational and statistical binding.
    \item It works for message spaces of any dimension.
    \item It composes in parallel and in sequence across multiple commitments (see \cref{subsec:parallel,subsec:domain-extension,subsec:quantum-merkle}).
    \item It encompasses binding for classical messages. In fact, we show that our definition can be viewed as a natural generalization of Unruh's collapse-binding definition~\cite{EC:Unruh16}, which was originally defined for post-quantum-secure \emph{classical commitments} (see the discussion in \cref{subsubsec:consequences}).
\end{itemize}

A particularly interesting property of swap binding is that it guarantees that QSCs \emph{respect entanglement across multiple commitments}. For instance, if a sender prepares an EPR pair $(\ket{00} + \ket{11})/\sqrt{2}$ and commits to the two qubits with two independent QSCs, our definition ensures that if the sender opens both commitments, it must be to the same entangled state $(\ket{00} + \ket{11})/\sqrt{2}$. This ``entanglement-respecting'' property will be crucial for all of our applications, which involve committing to a large global state using many commitments.

\paragraph{Hiding-binding duality.} A surprising property of swap binding is that it is \emph{dual} to existing notions of \emph{hiding} for quantum messages (hiding was previously defined for quantum messages in prior works on quantum encryption, e.g.,~\cite[Definition 3.3]{C:BroJef15}) in the following sense: hiding requires that $\RegC$ alone reveals no information about the committed message $\ket{\psi}$, while binding requires that $\RegD$ alone reveals no information about $\ket{\psi}$. 

This has several interesting consequences that have no analogue in classical cryptography. For instance, given a computationally binding, statistically hiding $\QSC$ with commitment $\RegC$ and decommitment $\RegD$, there is a simple ``dual'' scheme satisfying statistical binding and computational hiding: just send $\RegD$ as the commitment and use $\RegC$ as the opening! We give further details in~\cref{subsec:duality}. We also point out a duality between hiding and binding for quantum commitments to \emph{classical messages} (see~\cref{subsubsec:qbc-duality}).

\subsubsection{Constructions}

Our new definition directly enables new constructions of QSCs. We build the first \emph{succinct} QSC, which allows a sender to commit to a $\poly(\lambda)$-qubit quantum state with a $\lambda$-qubit commitment $\RegC$. These succinct QSCs behave like a quantum analogue of collision-resistant hash functions (CRHFs) for quantum states; for instance, we show that classical techniques for composing CRHFs such as Merkle-Damgård and Merkle tree commitments naturally generalize to succinct QSCs. 

However, unlike CRHFs, our succinct QSCs are un-keyed and thus are \emph{completely non-interactive}.\footnote{While un-keyed CRHFs are impossible classically due to non-uniform attacks, such attacks can be ruled out in the quantum setting by monogamy of entanglement!} Moreover, the assumptions required for succinct QSCs appear to be much weaker than those required for CRHFs. 

\begin{theorem}[Succinct QSCs from encryption; see~\cref{sec:succinct-constructions}]
Assume the existence of one-time secure private-key encryption for $(\lambda+1)$-qubit quantum messages with $\lambda$-bit classical keys. Then for any $k = \poly(\lambda)$, there exists a non-interactive succinct $\QSC$ for committing to $k$-qubit quantum messages with commitments of only $\lambda$ qubits.

Such one-time secure encryption schemes are implied by one-way functions, and even potentially weaker assumptions such as (one-time secure) pseudorandom unitaries.
\end{theorem}

We also use our definition to formalize security of the ``folklore'' scheme described above. As a consequence, hiding and binding QSCs exist as long as hiding and binding quantum \emph{bit} commitments (QBCs) --- i.e., quantum commitments to classical messages --- exist. Since QSCs encompass commitments to classical messages, this implies that the existence of QBCs and QSCs are equivalent assumptions.

\begin{theorem}[\cref{sec:folklore-construction}]
Quantum state commitments exist (satisfying hiding and binding for quantum messages) if and only if quantum bit commitments exist (satisfying hiding and binding for classical messages).

Quantum bit commitments are implied by one-way functions, and even potentially weaker assumptions~\cite{AC:Yan22,Kretschmer21,C:AnaQiaYue22,C:MorYam22,EPRINT:BraCanQia22}.
\end{theorem}

\subsubsection{Applications}

Our flagship application is a three-message quantum succinct argument based on succinct QSCs. The protocol is a direct quantum analogue of Kilian's PCP-based succinct arguments~\cite{STOC:Kilian92}. Recall that Kilian's protocol assumes the existence of a succinct classical commitment (i.e., a collision-resistant/collapsing hash function) and can be instantiated for any language that has classical PCPs with constant error and polylogarithmic locality, i.e., any $\NP$ language~\cite{STOC:BFLS91,FOCS:FGLSS91,AroraS98,AroraLMSS98}. Correspondingly, our quantum Kilian protocol assumes the existence of a succinct QSC and can be instantiated for any language that has quantum PCPs\footnote{A quantum PCP is a quantum proof that can be probabilistically checked by measuring a few qubits.} with constant error and polylogarithmic locality.\footnote{Our protocol is inspired by a proposal of Chen and Movassagh~\cite{ARXIV:CheMov21}, who showed that Kilian's protocol has a \emph{syntactic} quantum analogue and conjecture (but do not prove) its security in the ``Haar Random Oracle Model.''} The set of such languages includes $\NP$, and is famously conjectured to include all of $\QMA$~\cite{AALV, AAV13}.

% --- biggest open questions in quantum complexity theory.

% and can be instantiated for any language that has quantum PCPs (i.e., a quantum proof string that can be probabilistically checked by measuring a few qubits) with constant error and polylogarithmic locality.\footnote{Our protocol is inspired by a proposal of Chen and Movassagh~\cite{ARXIV:CheMov21}, who showed that Kilian's protocol has a \emph{syntactic} quantum analogue and conjecture (but do not prove) its security in the ``Haar Random Oracle Model.''}

\begin{theorem}[Quantum succinct arguments; see~\cref{theorem:kilian-security}]
\label{theorem:intro-kilian}
        Assuming the existence of succinct QSCs, there is a three-message quantum-communication succinct argument for any language that has quantum PCPs with constant error and polylogarithmic locality. 
\end{theorem}

% We highlight two important corollaries of~\cref{theorem:intro-kilian} for $\NP$ and $\QMA$.

\begin{corollary}
\label{corollary:succinct-np}
Assuming the existence of succinct QSCs, there is a three-message quantum succinct argument for $\NP$.\footnote{This follows by combining~\cref{theorem:intro-kilian} with the classical PCP theorem~\cite{STOC:BFLS91,FOCS:FGLSS91,AroraS98,AroraLMSS98}, since quantum PCPs encompass classical PCPs.}
\end{corollary}

\begin{corollary}
If the quantum PCP conjecture holds~\cite{AALV, AAV13}, then assuming the existence of succinct QSCs, there is a three-message quantum succinct argument for $\QMA$.
\end{corollary}

Even in the context of succinct arguments for $\NP$, \cref{corollary:succinct-np} lowers the round complexity\footnote{\cite{STOC:BitKalPan18} constructs three message succinct arguments for $\NP$ under a non-standard assumption called ``keyless multi-collision-resistant hash functions'' and the hardness of learning with errors.} and relies on weaker assumptions than Kilian's result. 

\paragraph{Swap-based rewinding for quantum protocols.} Proving \cref{theorem:kilian-security} is the most technical component of this work. At a high level, proving soundness requires extracting \emph{quantum information} from the malicious prover. This is qualitatively different from prior works on extracting from quantum adversaries~\cite{EC:Unruh12,EC:Unruh16,FOCS:CMSZ21} which only needed to extract classical information and relied heavily on the ability to record (i.e., copy) the adversary's responses.

Our rewinding procedure works by swapping out the messages underlying the adversary's responses. This explains how QSCs make cryptographic applications possible: \emph{swap binding enables undetectable extraction of quantum information in a rewinding-based security analysis.}

While the swapping idea is simple, it introduces a fundamental problem that was not present in the classical setting. In our setting, by swapping out the adversary's message, we have implicitly forced the adversary to \emph{forget} part of its committed PCP. As a result, \cite{FOCS:CMSZ21}-style state repair --- the only known technique for rewinding succinct protocols --- is now \emph{information-theoretically impossible}. We explain our new rewinding approach in~\cref{subsec:to-applications}.

\paragraph{Quantum sigma protocols.} Finally, as an additional application of our techniques, we prove in~\cref{sec:quantum-sigma} that any~\cite{FOCS:GolMicWig86,FOCS:BroGri20}-style quantum sigma protocol is sound when instantiated with our hiding and binding QSCs. We formalize this using a quantum version of the zero-knowledge PCP framework of~\cite{STOC:IKOS07}. This also provides a proof of computational soundness for a quantum sigma protocol due to~\cite{FOCS:BroGri20}, which was missing a security analysis (see~\cref{subsec:related} for additional details).

\newpage

\section{Technical overview}
\label{sec:tech-overview}

\subsection{Definitions}
\label{subsec:to-definitions}

\paragraph{Syntax.} Recall our syntax for non-interactive quantum state commitments from~\cref{subsubsec:intro-defs}:
\begin{itemize}
    \item The sender commits to a state $\ket{\psi}$ by applying a public unitary $\Com$ to $\ket{\psi} \ket*{0^\lambda}$, obtaining a state on registers $(\RegC, \RegD)$. The commitment is the register $\RegC$.
    \item The sender can later open the commitment by sending the register $\RegD$. The receiver verifies by applying $\Com^\dagger$ and checking that the last $\lambda$ qubits are $0$. If verification succeeds, the receiver recovers the committed message $\ket{\psi}$ from the unmeasured registers. 
\end{itemize}

We emphasize one important difference from the classical setting: classically, the decommitment $d$ can always be assumed to contain $m$ in the clear, i.e., $d= (m,r)$ where $m$ is the committed message and $r$ is the randomness/opening information. For QSCs, however, the $\RegD$ register crucially \emph{does not} contain the committed message in the clear.

\subsubsection{What does binding mean?}

Our first goal in this overview is to answer the following question:
\begin{center}
\emph{What does binding mean for a commitment to quantum messages?}
\end{center}

To understand the subtleties that arise, we first recall the classical definition of binding.

\begin{definition}[Classical binding, informal]
A commitment scheme is binding if an adversary cannot generate a commitment $c$ along with two valid decommitments $d_1 = (m_1,r_1)$ and $d_2 = (m_2,r_2)$ for two different messages $m_1 \neq m_2$.
\end{definition}

What happens if we try to use this definition for commitments to quantum states? Two closely related issues arise:
\begin{itemize}
    \item First, this definition requires the challenger to verify that the two messages are actually different. If the messages are arbitrary quantum states, how should the challenger implement this check? A natural idea is to run a swap test\footnote{The swap test on two quantum states $\ket{\psi}, \ket{\phi}$ outputs $1$ with probability $\frac{1}{2} + \frac{1}{2}\abs{\braket{\psi}{\phi}}^2$.}, but performing the swap test on message registers $\RegM_1,\RegM_2$ requires the challenger to obtain two openings \emph{simultaneously}. This brings us to our second issue.
    \item A quantum adversary does not have to produce two openings simultaneously. Consider a quantum adversary that prepares a valid commitment-decommitment pair $(\RegC,\RegD)$ corresponding to a message $\ket{\psi}$, and suppose it has the ability to modify the state on $\RegD$ so that the commitment opens to a completely different message $\ket{\psi^\perp}$. This adversary clearly violates any reasonable notion of binding, but since it produces the two openings sequentially rather than in parallel, it is not captured by the classical-style definition.\footnote{This issue was first observed by~\cite{FOCS:AmbRosUnr14,EC:Unruh16} in the setting of classical commitments secure against quantum attacks.}
    
    % And yet, this adversary completely breaks a iadversary  produce openings to $\ket{\psi}$ and $\ket{\psi^\perp}$ \emph{simultaneously}, and yet it completely any reasonable notion of binding.
    
    % and suppose  opens to a quantum message $\ket{\psi}$,  to a following quant to be able to  be able to change the committed message without being able to produ adversary might be able to change the committed message, without being able to produce \emph{two different openings simultaneously}. \footnote{This issue was first observed by~\cite{FOCS:AmbRosUnr14,EC:Unruh16} in the setting of classical commitments secure against quantum attacks.} Instead, a quantum adversary 
\end{itemize}

At first glance, it may seem impossible to handle both issues at once. How can a challenger possibly verify that two arbitrary quantum states are different if it only gets to see one state at a time?

We propose the following solution. In the security game, the adversary will open the commitment twice \emph{in sequence}, i.e., it sends registers $(\RegC,\RegD)$ to the challenger (corresponding to the first opening), the challenger later returns register $\RegD$ to the adversary, and finally the adversary sends the challenger a new state on $\RegD$ (corresponding to a second opening). Meanwhile, the challenger picks one of the openings at random and swaps out the committed message $\RegM$ into an internal register $\RegM'$, replacing the committed message with a junk $\ket{0}$ state. At the end of the experiment, the challenger sends $\RegM'$ to the adversary and asks the adversary to distinguish whether it came from its first opening or its second opening. If the adversary can distinguish the two cases, the challenger can be convinced the two messages are different!

% In the security game, the adversary opens the commitment twice \emph{in sequence}, i.e., it sends a state on $(\RegC,\RegD)$ to the challenger, the challenger returns a state on $\RegD$ to the adversary, and finally the adversary sends a second opening on register $\RegD$ to the challenger.

% to the  only needs to produce the second opening after the challenger sends back its first opening.challenger will return the first opening $\RegD$ (but potentially modified) to the adversary before asking it to produce a second opening. The challenger will pick one of the openings at random, swap out the underlying message into an internal register $\RegM'$, send $\RegM'$ to the adversary at the end of the experiment, and ask the adversary to distinguish whether $\RegM'$ was its first message or its second. If the adversary can distinguish the two cases, the challenger can be convinced two messages are different!

\begin{remark} This idea is reminiscent of the~\cite{C:GolMicWig86} zero-knowledge protocol for graph non-isomorphism (in which the prover demonstrates that two graphs are not isomorphic by showing that it can distinguish them), except that we are using it to define a security property rather than to build a protocol.
\end{remark}

In slightly more detail, the security game is the following:

\begin{itemize}
    \item[] $\texttt{DoubleOpenExpt}$:
    \item[] \begin{enumerate}
    \item The adversary sends  $(\RegC, \RegD)$.
    \item The challenger checks that the decommitment is valid by applying $\Com^\dagger$ and measuring the last $\lambda$ qubits (and aborts if it's not $0^\lambda$). The challenger samples a bit $b \gets \{0,1\}$ and does one the following:
    \begin{itemize}
        \item (If $b = 0$) It applies $\Com$ to recompute the commitment/decommitment.
        \item (If $b = 1$) It applies $\SWAP_{\RegM,\RegM'}$ where $\RegM$ is the opened message, and $\RegM'$ is initialized to $\ket{0}$. It then applies $\Com$ to recompute the commitment/decommitment.
    \end{itemize}
    Finally, it sends $\RegD$ back to the adversary (but not $\RegC$).
    \item The adversary sends another decommitment on $\RegD$ back to the challenger.
    \item The challenger checks that the decommitment is valid by applying $\Com^\dagger$ and measuring the last $\lambda$ qubits (and aborts if it's not $0^\lambda$). Then:
    \begin{itemize}
        \item (If $b = 0$) It applies $\SWAP_{\RegM,\RegM'}$ where $\RegM$ is the opened message, and $\RegM'$ is initialized to $\ket{0}$. It then applies $\Com$ to recompute the commitment/decommitment.
        \item (If $b = 1$) It applies $\Com$ to recompute the commitment/decommitment.
    \end{itemize}
    Finally, it sends all of its registers (including $\RegM'$) back to the adversary.
    \item The adversary outputs a guess $b'$ and wins if $b' = b$.
\end{enumerate}
\end{itemize}

\begin{definition}[Binding for QSCs, Attempt 1]
\label{def:binding-attempt-1} A quantum state commitment is binding if it is hard to win $\emph{\texttt{DoubleOpenExpt}}$ with advantage better than $1/2 + \negl(\lambda)$.
\end{definition}

Intuitively, this definition captures the \emph{semantic meaning} of binding for quantum states. If an adversary can change its committed quantum state across two openings, it should also be able to distinguish between the two opened messages. But this definition requires several rounds of back-and-forth interaction --- can we make it any simpler?

\paragraph{A simpler definition: swap binding.}

It turns out that the second opening in $\texttt{DoubleOpenExpt}$ is unnecessary, and a simpler ``single opening'' game suffices:

\vbox{
\begin{itemize}
    \item[] $\texttt{BindExpt}$:
    \item[] \begin{enumerate}
    \item The adversary sends the registers $(\RegC, \RegD)$.
    \item The challenger checks that the decommitment is valid by applying $\Com^\dagger$ and measuring the last $\lambda$ qubits (and aborts if it's not $0^\lambda$). The challenger samples a bit $b \gets \{0,1\}$ and does one of the following:
    \begin{itemize}
        \item (If $b = 0$) It applies $\Com$ to recompute the commitment/decommitment.
        \item (If $b = 1$) It applies $\SWAP_{\RegM,\RegM'}$ where $\RegM$ is the opened message, and $\RegM'$ is initialized to $\ket{0}$. It then applies $\Com$ to recompute the commitment/decommitment.
    \end{itemize}
    Finally, it sends the $\RegD$ register back to the adversary.
    \item The adversary outputs a guess $b'$ and wins if $b' = b$.
\end{enumerate}
\end{itemize}}

This gives our final definition, which we call \emph{swap binding}.

\begin{definition}[Swap binding]
\label{def:binding-attempt-2} A quantum state commitment is swap binding if it is hard to win $\emph{\texttt{BindExpt}}$ with advantage better than $1/2 + \negl(\lambda)$.
\end{definition}

Why is this definition \emph{equivalent} to~\cref{def:binding-attempt-1}? At a very high level, any adversary for the double-opening game can be ``folded'' into an adversary for the single opening game by running the first and second opening rounds in superposition. In slightly more detail, let $\ket{\psi}_{\RegC,\RegD,\RegE}$ (where $\RegE$ is the internal register of the double opening adversary) be the initial state of the adversary for the double opening game, and let $U$ denote the unitary the adversary applies before sending the second opening. Consider the following adversary for the single opening game:
\begin{itemize}
    \item The adversary prepares a control qubit initialized to $\ket{+}$, and applies $U$ controlled on $\ket{+}$ to obtain the state $(\ket{0}\ket{\psi}_{\RegC,\RegD,\RegE} + \ket{1} U \ket{\psi}_{\RegC,\RegD,\RegE})/\sqrt{2}$. It then sends the registers $(\RegC, \RegD)$ to the challenger.
    \item Upon receiving the $\RegD$ register back from the challenger, it applies controlled-$U^\dagger$, measures the control qubit in the Hadamard basis, and uses the resulting outcome to guess $b$.
\end{itemize}
It is possible to show that if the original adversary can win the double opening game, this adversary will win the single opening game.

\subsubsection{Consequences}
\label{subsubsec:consequences}

We briefly discuss some consequences of our new definition:
\begin{itemize}
\item \textbf{Relationship to classical binding.} Since quantum information encompasses classical information, we should expect our definition to capture binding for classical messages. Indeed, we can show this is the case:

\begin{itemize}
    \item By a hybrid argument, we can replace the $\SWAP_{\RegM,\RegM'}$ operation in $\texttt{BindExpt}$ with any other quantum operation $P$ that acts on $\RegM$ and an ancilla held by the challenger.\footnote{We sketch the argument here. In Hybrid 0, the challenger simply sends $\RegD$ back to the adversary after verifying the commitment. In Hybrid 1, the challenger performs $\SWAP_{\RegM,\RegM'}$ before sending $\RegD$ back; this is indistinguishable from Hybrid 0 by our binding definition. In Hybrid 2, the challenger applies $P$ to $\RegM$ and then applies $\SWAP_{\RegM,\RegM'}$; this is perfectly indistinguishable from Hybrid 1 since the register $P$ acts on is independent of the adversary's view. Finally, in Hybrid 3, the challenger only applies $P$; this is indistinguishable from Hybrid 2 by our definition, which allows us to undetectably remove $\SWAP_{\RegM,\RegM'}$.}
    \item Define $P$ to be the operation that applies $Z^s$ to $\RegM$ for a random $s$; this is equivalent to \emph{measuring} $\RegM$ in the standard basis and discarding the outcome. Our definition implies that the adversary cannot distinguish its original decommitment from a decommitment in which the message has been measured in the standard basis. This is essentially Unruh's definition~\cite{EC:Unruh16} of collapse-binding for classical messages!\footnote{Technically, Unruh's definition is for classical commitments to classical messages. However, it is not hard to adapt his definition to the syntax of quantum-communication commitments.}
\end{itemize}

In fact, we can say more. In $\texttt{BindExpt}$, we could initialize $\RegM'$ to the maximally mixed state (instead of $\ket{0}$), and the definition would remain the same. From the adversary's point of view, $\SWAP_{\RegM,\RegM'}$ is equivalent to the challenger applying $X^rZ^s$ to $\RegM$ for uniformly random strings $r,s$, since this also maximally mixes the message. This gives an alternative view of our definition: a $\QSC$ is binding for quantum states if it is (collapse-)binding for classical messages in both the standard and Hadamard bases!

\item \textbf{Hiding-binding duality.} We have not yet defined hiding for QSCs, but it is not hard to write down a definition: the adversarial receiver sends a register $\RegM$ containing its message $\ket{\psi}$, the challenger either sends a commitment to $\ket{\psi}$ or a commitment to an unrelated message, and the adversary must guess which one it received.\footnote{This hiding definition is virtually identical to existing definitions of semantic security for quantum encryption schemes~\cite[Definition 3.3]{C:BroJef15}.} Perhaps the most surprising consequence of our definition is that for \emph{non-interactive} QSCs, this hiding experiment is the same as our binding experiment --- except that the challenger sends $\RegC$ instead of $\RegD$ to the adversary. Indeed, sending registers $(\RegC , \RegD)$ containing a valid commitment/decommitment state is equivalent to sending a register $\RegM$ containing the committed message state in the clear, since a valid commitment/decommitment is $\Com \ket{\psi}_\RegM \ket*{0^\lambda}$.

One consequence of this duality is that if we have, for example, a statistically binding and computationally hiding QSC, there is a simple ``dual'' QSC satisfying computationally binding and statistical hiding: just send $\RegD$ as the commitment and use $\RegC$ as the opening!
\end{itemize}

\subsection{Constructions}
\label{subsec:to-constructions}

\subsubsection{Warm-up: hiding and binding QSCs}

Having defined the security guarantees of QSCs, we will now prove that the ``folklore'' construction is secure. Recall that the sender commits to $\ket{\psi}$ by sampling random $r,s$ and sending $X^rZ^s\ket{\psi}$ together with a commitment $\ket{c_{r,s}}$ to the classical string $(r,s)$. Writing this as a non-interactive QSC (which requires replacing the randomness of $r,s$ with a corresponding superposition), the commitment/decommitment state is (up to normalization)
\[ \sum_{r,s} (\ket{r,s}\ket{d_{r,s}})_{\RegD} (X^rZ^s\ket{\psi} \ket{c_{r,s}})_\RegC, \]
where $\ket{d_{r,s}}$ is the decommitment for $\ket{c_{r,s}}$.
\begin{itemize}
    \item For hiding, we must show that the $\RegC$ register alone does not reveal $\ket{\psi}$. This can be easily formalized with a hybrid argument. First, we can replace $\ket{c_{r,s}}$ with $\ket{c_{0,0}}$ by the hiding of the commitment to $(r,s)$. Then, because the superposition over $r,s$ is held on the challenger's registers, the adversary's view in this hybrid is just a maximally mixed state and a commitment to $(0,0)$. 
    \item For binding, we must show that the $\RegD$ register alone does not reveal $\ket{\psi}$. By collapse-binding security of the $\ket{c_{r,s}}$ commitment, the adversary cannot distinguish a superposition of valid decommitments from one in which the message $(r,s)$ is measured. In particular, this means the adversary's decommitment $\RegD$ is indistinguishable from $\sum_{r,s} \ketbra{r,s} \otimes \ketbra{d_{r,s}}$. But this mixed state is just a uniform mixture of valid decommitments to uniformly random $r,s$ and is therefore independent of $\ket{\psi}$.
\end{itemize}

We therefore conclude the following.

\begin{theorem}
Hiding and binding QSCs exist if and only if hiding and binding QBCs exist (i.e., quantum-communication commitments to classical bits). 
\end{theorem}

A pair of recent works~\cite{C:MorYam22,C:AnaQiaYue22} showed that hiding and binding QBCs --- which were previously known assuming one-way functions --- can be built from quantum assumptions that are potentially weaker than one-way functions~\cite{Kretschmer21}. Consequently, hiding and binding QSCs exist under the same weak assumptions.

\subsubsection{Beyond the folklore construction: succinct QSCs.}

While we have successfully formalized the security of the folklore construction using swap binding, a definition really begins to shine when it enables new constructions. As mentioned in~\cref{sec:intro}, the folklore construction requires the commitment $\RegC$ to be \emph{at least as large} as the committed message. We will now see how to construct a \emph{succinct} QSC, meaning a binding QSC for $n$-qubit messages where the commitment $\RegC$ is smaller than $n$ qubits.\footnote{Notice that we have dropped the hiding property; as in the classical setting, many applications of succinct commitments do not explicitly require hiding. We remark that in the classical setting, succinct commitments are known to imply standard hiding-binding commitments (and thus one-way functions).}

\paragraph{A succinct QSC from any collapsing hash function.} A natural starting point is to assume the existence of collision-resistant hash functions, or rather, their post-quantum analogue: collapsing hash functions~\cite{EC:Unruh16}. Collapsing hash functions are succinct, collapse-binding commitments to \emph{classical messages}. That is, given a hash function description $h \gets H_\lambda$, a sender can succinctly commit to a message $m$ by sending the classical string $h(m) = y$, and can open by revealing $m$. 

How can we turn a succinct commitment $h$ for classical messages into one for quantum messages? Our idea is simple: use $h$ to commit to the quantum state in both the standard and Hadamard bases. Concretely, given a collapsing hash function $h: \{0,1\}^n \rightarrow \{0,1\}^{n/4}$, the sender commits to an $n$-qubit message $\RegM$ as follows:

\begin{enumerate}
    \item Interpret $\RegM$ in the standard basis and evaluate $h$ coherently, writing the output onto a fresh register $\RegC_0$.
    \item Apply the $n$-bit Hadamard transform $H^{\otimes n}$ to $\RegM$.
    \item Apply $h$ again, this time writing the output onto another fresh register $\RegC_1$. Finally, send the $(n/2)$-qubit state $\RegC = (\RegC_0, \RegC_1)$ as the commitment. 
    \item To decommit, the sender provides the $\RegD = \RegM$ register. This allows the receiver to uncompute the commitment, check validity, and recover the original message.
\end{enumerate}

Our swap-binding definition makes it easy to prove this scheme secure assuming that $h$ is collapse-binding. Recall that collapse-binding of $h$ guarantees that if an adversary sends a superposition of inputs for $h$ to a challenger, then it cannot distinguish between the following:
\begin{itemize}
    \item the challenger evaluates $h$ in superposition, measures the output, and then returns the input registers, or
    \item the challenger simply measures the input registers and returns them.
\end{itemize}
In our scheme, since the challenger keeps the register $\RegC = (\RegC_0, \RegC_1)$, from the adversary's point of view, both $\RegC_0$ and $\RegC_1$ are measured (in the standard basis). 
By carefully invoking collapse-binding security twice, once for each coherent evaluation of $h$ in our scheme, we can show that the decommitment held by the adversary is indistinguishable from one that is generated as follows: measure $\RegM$ in the standard basis, apply $H^{\otimes n}$, and measure $\RegM$ again. But this gives a totally mixed state, and so the decommitment completely hides the original quantum message.

\paragraph{Succinct QSCs from even weaker assumptions.} The above construction demonstrates that succinct QSCs can be built from collapsing hash functions (succinct classical commitments). Perhaps surprisingly, we show something stronger: succinct QSCs can even exist in a world where succinct classical commitments do not!

For instance, suppose we are only given a (post-quantum) one-way function, or equivalently, a (post-quantum) pseudorandom generator (PRG) $G: \{0,1\}^{n/2} \rightarrow \{0,1\}^{2n}$~\cite{SICOMP:HILL99,FOCS:Zhandry12}. We show that it is still possible to commit to an $n$-qubit state $\ket{\psi}$, as follows:

\begin{enumerate}
    \item Initialize an $n/2$-qubit register to $\ket{0}$ and apply $H^{\otimes \frac{n}{2}}$ to obtain $2^{-n/4}\sum_{k \in \{0,1\}^{n/2}} \ket{k}$.
    \item Controlled on $k$, apply $X^{G_0(k)}Z^{G_1(k)}$ to $\ket{\psi}$, where $G_0(k),G_1(k)$ are the first and last $n/2$ bits of $G(k)$, respectively. The resulting state is
    \begin{align*}
        \frac{1}{2^{n/4}}\sum_{k \in \{0,1\}^{n/2}} \ket{k}_{\RegC} X^{G_0(k)}Z^{G_1(k)} \ket{\psi}_{\RegD},
    \end{align*}
    where the registers containing $k$ are the succinct commitment $\RegC$, and the one-time padded message is the decommitment $\RegD$.
\end{enumerate}

Our swap-binding definition makes it easy to prove this scheme secure:
\begin{itemize}
    \item When the challenger holds the $\RegC$ register (i.e., it is traced out in the adversary's view), the decommitment for $\ket{\psi}$ is the mixed state resulting from applying a $X^{G_0(k)}Z^{G_1(k)}$ to $\ket{\psi}$ for a uniformly random $k$. 
    \item By PRG security, this is indistinguishable from $X^rZ^s\ket{\psi}$ for uniformly random $r,s$, which is the maximally mixed state so swap-binding follows.
\end{itemize}

More generally, all we need is a one-time secure encryption scheme for $n$-qubit quantum messages with short keys. In fact, any scheme with keys shorter than $n$ bits will suffice, since it turns out that using a quantum analogue of Merkle-Damgård domain extension (see~\cref{subsec:domain-extension}), succinct QSCs with 1 qubit of compression can be composed to achieve any desired compression!

As we just saw, one-way functions give one possible instantiation of such an encryption scheme: to encrypt a state $\ket{\psi}$ with key $k$, just apply $X^{G_0(k)}Z^{G_1(k)}$. However, there are other instantiations from assumptions that may be even weaker than one-way functions. In particular, we show that one-time secure quantum encryption with short keys is implied by any pseudorandom unitary~\cite{C:JiLiuSon18} (see~\cref{claim:oracle-separation}), and the latter is known to be black-box separated from one-way functions~\cite{Kretschmer21}. In summary, we have the following theorem.

\begin{theorem}[Succinct QSCs from encryption]
Assuming the existence of one-time secure encryption for $n$-bit quantum messages with $(n-1)$-bit classical keys, succinct QSCs exist. Such one-time encryption schemes are implied by one-way functions, but their existence is a potentially weaker assumption.
\end{theorem}

\subsection{Applications}
\label{subsec:to-applications}

So far, we have proposed a definition and argued that it captures the semantic meaning of binding for quantum messages. We have also shown several instantiations of schemes achieving this definition. However, there is still a crucial missing piece to our theory: can commitments satisfying our definition actually be used to achieve cryptographic ends?

In this section, we describe our flagship application: a three-message quantum succinct argument system based on succinct QSCs. Recall that a succinct argument is a protocol that enables a prover to convince a verifier that a claim is true (e.g., an $\NP$ statement), where the total communication and verifier runtime is extremely short. Our protocol is a direct quantum analogue of Kilian's classical succinct argument protocol~\cite{STOC:Kilian92} and is based on a recent proposal of Chen and Movassagh~\cite{ARXIV:CheMov21} (see~\cref{subsec:related} for more details). While the protocol is simple to state, proving security will be far more difficult. Even for Kilian's classical protocol, security against quantum attacks was a long-standing open question that was only recently settled by~\cite{FOCS:CMSZ21}.  Moreover, the~\cite{FOCS:CMSZ21} technique fundamentally relies on the ability to copy information and is therefore limited to extracting classical information from classical protocols.

In this work, we resolve several challenges unique to the quantum setting and obtain a general framework for proving security of quantum protocols via \emph{rewinding}. This is our most technically involved contribution. For this overview, we will present these techniques in the context of quantum succinct arguments (see~\cref{sec:quantum-sigma} for applications to quantum sigma protocols).

\subsubsection{The quantum succinct argument protocol}
\label{subsubsec:to-kilian-protocol}

\paragraph{Kilian's protocol.} Kilian's protocol~\cite{STOC:Kilian92} is a compiler that uses succinct commitments to turn a probabilistically checkable proof (PCP) into a four-message interactive succinct argument. At the start of the protocol, the verifier samples a classical key for the succinct commitment, i.e., a hash function description $h$, and sends it to the prover. The prover uses $h$ to generate a \emph{tree commitment} $c$ to the PCP and sends $c$ to the verifier.\footnote{Assuming $h: \{0,1\}^n \rightarrow \{0,1\}^{n/2}$, the tree commitment is generated as follows: (1) partition the PCP into blocks of size $n/2$, (2) create a binary tree of hash values where the PCP blocks form the leaves and each internal node is the hash of its two children, and (3) send the root as the commitment. The local opening for a PCP index $i$ is a set consisting of (i) the PCP block containing $i$ and (ii) the values of every node adjacent to the path from block containing $i$ up to the root. The opening for a set of PCP indices $S$ is the union of the sets corresponding to the decommitments for each $i \in S$.} Next, the verifier samples random coins corresponding to the randomness of the PCP verifier. Finally, the prover sends \emph{local openings} for the indices $S$ specified by the random coins, which are at most $|S|\cdot \poly(\lambda)$ bits long in total. The verifier accepts if the PCP verifier accepts and the openings are valid.

\paragraph{The quantum protocol.} Following Kilian's template, our quantum protocol is a compiler that uses succinct QSCs to turn any \emph{classical or quantum} PCP into a quantum succinct argument. As mentioned earlier, the resulting protocol is very similar to a heuristic proposal of Chen and Movassagh~\cite{ARXIV:CheMov21} (see~\cref{subsec:related} for more details). 
\begin{itemize}
    \item \textbf{First message.} The prover generates a \emph{quantum tree commitment} to its PCP using a succinct QSC that maps $n$-qubit messages to $n/2$-qubit commitments. The quantum tree commitment is computed as follows:
    \begin{enumerate}
        \item The prover partitions the PCP into blocks of size $n/2$.
        \item Starting at the leaves, it creates a binary tree of commitments up to the root where each node register is a succinct commitment to its two children, and
        \item It sends the root commitment as the commitment to the whole PCP.
    \end{enumerate}
    We emphasize one difference from the classical setting: since the prover is committing to quantum states, the committed PCP is unavailable to the prover once it sends the root.
    \item \textbf{Second message.} The verifier sends random coins corresponding to the PCP verifier randomness.
    \item \textbf{Third message.} The prover sends local openings for the PCP indices $S$ specified by the verifier's challenge. For an individual leaf, the corresponding opening consists of every register on or adjacent to the root-to-leaf path (excluding the root itself). The opening for a set $S$ consists of all registers involved in the opening for each $i \in S$. We refer the reader to~\cref{subsec:quantum-merkle} and~\cref{fig:merkle-tree} for further details.
    \item \textbf{Verification.} Starting from the root, the verifier works its way down the tree, using the provided openings to reveal each subsequent layer of commitments. Eventually, this reveals the qubits of the PCP at the locations in $S$. The verifier accepts if every individual QSC was opened correctly and the PCP verifier would accept. 
\end{itemize}

This template is compatible with both classical and quantum PCPs. We emphasize that even when the PCP is classical, our protocol still requires the ability to commit to arbitrary \emph{quantum states}. In particular, a tree commitment requires composing quantum commitments in sequence, i.e., committing to commitments, so even if the leaves are classical, the internal nodes are still commitments to quantum messages.

The goal for the remainder of this subsection will be to prove that this protocol is sound for any classical or quantum PCP and thereby establish the following theorem.

\begin{theorem}
Assuming the existence of succinct QSCs, every $\NP$ language has a three-message quantum succinct argument. If we additionally assume the quantum PCP conjecture, this extends to $\QMA$.

\end{theorem}

\subsubsection{Rewinding for quantum protocols}
\label{subsubsec:to-rewinding-difficulties}

At a very high level, we will prove soundness of our protocol --- or rather, a stronger property called argument of knowledge --- by showing how to extract a convincing PCP from any malicious prover $\widetilde{P}$ that makes the verifier accept with noticeable probability. In the classical setting, this is achieved by a technique called \emph{rewinding}: the extractor queries $\widetilde{P}$ many times on random challenges, rewinding $\widetilde{P}$ back to its earlier state after every query. After sufficiently many queries, the extractor will have enough accepting responses to stitch together a PCP.

It is well known that traditional rewinding proofs completely fail to capture quantum attacks~\cite{vandeGraaf97,Watrous09,FOCS:AmbRosUnr14}: when the extractor runs $\widetilde{P}$ and records its response, it may irreversibly disturb $\widetilde{P}$'s state, rendering it useless for further queries. Fortunately, there has been a line of work showing how to rewind quantum adversaries for some important \emph{classical} protocols~\cite{Watrous09,EC:Unruh12,EC:Unruh16,C:ChiChuYam21,FOCS:CMSZ21,FOCS:LomMaSpo22,TCC:LaiMalSpo22}.

The work that turns out to be most relevant to our setting is~\cite{FOCS:CMSZ21}, which showed a general technique that enables repeatedly querying $\widetilde{P}$ and recording an arbitrary number of its accepting responses. At a high level, their technique consists of two steps:

\begin{itemize}
    \item \textbf{Step 1: Reduce the task to ``one-bit'' extraction.} First, they reduce the problem to a simplified setting where the extractor only measures the bit indicating whether $\widetilde{P}$ succeeds on a query (as opposed to measuring the entire response). This step invokes collapse-binding security~\cite{EC:Unruh12,EC:Unruh16} to argue that measuring this bit is computationally indistinguishable from measuring the entire response. 
    \item \textbf{Step 2: One-bit extraction via state repair.} The primary technical innovation of \cite{FOCS:CMSZ21} is a procedure for one-bit extraction. One of their key insights was to realize that the extractor does not need to restore the malicious prover's original state, but rather \emph{any state that gives it the same success probability as the original state}. Then, they designed a ``state repair'' procedure to restore the success probability of $\widetilde{P}$ after any one-bit disturbance. Their resulting extractor simply queries $\widetilde{P}$ repeatedly on random challenges and runs the state repair procedure in between each query.
\end{itemize}

In summary, the full~\cite{FOCS:CMSZ21} extractor repeats the following ``measure-and-repair'' step: (1) run the adversary on a random challenge, (2) check whether the response is valid; if so, measure and record the response, and finally (3) run the repair procedure to restore the success probability.

\paragraph{Extracting quantum information.} The~\cite{FOCS:CMSZ21} extractor was designed for a different setting than ours: the goal of \cite{FOCS:CMSZ21} is to extract classical information from a quantum prover $\widetilde{P}$ in a classical protocol, whereas our goal is to extract \emph{quantum information} from a quantum prover $\widetilde{P}$ in a quantum protocol.\footnote{Even when we instantiate our quantum succinct argument with a classical PCP, we will still need to extract quantum information to invoke security of the tree commitment.} In our setting it doesn't make sense to \emph{measure} $\widetilde{P}$'s response, but our swap-binding definition strongly suggests a quantum analogue: swap out $\widetilde{P}$'s opened message and replace it with a junk $\ket{0}$ state!

While this idea may seem promising, it introduces a fundamental problem that was not present in the classical setting. By swapping $\widetilde{P}$'s message with junk, we have implicitly forced the adversary to \emph{forget} part of its committed PCP. As a result, repairing the adversary is now \emph{information-theoretically impossible}: if we could repair the adversary, the extractor would eventually generate multiple copies of the quantum PCP.

\paragraph{The swap-augmented prover.} Fortunately, there is still a way to use \cite{FOCS:CMSZ21}-style state repair. To explain our idea, suppose that the extractor has swapped a single index $i$ of the prover's PCP onto some internal register $\RegM'_i$. Our idea is to define the following \emph{swap-augmented prover}:

\begin{itemize}
    \item[] Let $\RegC$ denote the register holding the (root of the) tree commitment. 
    \item The state of the swap-augmented prover consists of $\widetilde{P}$'s state augmented with the registers $(\RegM'_i, \RegC)$.
    \item To run the swap-augmented prover on a set of positions $S$:
    \begin{enumerate}
        \item First, run the real prover $\widetilde{P}$ on $S$ to generate a response on some register $\RegZ$. If $i \not\in S$, we are done.
        \item If $i \in S$:
        \begin{enumerate}
            \item\label[step]{step:to-recover} Use $(\RegC,\RegZ)$ to recover the message (corresponding to indices $S$) in the clear.
            \item Next, apply $\SWAP_{\RegM_i,\RegM'_i}$ where $\RegM_i$ is the register containing the message opened by the prover for index $i$.
            \item Apply the inverse of~\cref{step:to-recover}.
        \end{enumerate}
    \end{enumerate}
\end{itemize}

Roughly speaking, the swap-augmented prover is defined to output the responses that the prover \emph{would have given} if the swap had not occurred. More generally, at every stage in the extraction procedure, we define a swap-augmented prover that makes use of all registers that have been swapped out so far.

At the beginning of the extraction, the swap-augmented prover is simply the original malicious prover $\widetilde{P}$, which convinces the verifier to accept with non-negligible probability $p$. Then after each query, instead of repairing $\widetilde{P}$ itself (which is impossible), our extractor will \emph{run state repair on the swap-augmented prover}. At first glance, this may seem somewhat strange: unlike~\cite{FOCS:CMSZ21}, which only runs state repair on $\widetilde{P}$'s state, we give the state repair procedure the additional freedom to modify the extracted PCP itself. Nevertheless, we show the extractor works by arguing that there is eventually a point where:

\begin{itemize}
    \item with very high probability, the swap-augmented prover answers any PCP query using only the augmented registers $\{\RegM_i'\}$, and
    \item the swap-augmented prover still has success probability $\approx p$.
\end{itemize} 
If this happens, then the registers $\{\RegM'_i\}_i$ must contain a convincing PCP!

\paragraph{The oracle security lemma.} 

It turns out to be quite subtle to argue that swap-binding security can actually be invoked in our extraction procedure. In~\cite{FOCS:CMSZ21}, collapse binding is used to argue that whenever the prover produces a superposition of valid responses, we can undetectably measure it and copy it. In our setting, we would like to argue that a swap operation can be applied undetectably.

However in a quantum commitment scheme, the adversarial committer does not have access to the commitment register after sending the commitment. But our extractor explicitly needs access to the commitment (for instance, the extractor must be able to run swap-augmented prover, which makes use of the commitment to generate its responses). Thus, it is completely unclear how we can argue security while running this extractor.

Our solution is the following:
\begin{itemize}
    \item First, we show that swap binding generically implies a stronger security property we call \emph{oracle swap binding}. Recall that in the swap-binding security experiment, the adversary's message in the $\Com^\dagger$ basis is either left alone on its original message register $\RegM^{(0)}$ or swapped onto a challenger register $\RegM^{(1)}$. What the adversary actually sees is the $\RegD$ register in the original basis.
    
    We define an oracle swap-binding experiment that gives the adversary additional power. In particular, it can access an oracle $\mathcal{O}$ that performs any efficient operation on $\RegM^{(b)}$ in the $\Com^\dagger$ basis, where $b$ is the challenge bit. In other words, the operation is performed on the register that contains the adversary's originally committed message. The adversary's goal in this experiment is the same as before, which is to guess $b$.
    
    It is not a priori obvious that swap binding implies oracle swap binding, but we prove this in~\cref{sec:oracle-security}. 
    
    \item We show that the swap operations performed in our extractor are undetectable to the adversary by reduction to oracle swap binding.
    
    Specifically, we prove that the extractor's access to the $\RegC$ register can be implemented with the oracle provided in this security game. At a high level, this oracle enables the extractor to (a) check whether an opening is valid and (b) check if the opened message passes PCP checks. These functionalities turn out to be sufficient to implement the~\cite{FOCS:CMSZ21} repair procedure given the oracle but not the $\RegC$ register, which is what allows us to appeal to oracle swap-binding.
\end{itemize}

\subsection{Related works}
\label{subsec:related}

\paragraph{Comparison to~\cite{ARXIV:CheMov21}.} The quantum succinct argument protocol we analyze is based on a proposal from Chen and Movassagh~\cite{ARXIV:CheMov21}, who were the first to realize that a quantum-communication analogue of Kilian's protocol was even possible. The syntax for our tree commitments is similar to the syntax proposed in their work (modulo some changes to handle generic succinct QSCs).

However,~\cite{ARXIV:CheMov21} did not prove any security guarantees for their protocol, nor did they define the intermediate primitive of succinct QSCs. Instead they conjectured (but did not prove) that their protocol is secure when the commitments are modeled as \emph{Haar random oracles}. In this work, we give a full proof of security \emph{in the plain model} under very mild cryptographic assumptions. We additionally construct and analyze other protocols that do not rely on the quantum PCP conjecture (e.g., when our succinct argument template is instantiated with classical PCPs, we obtain a protocol that requires fewer messages and weaker assumptions than Kilian's protocol).

\paragraph{Computational soundness of~\cite{FOCS:BroGri20}.} The quantum sigma protocol we analyze is essentially the same as the one in~\cite{FOCS:BroGri20}. However, the proof that the protocol is computationally sound (in fact, an argument of knowledge) is new to this work. While~\cite{FOCS:BroGri20} originally proposed this claim, a recent version of their work~\cite{SIAM:BroGri22} retracts it since it was missing a rewinding-based analysis, which is essential
in the setting of computational soundness. In~\cref{sec:quantum-sigma}, we prove computational soundness using our new rewinding techniques. As an additional contribution, we also place their protocol in a general framework using our new abstraction of hiding and binding QSCs.

\paragraph{Comparison to \cite{FOCS:Mahadev18a}.}
\cite{FOCS:Mahadev18a} constructs a ``weak quantum state commitment'' with classical communication and a classical receiver. As Mahadev explains, her protocol is only a ``weak'' commitment because the commit phase of does not actually bind the prover to a fixed quantum state. Instead, she shows that the commitment \emph{together with any malicious opening attack} fixes a state. While Mahadev's construction suffices for her applications, it is not a commitment in the typical sense and her techniques do not directly yield QSCs.

\paragraph{Prior work on quantum \emph{bit} commitments.} The basic syntax for quantum state commitments (QSCs) is similar to the ``canonical form'' described in prior work of Yan~\cite{AC:Yan22}. However, Yan only considers commitments to \emph{single-bit classical messages}, whereas we consider commitments to general quantum states. Even if we restrict to commitments to classical messages, there is one syntactic difference between our form and~\cite{AC:Yan22}: Yan's ``canonical form'' assumes that the bit commitment scheme is specified by two circuits $Q_0$ and $Q_1$ and the committed message is revealed \emph{in the clear} in the opening phase.

It was also well known from prior work on quantum bit commitments that quantum communication enables removing interaction~\cite{KO09,KO11,YWLQ15,TCC:BitBra21,AC:Yan22}. The non-interactive compiler we describe in~\cref{sec:non-interactive}) follows from very similar ideas and is included for the sake of completeness.

\paragraph{Other succinct arguments for $\QMA$.} \cite{TCC:ChiChuYam20,C:BKLMMVVY22} construct certain variants of succinct arguments for $\QMA$ with classical communication. However, both works rely on extremely heavy cryptographic hammers including post-quantum indistinguishability obfuscation (for which there are only heuristic candidates). Our succinct argument for $\QMA$ is significantly simpler and can be instantiated from far weaker cryptographic assumptions at the cost of relying on the quantum PCP conjecture.

\paragraph{Authentication of quantum messages.} Our hiding-binding duality is reminiscent of a well-known result of~\cite{FOCS:BCGST02}, which showed that any secure authentication scheme for quantum messages must also encrypt the message. However, there are several major differences between our results and~\cite{FOCS:BCGST02}:
\begin{itemize}
    \item While binding and authentication both consider adversaries that attempt to change some underlying message, the two notions capture different threat models. Binding guarantees security against the sender (who prepares the initial messages), whereas authentication guarantees security against an adversary who intercepts communication between the sender and receiver. 
    \item There is no ``duality'' in the setting of authentication. In particular, authentication implies encryption, but \emph{encryption does not imply authentication}. 
    \item \cite{FOCS:BCGST02} only consider statistical security, whereas our definitions apply to both the computational and statistical settings.
\end{itemize}

In fact, our work highlights a potential limitation of the~\cite{FOCS:BCGST02} definition: their definition does not say what happens to entangled messages.\footnote{It is claimed on~\cite[Page~1]{FOCS:BCGST02} that their definition ``implies security for mixed or entangled states,'' but what they mean is that the mixed state of the authenticated message is preserved (see~\cite[Appendix~B]{FOCS:BCGST02}). This does not immediately rule out the possibility we raise here, in which the mixed state is preserved but the entanglement is broken.} For instance, suppose a sender prepares an entangled state on registers $(\RegA, \RegB)$ and authenticates $\RegB$. Is the receiver guaranteed to obtain a state that is still entangled with $\RegA$? Their definition for mixed states~\cite[Appendix~B]{FOCS:BCGST02} guarantees that if the receiver accepts, the message it obtains is negligibly close to the sender's original message (i.e., fidelity $1- \negl(\lambda)$). However, it does not say that this message remains entangled with $\RegA$. While it is almost certainly true that the constructions in~\cite{FOCS:BCGST02} preserve entanglement in this manner, it is not clear how to deduce this from the definition.

In our setting, proving such guarantees is made possible by the fact that our definition composes in parallel. In slightly more detail, if a sender commits to $\RegB$ and leaves $\RegA$ untouched, we can view this as a parallel composition of a trivially binding (identity) commitment to $\RegA$ and a commitment to $\RegB$. The security of this larger commitment ensures that entanglement is preserved.

\newpage

\section{Preliminaries}

\paragraph{Notation.}

The security parameter is written as $\lambda$. A function $f: \mathbb{N} \rightarrow [0,1]$ is \emph{negligible} (denoted $f(\lambda) = \negl(\lambda)$) if it decreases faster than the inverse of any polynomial. A probability is \emph{overwhelming} if is at least $1-\negl(\secp)$ for a negligible function $\negl(\secp)$. For any nonzero $n \in \mathbb{N}$, let $[n] = \{1,2,\dots,n\}$ and let $[0]$ be the empty set. For a set $R$, let $r \leftarrow R$ be a uniformly random sample from $R$.

\paragraph{Quantum information preliminaries.}

Let $\cH$ be a finite-dimensional Hilbert space. Let $\bfS(\cH)$ be the space of all Hermitian operators on $\cH$. Let $\bfD(\cH)$ be set of all $\brho \in \bfS(\cH)$ with $\Tr(\brho) = 1$. A \emph{pure quantum state} is a unit vector $\ket{\psi} \in \cH$. A \emph{density matrix} (and mixed quantum state) is an operator $\brho \in \bfD(\cH)$. For any vector $\ket{\phi} \in \cH$, we define its norm $\norm{\ket{\phi}} \coloneqq \sqrt{\braket{\phi}}$. For an operator $M$ on $\cH$, we define its operator norm as 
\[
    \norm{M}_{op} \coloneqq \max_{\ket{\psi} : \norm{ \ket{\psi} }=1} \norm{M\ket{\psi}} .
\]

We use the term quantum register to refer to a collection of qubits that we wish to treat as a single unit. These will be denoted with uppercase sans serif font, e.g., $\RegA,\RegB,\RegC$. Each register is associated with a finite-dimensional Hilbert space, denoted by writing the same letter in uppercase calligraphic font, e.g., $\cA,\cB,\cC$ are the spaces corresponding to registers $\RegA,\RegB,\RegC$. For an $n$-qubit register $\RegA$, we sometimes use the shorthand $\ket{0}$ to denote the all-zero state $\ket{0^n}$ on $\RegA$.

An \emph{observable} is represented by a Hermitian operator $O$ on $\RegH$. When $O^2=\Id$ (equivalently, $O$ has eigenvalues in $\{-1,1\}$), $O$ is also a \emph{binary observable}. For binary observables $O$ we define the projector onto its $+1$ eigenspace $O^+$ and the projector onto its $-1$ eigenspace $O^-$ so that $O = O^+ - O^-$. For any projector $\Pi$, we define a corresponding binary projective measurement $\{\Pi, \Id-\Pi\}$ and say that the measurement \emph{accepts} if the outcome is $\Pi$; otherwise, we say that the measurement rejects.

General (non-unitary) evolution of a quantum state can be represented via a completely-positive trace-preserving (CPTP) map $\Phi : \bfS(\cH) \rightarrow \bfS(\cH')$. We define the diamond norm~\cite{Kitaev97,KSV02} of any CPTP map $\Phi : \bfS(\cH) \rightarrow \bfS(\cH')$ to be
\[ \norm{\Phi}_{\diamond} \coloneqq \max_{\brho \in \bfD(\cH \otimes \cX)}{\norm{(\Phi \otimes \mathbb{I}_{\bfS(\cX)})(\brho)}_1}\]
where $\cX$ is any Hilbert space with the same dimension as $\cH$, and $\norm{\cdot}_1$ is the Schatten 1-norm.

We occasionally denote the action of a unitary $U$ on a mixed state $\brho$ by $U(\brho) := U \brho U^\dagger$.

\paragraph{Computational indistinguishability.} The following is taken verbatim from~\cite{C:BKLMMVVY22}. Two quantum state ensembles $\{\brho_0^{(\secp)}, \brho_1^{(\secp)}\}_\secp$ are said to be \emph{computationally indistinguishable} if for every non-uniform QPT algorithm $A = \{A^{(\secp)}, \brho^{(\secp)} \}$ (that outputs a bit $b$), we have that 
\[\Big|  \mathop{\mathbb{E}}\left[A^{(\secp)}(\brho^{(\secp)},\brho^{(\secp)}_0)\right] -  \mathop{\mathbb{E}}\left[A^{(\secp)}(\brho^{(\secp)},\brho^{(\secp)}_1)\right]\Big| = \negl(\secp).
\]

Equivalently, $\{\brho_0^{(\secp)}, \brho_1^{(\secp)}\}_\secp$ are computationally indistinguishable if for every efficiently computable non-uniform binary observable ($R, \bsigma)$, we have that 
\[\Big|\Tr(R (\brho_0 \tensor \bsigma) ) - \Tr(R (\brho_1 \tensor \bsigma)) \Big| = \negl(\secp).
\]
We will occasionally use the notation $\brho_0 \approx_c \brho_1$ to denote computational indistinguishability of $\{\brho_0^{(\secp)}, \brho_1^{(\secp)}\}_\secp$. When $\brho_0$ and $\brho_1$ are statistically indistinguishable, i.e. $\norm{\brho_0 - \brho_1}_1 = \negl(\lambda)$, we write $\brho_0 \approx_s \brho_1$.

\newpage

\section{Defining quantum state commitments (QSCs)}
\label{sec:definitions}

\subsection{Syntax}

We describe the syntax for non-interactive quantum state commitments. 

\begin{definition}[Non-interactive QSC syntax] \label{def:noninteractive_qsc}
A non-interactive quantum state commitment $\QSC$ is specified by a unitary $\Com_{\QSC}$ on $n+\lambda$ qubits. 
\begin{itemize}
    \item (Commitment phase) To commit to a $n$-qubit state in register $\RegM$, the sender initializes register $\RegW$ to the $\ket*{0^{\lambda}}$ state and applies $\Com_{\QSC}$ to $(\RegM,\RegW)$. The resulting state is divided into registers $(\RegC,\RegD)$. It sends $\RegC$ to the receiver.
    \item (Opening phase) The sender decommits by sending $\RegD$. The receiver applies $\Com_{\QSC}^\dagger$ to the pair $(\RegC, \RegD)$, obtaining $(\RegM, \RegW)$. It verifies the decommitment by measuring $\RegW$ in the standard basis and checking that the outcome is $\ket*{0^{\lambda}}$. If verification accepts, the receiver interprets the state on $\RegM$ as the committed message.
\end{itemize}
\end{definition}
For convenience, we will often write $\Com$ in place of $\Com_{\QSC}$. 

\paragraph{What about interactive commitments?} Our swap-binding definition (\cref{def:swap-binding}) will only apply to non-interactive QSCs. This is justified in part by the fact that any interactive QSC can be made non-interactive via a simple transformation that was previously known in the setting of quantum \emph{bit} commitments (see~\cref{sec:non-interactive} for full details). Nevertheless, we believe it could be useful in some settings to have a definition that directly handles the interactive case. Towards this end, we state an alternative definition of binding that we call ``Pauli binding'' in~\cref{subsec:pauli-binding}. Pauli binding naturally handles interactive QSCs and is equivalent to swap-binding for non-interactive QSCs. We choose to present swap binding as our main definition since we found it significantly easier to work with.

\subsection{The swap binding definition}
\label{sec:swap-binding-def}

We now present our new swap binding security definition. For two quantum registers of the same size, $\RegM$ and $\RegM'$, $\SWAP[\RegM, \RegM']$ denotes the unitary that maps $\ket{\psi}_{\RegM} \ket{\phi}_{\RegM'}$ to $\ket{\phi}_{\RegM} \ket{\psi}_{\RegM'}$.

\begin{definition}[Swap binding] \label{def:swap-binding}
For a non-interactive quantum commitment scheme $\Com$, an interactive adversary $A$, a challenge bit $b \in \{0,1\}$, and a security parameter $\lambda$, define the swap binding security experiment $\emph{\texttt{SwapBindExpt}}_{\QSC,A,b}(\lambda)$ as follows.

\indent $\emph{\texttt{SwapBindExpt}}_{\QSC,A,b}(\lambda)$:
\begin{enumerate}
    \item The adversary $A$ (acting as a malicious sender) sends commitment register $\RegC$ and decommitment register $\RegD$ to the challenger (acting as an honest receiver).
    \item The challenger applies $\Com^\dagger$ and measures $\{\ketbra{0}, \Id - \ketbra{0}\}$ on $\RegW$; if the measurement rejects, then it aborts and outputs a random $b' \gets \{0,1\}$.
    \item\label[step]{step:swap-bind-expt-challenge} The challenger does the following: 
    \begin{itemize}
        \item if $b = 0$, the challenger simply applies $\Com$, and sends the $\RegD$ register to $A$;
        \item if $b = 1$, the challenger replaces the contents of the $\RegM$ register with $\ket{0}$ (i.e., it initializes a fresh register $\RegE = \ket{0}$ of the same dimension as $\RegM$ and applies $\SWAP[\RegM, \RegE]$). It then applies $\Com$ and sends the $\RegD$ register to $A$.
    \end{itemize}
    \item The output of the experiment is $b' \gets A$.
\end{enumerate}

$\Com$ is computationally (resp. statistically) swap binding if there exists a negligible function $\mu(\lambda)$ such that for all polynomial-time (resp. unbounded-time) quantum interactive adversaries $A$,
\[ \Pr_{b \gets \{0,1\}} [\emph{\texttt{SwapBindExpt}}_{\QSC,A,b}(\lambda)=b] \leq \frac{1}{2} + \mu(\lambda). \]
\end{definition}

\subsection{Additional definitions: hiding and succinctness}

We now define hiding and succinctness for QSCs. While the notion of hiding for \emph{commitments} to quantum messages had not been defined in the literature before, it is easy to write down a definition based on existing definitions for \emph{encryption} of quantum messages. The following definition is essentially the same as the definition of indistinguishability under chosen-plaintext-attacks (q-IND-CPA) given by~\cite{C:BroJef15}.

\begin{definition}[Hiding for quantum state commitments] \label{def:quantum-hiding}
For a non-interactive quantum commitment scheme $\Com$, an interactive adversary $A$, a challenge bit $b \in \{0,1\}$, and a security parameter $\lambda$, define a security experiment $\emph{\texttt{HideExpt}}_{\QSC,A,b}(\lambda)$ as follows.

\indent $\emph{\texttt{HideExpt}}_{\QSC,A,b}(\lambda)$:
\begin{enumerate}
    \item\label[step]{step:qsc-hiding-expt-measurement} $A$ prepares a message $\RegM$ and sends it to the challenger. 
    \item\label[step]{step:hiding-commit} Next, the challenger (acting as an honest sender) performs the commit phase with $A$ (acting as the receiver) using the quantum state in $\RegM$ if $b=0$, or $\ket{0}$ if $b=1$.
    \item The output of the experiment is $b' \gets A$. 
\end{enumerate}

$\Com$ is computationally (resp. statistically) hiding if there exists a negligible function $\mu(\lambda)$ such that for all polynomial-time (resp. unbounded-time) quantum interactive adversaries $A$,
\[ \Pr_{b \gets \{0,1\}} [\emph{\texttt{HideExpt}}_{\QSC,A,b}(\lambda)=b] \leq \frac{1}{2} + \mu(\lambda). \]
\end{definition}

\begin{definition}[Succinct non-interactive QSCs] \label{def:succinct-quantum-commitments}
We say that a non-interactive QSC is \emph{succinct} if the commitment is shorter than the message, i.e., the size of register $\RegC$ is smaller than $\RegM$.
\end{definition}

We remark that our definition of succinct QSCs does not require hiding. This is analogous to the definitions of succinct classical commitments (i.e., collision-resistant hash functions, or their post-quantum analogue, collapsing hash functions~\cite{EC:Unruh16}), which only require binding and succinctness but not hiding. Moreover, hiding is not necessary for our primary application to quantum succinct arguments (\cref{sec:quantum-kilian}).

\subsection{Hiding-binding duality}
\label{subsec:duality}
We observe that the hiding and binding security experiments are nearly identical for non-interactive QSCs. In particular, we can write both of the experiments as follows.

\begin{property}[Hiding-binding duality for QSCs]
\label{property:hiding-binding-duality}
    For a non-interactive quantum commitment scheme $\Com$, an interactive adversary $A$, a challenge bit $b \in \{0,1\}$, and a security parameter $\lambda$, the binding and hiding security experiments can be written as follows.
    \begin{enumerate}
        \item The adversary $A$ sends $(\RegC, \RegD)$ to the challenger.
        \item The challenger applies $\Com^{\dagger}$ and measures $\{\ketbra{0}_{\RegW}, \Id - \ketbra{0}_{\RegW}\}$. If the measurement rejects, abort the experiment (i.e., output a random bit $b'$). Otherwise:
        \begin{enumerate}
            \item The challenger initializes a register $\RegE$ to $\ket{0}$ and applies $\SWAP[\RegM, \RegE]^b$, then $\Com$.
            \item For binding, the challenger returns $\RegD$; for hiding, the challenger returns $\RegC$.
        \end{enumerate}
        \item The output of the experiment is $b' \gets A$.
    \end{enumerate}
\end{property}
\begin{proof}
The binding game is exactly equivalent to \cref{def:swap-binding}. The hiding game is equivalent to \cref{def:quantum-hiding}, except that here the adversary sends $(\RegC, \RegD)$ whereas the adversary in \cref{def:quantum-hiding} sends $\RegM$ (and there is no $\RegW$ measurement). However, these notions of hiding are equivalent by essentially the same argument as \cite[Lemma 14]{EC:Unruh16}. The proof follows by observing that sending $(\RegC, \RegD)$ registers containing an invalid commitment-decommitment pair cannot help the adversary; if the adversary sends $(\RegC, \RegD)$ registers containing a valid commitment-decommitment pair, this is equivalent to sending $\RegM$ (in the $\Com$ basis).
\end{proof}

\cref{property:hiding-binding-duality} motivates the definition of a \emph{dual scheme}, where the roles of the $\RegC$ and $\RegD$ registers are reversed.
\begin{definition}[Dual commitments]
For a non-interactive commitment scheme $\Com$ with commitment register $\RegC$ and decommitment register $\RegD$, the \emph{dual scheme} is the commitment scheme with the same commitment unitary $\Com$, but $\RegD$ is used as the commitment and $\RegC$ is used as the decommitment.
\end{definition}

As an immediate consequence of \cref{property:hiding-binding-duality}, we obtain the following corollary.
\begin{corollary}
    Let $\Com$ be a computationally (resp. statistical) swap-binding and statistically (resp. computationally) hiding non-interactive commitment scheme. Then the dual commitment scheme $\Com$ is a computationally (resp. statistical) hiding and statistically (resp. computationally) swap-binding non-interactive commitment scheme.
\end{corollary}

\subsubsection{Discussion: is there a hiding-binding duality for QBCs?}
\label{subsubsec:qbc-duality}
As it turns out, a similar duality also exists for quantum bit(string) commitments (QBCs), i.e., quantum-communication commitments to \emph{classical messages}. Previously,~\cite{AC:Yan22} showed that it was possible to transform any statistically binding, computationally hiding QBC into a computationally binding, statistically hiding QBC (or vice versa). However, using the insights developed above for QSCs, we can make this connection between hiding and binding significantly simpler.

For the following discussion, we consider QBCs satisfying our syntax for non-interactive QSCs. This is essentially Yan's ``canonical form''~\cite{AC:Yan22} with a few syntactic differences (most notably, we allow for arbitrary-length messages, not just bits). 

Consider the following version of Unruh's collapse-binding definition~\cite{EC:Unruh16} for \emph{classical messages} (restated for QBCs):
\begin{enumerate}
    \item[] Computational (resp. statistical) security requires that any efficient (resp. inefficient) $A$ guesses $b$ with probability at most $1/2 + \negl(\lambda)$ in the following experiment:
    \item The adversary $A$ sends $(\RegC,\RegD)$ to the challenger. 
    \item The challenger applies $\Com^\dagger$ and measures $\{\ketbra{0}_{\RegW}, \Id-\ketbra{0}_{\RegW}\}$. If the measurement rejects, abort the experiment (i.e., output a random bit $b'$). Otherwise:
    \begin{enumerate}
        \item The challenger samples a random bitstring $s$ (whose length equals the number of qubits of $\RegM$) and applies $Z^s$ to $\RegM$. 
        \item Next, the challenger applies $\Com$ and returns $\RegD$ to the adversary.
    \end{enumerate}
    \item The adversary wins the game if it guesses $b$.
\end{enumerate}

Security requires that it is hard to guess $b$ with advantage noticeably greater than $1/2$. How does this security experiment relate to hiding? Suppose we modify the collapse-binding experiment as follows (differences in \textcolor{red}{red}):
\begin{enumerate}
    \item[] Computational (resp. statistical) security requires that any efficient (resp. inefficient) $A$ guesses $b$ with probability at most $1/2 + \negl(\lambda)$ in the following experiment:
    \item The adversary $A$ sends $(\RegC,\RegD)$ to the challenger. 
    \item The challenger applies $\Com^\dagger$ and measures $\{\ketbra{0}_{\RegW}, \Id-\ketbra{0}_{\RegW}\}$. If the measurement rejects, abort the experiment (i.e., output a random bit $b'$). Otherwise:
    \begin{enumerate}
        \item The challenger samples a random bitstring $s$ (whose length equals the number of qubits of $\RegM$) and applies \textcolor{red}{$X^s$} to $\RegM$. 
        \item Next, the challenger applies $\Com$ and returns \textcolor{red}{$\RegC$} to the adversary.
    \end{enumerate}
    \item The adversary wins the game if it guesses $b$.
\end{enumerate}
It is not hard to see that this definition, which we call ``$X$-hiding,'' implies the commitment hides any \emph{classical message}. In particular, if the adversary sends any classical message $\ket{m}$ (i.e., by initializing $(\RegC,\RegD)$ to $\Com( \ket{m}\ket{0}_\RegW)$) then the game above amounts to distinguishing a commitment to $m$ from a commitment to a uniformly random string. However, $X$-hiding is stronger than ordinary hiding for classical messages, since the adversary may choose a superposition of classical messages $\ket{m}$.

This connection between $X$-hiding and collapse-binding (which we will call $Z$-binding here) leads to a simple dual scheme for any QBC satisfying our non-interactive commitment syntax. In particular, suppose $\Com$ is statistically $Z$-binding and computationally $X$-hiding. Now consider the following ``dual'' commitment scheme:
\begin{itemize}
    \item To commit to an $n$-bit classical string $\ket{m}$, apply $H^{\otimes n}$ followed by $\Com$. Send the decommitment register $\RegD$ as the commitment.
    \item To decommit, send the commitment $\RegC$. To verify, apply $H^{\otimes n}\Com^\dagger$ and recover the $\ket{m}$ from the $\RegM$ register.
\end{itemize}
It is immediate from these definitions that this scheme scheme is \emph{computationally} $Z$-binding and statistically $X$-hiding!

\subsection{Parallel composition}
\label{subsec:parallel}
In this section, we show that our binding definition for quantum messages composes in parallel. 

\begin{theorem}[Binding composes in parallel] \label{theorem:parallel_composition}
Let $\{ \Com_i \}_{i \in [n]}$ be a set of $n$ non-interactive binding commitment schemes with message and ancillary registers $(\RegM_i, \RegW_i)_{i \in [n]}$. Then the parallel composition of the commitment schemes $\widetilde{\Com} \coloneqq \bigotimes \Com_i$ with message register $\RegM \coloneqq (\RegM_i)_{i \in [n]}$ and ancilla register $\RegW \coloneqq (\RegW_i)_{i \in [n]}$ is also a binding commitment scheme. 
\end{theorem}
\begin{proof}
    We will show by a sequence of $n+1$ hybrids that the parallel composition of the $n$ commitment schemes is binding. That is, we will show that an adversary cannot distinguish between decommitments for messages of their choice in all $n$ slots or for $\ket{0}$ in all $n$ slots. We define hybrids $H_j$ for $j \in \{0,1,\dots,n\}$ as follows.
    
    \noindent
    \begin{itemize}
        \item[] $\textbf{Hybrid}$ $H_j$:
        \begin{enumerate}
            \item The adversary sends $n$ commitment and decommitment registers $(\RegC_i, \RegD_i)_{i\in [n]}$ to the challenger. 
            \item The challenger then does the following:
            \begin{enumerate}
            \item Apply $\widetilde{\Com}^{\dagger}$. 
            \item Measure $\{\ketbra{0}_{\RegW} , \Id - \ketbra{0}_{\RegW}\}$. If the measurement accepts, then continue to the next step. Otherwise, the experiment aborts and a random bit $b' \leftarrow \{0,1\}$ is output. 
            \item Swap the messages in $\RegM_i$ for $i \leq j$ with $\ket{0}$.
            \item Apply $\widetilde{\Com}$.
        \end{enumerate}
        \item The challenger returns the decommitment $(\RegD_i)_{i \in [n]}$ to the adversary. 
        \end{enumerate}
    \end{itemize}
    
    Breaking binding for the commitment scheme $\widetilde{\Com}$ is equivalent to distinguishing between hybrids $H_0$ and $H_n$. Suppose that an adversary $A$ can distinguish hybrids $H_0$ and $H_n$ with advantage $\varepsilon$. Then for some $j \in [n]$, the adversary $A$ can distinguish hybrid $H_{j-1}$ from hybrid $H_j$ with advantage $\varepsilon / n$. We now use $A$ to construct an adversary $A'$ on the binding experiment for the commitment scheme $\Com_j$. The adversary $A'$ receives the registers $(\RegC_i, \RegD_i)_{i=1}^n$ from $A$, then does the following:
    \begin{enumerate}
        \item Apply $\widetilde{\Com}^{\dagger}$.
        \item Measure $\{\ketbra{0}_{\RegW} , \Id - \ketbra{0}_{\RegW}\}$. If the measurement accepts, then continue to the next step. Otherwise, output a random bit $b' \leftarrow \{0,1\}$.
        \item Swap the messages in $\RegM_i$ for $i \leq j - 1$ with $\ket{0}$.
        \item Apply $\widetilde{\Com}$. 
        \item Forward $(\RegC_j, \RegD_j)$ to the challenger, receive $\RegD_j$ from the challenger, and send $(\RegD_i)_{i=1}^n$ to $A$. 
        \item Output $b' \leftarrow A$.
    \end{enumerate}
    By construction of $A'$, the view of $A$ in which the challenger either does nothing or swaps out the $j$th message with $\ket{0}$ is identical to that of $H_{j-1}$ and $H_j$, respectively, so $A'$ wins the binding experiment for $\Com_j$ with advantage $\varepsilon / n$. 
\end{proof}

\newpage

\section{Constructions of QSCs}
\label{sec:constructions}

\subsection{Succinct QSCs from one-time quantum encryption}
\label{sec:succinct-constructions}
\begin{definition}[Quantum encryption syntax] 
    A quantum encryption scheme $\mathsf{QEnc}$ for quantum message space $\RegM$ and classical key space $\{0,1\}^{d(\lambda)}$ is specified by a family $\{U_k\}_{k \in \{0,1\}^{d(\lambda)}}$ of unitaries acting on $\RegM$, for each security parameter $\lambda$. We require that for any $k \in \{0,1\}^{d(\lambda)}$, there is a $\poly(\lambda)$-time classical procedure to generate a description of $U_k$. The encryption of $\ket{\psi} \in \RegM$ under secret key $k$ is $U_k \ket{\psi}$.
\end{definition}
To simplify notation, we will often drop the dependence on $\lambda$ and write $\QEnc = \{U_k\}_k$. We say that a quantum encryption scheme has \emph{short keys} if $d < n$. We define \emph{one-time security} for a quantum encryption scheme $\mathsf{QEnc}$ with the following security game:

\begin{definition}[One-time secure quantum encryption] \label{def:one-time-quantum-encryption}
For a quantum encryption scheme $\QEnc$, an interactive adversary $A$, a challenge bit $b \in \{0,1\}$, and a security parameter $\lambda$, define a security experiment $\emph{\texttt{QEncExpt}}_{\QSC,A,b}(\lambda)$ as follows.

\indent $\emph{\texttt{QEncExpt}}_{\QEnc,A,b}(\lambda)$:
\begin{enumerate}
    \item $A$ sends a quantum message $\RegM$ to the challenger. 
    \item The challenger samples a random $k \from \{0,1\}^{d(\lambda)}$ and does the following: 
        \begin{enumerate}
            \item If $b=0$, apply $U_k$ to $\RegM$ and send $\RegM$ to $A$. 
            \item If $b=1$, swap the contents of $\RegM$ with $\ket{0}$, apply $U_k$ to $\RegM$, and send $\RegM$ to $A$. 
        \end{enumerate}
    \item The output of the experiment is $b' \gets A$. 
\end{enumerate}
$\QEnc$ is one-time secure if there exists a negligible function $\mu(\lambda)$ such that for all polynomial-time quantum adversaries $A$,
\[ \Pr_{b \gets \{0,1\}} [\emph{\texttt{QEncExpt}}_{\QEnc,A,b}(\lambda)=b] \leq \frac{1}{2} + \mu(\lambda). \]
\end{definition}

Using any one-time secure quantum encryption scheme $\QEnc = \{U_k\}_k$, we can define a non-interactive QSC by the following commitment unitary:
\begin{equation}
\label{eq:qsc-from-pru}
    \Com = \sum_{k \in \{0,1\}^{d}} \left(\ketbra{k}{k} H^{\otimes d}\right)_\RegW \otimes (U_k)_\RegM,
\end{equation}
using $\RegC = \RegW$ as the commitment and $\RegD = \RegM$ as the decommitment. For example, the honest commitment to the pure state $\ket{\psi}_{\RegM}$ is
\[ 
    2^{-d/2} \sum_{k \in \{0,1\}^{d}} \ket{k}_{\RegC} \otimes (U_k \ket{\psi})_{\RegD}.
\]
If the encryption scheme has short keys, then the scheme is succinct. It is easy to show that the scheme is also binding.

\begin{theorem} \label{theorem:pru-construction}
If $\{U_k\}_k$ is a one-time secure quantum encryption scheme, then the QSC scheme $\Com$ defined by \cref{eq:qsc-from-pru} is binding.
\end{theorem}
\begin{proof}
The proof is immediate from the definitions. Consider the binding security experiment for $\Com$ (\cref{def:swap-binding}). We denote the register containing the adversary's internal state by $\RegR$. Let $\brho_{\RegC,\RegD,\RegR}$ be the state of the experiment after the adversary sends $\RegD$ to the challenger, let $\brho_{\RegD,\RegR}^b$ be the state held by the adversary after it receives $\RegD$ back from the challenger in the $b \in \{0,1\}$ world of the binding experiment, and let $\sigma_{\RegM,\RegR} = \Tr_\RegC(\Pi \cdot \Com^\dagger \brho_{\RegC,\RegD,\RegR} \Com \cdot \Pi)$ be the sub-normalized state on $(\RegM, \RegR)$ after a successful opening. By the one-time security of the quantum encryption scheme,

\begin{align*}
    \brho_{\RegD,\RegR}^0 &= 2^{-d} \sum_{k\in\{0,1\}^{d}} (U_k)_\RegM \sigma_{\RegM,\RegR} (U_k^\dagger)_\RegM \\
    &\approx_c 2^{-d} \sum_{k \in \{0,1\}^d} \left( U_k \ketbra{0}_{\RegM} U_k^{\dagger} \right)_{\RegM} \otimes \sigma_{\RegR} \\
    &= \brho_{\RegD,\RegR}^1. \qedhere
\end{align*}
\end{proof}

\begin{claim}[One-time secure quantum encryption with short keys from OWFs]
One-time secure quantum encryption with short keys exists assuming the existence of any post-quantum one-way function.
\end{claim}
\begin{proof}
There exists a post-quantum pseudo-random generator $G : \{0,1\}^{n/2} \to \{0,1\}^{2n}$ assuming the existence of any post-quantum one-way function \cite{SICOMP:HILL99,FOCS:Zhandry12}. Let $U_k = X^{G_0(k)} Z^{G_1(k)}$ be a quantum encryption scheme, where $G_0(k),G_1(k)$ are the first and last $n/2$ bits of $G(k)$, respectively. This scheme has key size $d=n/2$ for $n$-qubit messages, and security is immediate from the definitions.

Consider the security experiment for $\{U_k\}_k$ (\cref{def:one-time-quantum-encryption}). We denote the register containing the adversary's internal state by $\RegR$. Let $\brho_{\RegM,\RegR}$ be the state of the experiment after the adversary sends $\RegM$ to the challenger and let $\brho_{\RegM,\RegR}^b$ be the state held by the adversary after it receives $\RegM$ back from the challenger in the $b \in \{0,1\}$ world of the security experiment. Since $G$ is a post-quantum pseudo-random generator,
\begin{align*}
    \brho_{\RegM,\RegR}^0 &= 2^{-d} \sum_{k\in\{0,1\}^{d}} (X^{G_0(k)} Z^{G_1(k)})_\RegM \brho_{\RegM,\RegR} (Z^{G_1(k)} X^{G_0(k)})_\RegM \\
    &\approx_c 2^{-d} \sum_{k\in\{0,1\}^{d}} (X^{G_0(k)} Z^{G_1(k)} \ketbra{0}_{\RegM} Z^{G_1(k)} X^{G_0(k)})_\RegM \otimes \brho_{\RegR} \\
    &= \brho_{\RegM,\RegR}^1. \qedhere
\end{align*}
\end{proof}

\begin{claim}
\label{claim:oracle-separation}
There is a quantum oracle $\mathcal{O}$ relative to which one-time secure quantum encryption with short keys exists, but $\BQP^\mathcal{O} = \QMA^\mathcal{O}$ (and in particular, post-quantum one-way functions do not exist).
\end{claim}
\begin{proof}
By \cite[Theorem 2]{Kretschmer21}, there is a quantum oracle $\mathcal{O}$ relative to which a pseudo-random unitary family $\{U_k : \mathbb{C}^{2^n} \to \mathbb{C}^{2^n}\}_{k \in \{0,1\}^n}$ exists, but $\BQP^\mathcal{O} = \QMA^\mathcal{O}$. We would like to view the family as a quantum encryption scheme with keys $k \in \{0,1\}^n$, but this scheme does not have short keys (as $d=n$).

However, \cref{sec:pru-expansion} shows that we can exchange the PRU family's many-time security for one-time security on a larger message space. We therefore use $\{\Expand(U_k, \ell)\}_{k \in \{0,1\}^n}$ for any constant $\ell > 1$, which is a secure one-time quantum encryption scheme acting on at least $\ell \cdot n$ qubits by \cref{theorem:pru-expansion}.
\end{proof}

\subsection{Domain extension}
\label{subsec:domain-extension}

In this section, we show that a quantum analogue of Merkle-Damgård domain extension for classical compressing hash functions works for succinct QSCs.

Concretely, suppose we are given a \emph{succinct} QSC scheme $\SQSC$ for $(m+1)$-qubit messages that produces $m$-qubit commitments.\footnote{One-qubit compression is without loss of generality, since if a scheme compresses by more than one qubit, we can pad the input with $0$'s (and additionally require that these qubits are $0$ during verification).} That is, $\Com_{\SQSC}$ acts on message-ancilla registers $(\RegM,\RegW)$ and produces commitment-decommitment registers $(\RegC,\RegD)$ where:
\begin{itemize}
\item $\RegM$ is an $(m+1)$-qubit message register,
\item $\RegW$ is a $d$-qubit ancilla register,
\item $\RegC$ is an $m$-qubit commitment register, and
\item $\RegD$ is a $(d+1)$-qubit decommitment register.
\end{itemize}

Using $\SQSC$, we construct a new succinct QSC scheme $\SQSC^{\MD}_k$ (where $\MD$ stands for Merkle-Damgård) enabling commitments to $k$-qubit messages with commitments of $m$ qubits for any $k = \poly(\lambda)$.

In $\SQSC^{\MD}_k$, the sender commits to a $k$-qubit message $\RegM$ as follows:
\begin{enumerate}
    \item[] \textbf{Setup:} Initialize an $(m + k d)$-qubit register $\RegW$ to $\ket{0^{m+kd}}$. Separate $\RegW$ into subregisters $\RegW = (\RegC_0,\RegY_0,\RegY_1,\dots,\RegY_{k-1})$ where $\RegC_0$ is $m$ qubits and each $\RegY_i$ is $d$ qubits. Separate $\RegM$ into subregisters $\RegM = (\RegM_0,\RegM_1,\dots,\RegM_{k-1})$ where each $\RegM_i$ is a single qubit register.
    \item For $i = 0,\dots,k-1$:
    \begin{enumerate}
        \item[] Apply $\Com_{\SQSC}$ to $(\RegM_i,\RegC_i,\RegY_i)$, i.e., commit to the $(m+1)$-qubit message register $(\RegM_i,\RegC_i)$ using the $d$-qubit ancilla $\RegY_i$, to obtain $(\RegC_{i+1},\RegD_{i+1})$.
    \end{enumerate}
    \item The commitment is the $m$-qubit register $\RegC = \RegC_k$ and the decommitment is the $k \cdot (d+1)$-qubit register $\RegD = (\RegD_1,\dots,\RegD_k)$.
\end{enumerate}

\begin{theorem}[Domain extension] \label{theorem:domain-extension}
For any $k = \poly(\lambda)$, the scheme $\SQSC^{\MD}_k$ is swap binding assuming that $\SQSC$ is swap binding.
\end{theorem}
\begin{proof}
    We prove this claim by a sequence of hybrid arguments. To simplify notation, let $U_j$ be the unitary corresponding to the $j$th application of $\Com_{\SQSC}$ in $\SQSC^{\MD}_k$, so that the commitment is implemented with the unitary $U \coloneqq U_k U_{k-1} \cdots U_1$. We first state all of the hybrids below and then prove indistinguishability of the hybrids in~\cref{claim:hybrid01,claim:hybrid12,claim:hybrid23}. We highlight the differences from the previous hybrid in \textcolor{red}{red}.
    
    \begin{itemize}
        \item \textbf{Hybrid} $H_0$:
        \begin{enumerate}
            \item The adversary sends the registers $(\RegC, \RegD)$ to the challenger.
            \item The challenger does the following:
            \begin{enumerate}
                \item Apply a binary projective measurement to check that $(\RegC, \RegD)$ is valid.
                \item Return $\RegD = (\RegD_1,\dots,\RegD_k)$ to the adversary.
            \end{enumerate}
        \end{enumerate}
    \end{itemize}
        Hybrid $H_0$ corresponds to the $b = 0$ case of the $\QSC$ binding security game~\cref{def:swap-binding}.
    \begin{itemize}
        \item \textbf{Hybrid} $H_1$:
        \begin{enumerate}
            \item The adversary sends $(\RegC, \RegD)$ to the challenger.
            \item The challenger does the following:
            \begin{enumerate}
                \item Apply a binary projective measurement to check that $(\RegC, \RegD)$ is valid.
                {\color{red}
                \item Apply $U^\dagger$ to recover the committed message on registers $(\RegM_0,\dots,\RegM_{k-1})$.
                \item For each $i = 0,\dots,k-1$:
                \begin{enumerate}
                    \item Initialize an $m$-qubit register $\RegC_i'$ to $\ket{0^m}$.
                    \item Use $\Com$ to commit to the $m+1$-qubit register $(\RegM_i, \RegC_i')$, producing a decommitment $\RegD_{i+1}$ and a commitment $\RegC_{i+1}$; discard $\RegC_{i+1}$.
                \end{enumerate}}
                \item Return $\RegD = (\RegD_1, \dots,\RegD_k)$ to the adversary. 
            \end{enumerate}
        \end{enumerate}
    \end{itemize}
        The difference between $H_1$ and $H_0$ is that in $H_1$, the decommitments $\RegD_1,\dots,\RegD_k$ are each generated \emph{in parallel}. That is, in hybrid $H_1$, each $\RegD_i$ the adversary receives is the decommitment resulting from committing to $(\RegM_{i-1}, \RegC'_{i-1})$ where $\RegC'_{i-1}$ is initialized to $\ket{0^m}$. In contrast, in hybrid $H_0$, the decommitments $\RegD_i$ are generated \emph{in sequence}, i.e., $\RegD_i$ is the decommitment resulting from committing to $(\RegM_{i-1}, \RegC_{i-1})$, where $\RegC_{i-1}$ is the commitment resulting from committing to $(\RegM_{i-2}, \RegC_{i-2})$ in the previous layer, etc.
        
        We prove that $H_0$ and $H_1$ are indistinguishable in~\cref{claim:hybrid01} by defining a careful sequence of $k$ sub-hybrids.
    \begin{itemize}
        \item \textbf{Hybrid} $H_2$:
        \begin{enumerate}
            \item The adversary sends $(\RegC, \RegD)$ to the challenger.
            \item The challenger does the following:
            \begin{enumerate}
                \item Apply a binary projective measurement to check that $(\RegC, \RegD)$ is valid.
                \item Apply $U^\dagger$ to recover the committed message on register $\RegM = (\RegM_0,\dots,\RegM_{k-1})$.
                {\color{red}
                \item Initialize $k$ qubits $\RegM' = (\RegM_0', \dots, \RegM_{k-1}')$ to $\ket{0^k}$, and apply $\SWAP[\RegM,\RegM']$.}
                \item For each $i = 0,\dots,k-1$:
                \begin{enumerate}
                    \item Initialize an $m$-qubit register $\RegC_i'$ to $\ket{0^m}$.
                    \item Use $\Com$ to commit to the $m+1$-qubit register $(\RegM_i, \RegC_i')$ and obtain $(\RegD_{i+1},\RegC_{i+1})$; discard $\RegC_{i+1}$.
                \end{enumerate}
                \item Return $\RegD = (\RegD_1, \dots, \RegD_k)$ to the adversary.
            \end{enumerate}
        \end{enumerate}
    \end{itemize}
        The difference between $H_2$ and $H_1$ is that in $H_2$, the decommitments $\RegD_1,\dots,\RegD_k$ are each decommitments resulting from (independent) commitments to the message $\ket{0^{m+1}}$. Since the $k$ decommitments in both $H_1$ and $H_2$ are all independent, indistinguishability of these hybrids follows from the fact that swap-binding composes in parallel (\cref{theorem:parallel_composition}); see \cref{claim:hybrid12}.
    \begin{itemize}
        \item \textbf{Hybrid} $H_3$:
        \begin{enumerate}
            \item The adversary sends $(\RegC, \RegD)$ to the challenger.
            \item The challenger does the following:
            \begin{enumerate}
                \item Apply a binary projective measurement to check that $(\RegC, \RegD)$ is valid.
                \item Apply $U^\dagger$ to recover the committed message on registers $(\RegM_0, \dots, \RegM_{k-1})$.
                \item Initialize $k$ qubits $\RegM' = (\RegM_0', \dots,\RegM_{k-1}')$ to $\ket{0^k}$, and apply $\SWAP[\RegM,\RegM']$.
                {\color{red}
                \item Apply $U$.}
                \item Return $\RegD = (\RegD_1, \dots, \RegD_k)$ to the adversary.
            \end{enumerate}
        \end{enumerate}
    \end{itemize}
    Hybrid $H_3$ corresponds to the $b = 1$ case of the $\QSC$ binding security game (\cref{def:swap-binding}). The differences between $H_3$ and $H_2$ are analogous to the difference between $H_0$ and $H_1$, except that now each $\RegM_i$ is replaced with $\ket{0}$. We prove $H_3 \approx_c H_2$ in~\cref{claim:hybrid23} using essentially the same arguments as in~\cref{claim:hybrid01}.
    
    It remains to prove~\cref{claim:hybrid01,claim:hybrid12,claim:hybrid23}.
\end{proof}
    
    \begin{claim}
    \label{claim:hybrid01}
    Hybrid $H_0$ is indistinguishable from Hybrid $H_1$.
    \end{claim}
    \begin{proof}
    We prove this claim by a sequence of $k+1$ hybrids $H_0^{(0-1)}, \dots, H_k^{(0-1)}$. Hybrid $H_0^{(0-1)}$ is identical to $H_0$, and $H_k^{(0-1)}$ is identical to $H_1$. For $i = 0, \hdots, k$ we define $H_i^{(0-1)}$ by the following behavior for the challenger:
    \noindent
    \begin{itemize}
        \item[] \textbf{Hybrid} $H_i^{(0-1)}$:
        \begin{enumerate}
            \item The adversary sends $(\RegC, \RegD)$.
            \item The challenger does the following:
            \begin{enumerate}
                \item Check that the commitment is valid. If not, the experiment aborts.
                \item If $i>0$, open the last $i$ commitments using $U_{k-i+1}^\dagger \cdots U_k^\dagger$, initialize a new register $(\RegM_{k-i}',\RegC_{k-i}')$ to $\ket{0}$, and apply $\SWAP[(\RegM_{k-i},\RegC_{k-i}),(\RegM_{k-i}',\RegC_{k-i}')]$.
                \item If $i>1$, for $j = k-i, \dots, k-2$ in sequence, apply $\Com$ to $(\RegM_j,\RegC_j,\RegY_j)$ to produce a decommitment on $\RegD_{j+1}$ and a commitment on $\RegC_{j+1}$, initialize a new register $(\RegM_{j+1}',\RegC_{j+1}')$ to $\ket{0}$, and apply $\SWAP[(\RegM_{j+1},\RegC_{j+1}),(\RegM_{j+1}',\RegC_{j+1}')]$.
                \item If $i>0$, apply $\Com$ to $(\RegM_{k-1}, \RegC_{k-1},\RegY_{k-1})$ to produce a decommitment on $\RegD_k$ and a commitment on $\RegC_k$.
            \end{enumerate}
            \item The challenger returns the decommitment $\RegD = (\RegD_1, \dots, \RegD_k)$ to the adversary.
        \end{enumerate}
    \end{itemize}
    Suppose that an adversary $A$ can distinguish hybrids $H_0^{(0-1)}$ and $H_k^{(0-1)}$ with advantage $\varepsilon$. Then for some $i \in \{0,\dots,k-1\}$, the adversary $A$ can distinguish hybrid $H_i^{(0-1)}$ from hybrid $H_{i+1}^{(0-1)}$ with advantage $\varepsilon / k$. We now construct a reduction $R_i^{(0-1)}$ that uses such an adversary to achieve advantage $\varepsilon / k$ in the binding experiment for the commitment scheme $\Com$.
    
    \noindent
    \begin{itemize}
        \item[] \textbf{Reduction} $R_i^{(0-1)}$:
        \begin{enumerate}
            \item The adversary sends $(\RegC,\RegD)$ to $R_i^{(0-1)}$.
            \item The reduction does the following:
            \begin{enumerate}
                \item Check that the commitment is valid. If not, the experiment aborts and the reduction makes a random guess $b' \from \{0,1\}$.
                \item If $i>0$, open the last $i$ commitments using $U_{k-i+1}^\dagger \cdots U_k^\dagger$. Initialize an internal register $(\RegM_{k-i}', \RegC_{k-i}')$ to $\ket{0}$ and apply $\SWAP[(\RegM_{k-i}, \RegC_{k-i}),(\RegM_{k-i}', \RegC_{k-i}')]$.
                \item If $i>1$, for $j = k-i, \dots, k-2$ in sequence, apply $\Com$ to $(\RegM_j, \RegC_j, \RegY_j)$ to obtain  $(\RegD_{j+1},\RegC_{j+1})$. Initialize an internal register $(\RegM_{j+1}', \RegC_{j+1}')$ to $\ket{0}$ and then apply $\SWAP[(\RegM_{j+1}, \RegC_{j+1}),(\RegM_{j+1}', \RegC_{j+1}')]$.
                \item If $i>0$, apply $\Com$ to $(\RegM_{k-1}, \RegC_{k-1}, \RegY_{k-1})$ to produce a decommitment on $\RegD_k$ and a commitment on $\RegC_k$. For convenience of notation, let $\RegC_k'=\RegC_k$.
                \item\label[step]{step:reduction-hsub} Send $(\RegD_{k-i}, \RegC_{k-i}')$ to the challenger. According to the challenge bit $b$, the challenger applies $\SWAP^b$ to an auxiliary $\ket{0}$ state and the message $(\RegM_{k-i-1}, \RegC_{k-i-1})$ underlying $(\RegD_{k-i}, \RegC_{k-i}')$. Then the challenger returns $\RegD_{k-i}$ to the reduction.
            \end{enumerate}
            \item The reduction returns the decommitment $\RegD = (\RegD_1, \dots, \RegD_k)$ to the adversary.
        \end{enumerate}
    \end{itemize}
    The view of the adversary in the $b=0$ world of $R_i^{(0-1)}$ is exactly the same as in $H_i^{(0-1)}$, and the view of the adversary in the $b=1$ world is exactly the same as in $H_{i+1}^{(0-1)}$.
    \end{proof}
    
    \begin{claim}
    \label{claim:hybrid12}
    Hybrid $H_1$ is indistinguishable from Hybrid $H_2$.
    \end{claim}
    \begin{proof}
    It is immediate that a QSC is binding to every subset of registers of the message; the same fact for collapse-binding commitments was proven in \cite[Lemma 15]{EC:Unruh16} and the proof for QSCs is nearly identical. In particular, $\Com$ is a binding commitment to $\RegM_i$ for each $i \in \{0,\dots,k-1\}$. Therefore, $H_1$ and $H_2$ are indistinguishable by the following reduction to the security of parallel composition (\cref{theorem:parallel_composition}).
    \begin{itemize}
        \item \textbf{Reduction} $R_1$:
        \begin{enumerate}
            \item The adversary sends $(\RegC, \RegD)$ to $R$.
            \item The reduction does the following:
            \begin{enumerate}
                \item Apply a binary projective measurement to check that $(\RegC, \RegD)$ is valid. If not, the experiment aborts and the reduction makes a random guess $b' \from \{0,1\}$.
                \item Apply $U^\dagger$ to recover the committed message on registers $(\RegM_0, \dots, \RegM_{k-1})$
                \item For each $i = 0,\dots,k-1$:
                \begin{enumerate}
                    \item Initialize an $m$-qubit register $\RegC_i'$ to $\ket{0^m}$.
                    \item Apply $\Com$ to the $m+1$-qubit register $(\RegM_i, \RegC_i', \RegY_i)$, producing a decommitment $\RegD_{i+1}$ and a commitment $\RegC_{i+1}$.
                \end{enumerate}
                \item Send $\{\RegD_i, \RegC_i\}_{i \in [k]}$ to the challenger.
                \item The challenger samples a random $b \from \{0,1\}$, applies
                \[
                \bigotimes_{i\in[k]} U_i (\SWAP[\RegM_i, \RegM_i']^b) U_i^\dagger,
                \]
                and returns $\RegD = (\RegD_1,\dots, \RegD_k)$ to the reduction.
                \item The reduction returns $\RegD$ to the adversary. \qedhere
            \end{enumerate}
        \end{enumerate}
    \end{itemize}
    \end{proof}
    
    \begin{claim}
    \label{claim:hybrid23}
    Hybrid $H_2$ is indistinguishable from Hybrid $H_3$.
    \end{claim}
    \begin{proof}
    The proof is identical to the proof of \cref{claim:hybrid01}, except that we view the commitments as swap-binding commitments to just the registers $\RegC_i$ instead of $(\RegM_i, \RegC_i)$. That is, we only swap $\RegC_i$ instead of $(\RegM_i, \RegC_i)$, and only make use of the binding property of the commitment scheme on the $\RegC_i$ register.
    
    Formally, we again use a sequence of $k+1$ hybrids $H_0^{(3-2)}, \dots, H_k^{(3-2)}$. Hybrid $H_0^{(3-2)}$ is identical to $H_3$, and $H_k^{(3-2)}$ is identical to $H_2$. For $i = 0, \hdots, k$ we define $H_i^{(3-2)}$ by the following behavior for the challenger:
    \noindent
    \begin{itemize}
        \item[] \textbf{Hybrid} $H_i^{(3-2)}$:
        \begin{enumerate}
            \item The adversary sends $(\RegC, \RegD)$.
            \item The challenger does the following:
            \begin{enumerate}
                \item Check that the commitment is valid. If not, the experiment aborts.
                \item Apply $U^\dagger$ to recover the committed message on registers $(\RegM_0, \dots, \RegM_{k-1})$.
                \item Initialize $k$ qubits $\RegM' = (\RegM_0', \dots,\RegM_{k-1}')$ to $\ket{0^k}$ and apply $\SWAP[\RegM,\RegM']$.
                \item Apply $U$.
                \item If $i>0$, open the last $i$ commitments using $U_{k-i+1}^\dagger \cdots U_k^\dagger$, initialize a new register $\RegC_{k-i}'$ to $\ket{0}$, and apply $\SWAP[\RegC_{k-i},\RegC_{k-i}']$.
                \item If $i>1$, for $j = k-i, \dots, k-2$ in sequence, apply $\Com$ to $(\RegM_j,\RegC_j,\RegY_j)$ to produce a decommitment on $\RegD_{j+1}$ and a commitment on $\RegC_{j+1}$, initialize a new register $\RegC_{j+1}'$ to $\ket{0}$, and apply $\SWAP[\RegC_{j+1},\RegC_{j+1}']$.
                \item If $i>0$, apply $\Com$ to $(\RegM_{k-1}, \RegC_{k-1},\RegY_{k-1})$ to produce a decommitment on $\RegD_k$ and a commitment on $\RegC_k$.
            \end{enumerate}
            \item The challenger returns the decommitment $\RegD = (\RegD_1, \dots, \RegD_k)$ to the adversary.
        \end{enumerate}
    \end{itemize}
    Suppose that an adversary $A$ can distinguish hybrids $H_0^{(3-2)}$ and $H_k^{(3-2)}$ with advantage $\varepsilon$. Then for some $i \in \{0,\dots,k-1\}$, the adversary $A$ can distinguish hybrid $H_i^{(3-2)}$ from hybrid $H_{i+1}^{(3-2)}$ with advantage $\varepsilon / k$. We now construct a reduction $R_i^{(3-2)}$ that uses such an adversary to achieve advantage $\varepsilon / k$ in the binding experiment for the commitment scheme $\Com$.
    
    \noindent
    \begin{itemize}
        \item[] \textbf{Reduction} $R_i^{(3-2)}$:
        \begin{enumerate}
            \item The adversary sends $(\RegC,\RegD)$ to $R_i^{(3-2)}$.
            \item The reduction does the following:
            \begin{enumerate}
                \item Check that the commitment is valid. If not, the experiment aborts and the reduction makes a random guess $b' \from \{0,1\}$.
                \item Apply $U^\dagger$ to recover the committed message on registers $(\RegM_0, \dots, \RegM_{k-1})$.
                \item Initialize $k$ qubits $\RegM' = (\RegM_0', \dots,\RegM_{k-1}')$ to $\ket{0^k}$ and apply $\SWAP[\RegM,\RegM']$.
                \item Apply $U$.
                \item If $i>0$, open the last $i$ commitments using $U_{k-i+1}^\dagger \cdots U_k^\dagger$. Initialize an internal register $\RegC_{k-i}'$ to $\ket{0}$ and apply $\SWAP[\RegC_{k-i},\RegC_{k-i}']$.
                \item If $i>1$, for $j = k-i, \dots, k-2$ in sequence, apply $\Com$ to $(\RegM_j, \RegC_j, \RegY_j)$ to obtain  $(\RegD_{j+1},\RegC_{j+1})$. Initialize an internal register $\RegC_{j+1}'$ to $\ket{0}$ and then apply $\SWAP[\RegC_{j+1},\RegC_{j+1}']$.
                \item If $i>0$, apply $\Com$ to $(\RegM_{k-1}, \RegC_{k-1}, \RegY_{k-1})$ to produce a decommitment on $\RegD_k$ and a commitment on $\RegC_k$. For convenience of notation, let $\RegC_k'=\RegC_k$.
                \item Send $(\RegD_{k-i}, \RegC_{k-i}')$ to the challenger. According to the challenge bit $b$, the challenger applies $\SWAP^b$ to an auxiliary $\ket{0}$ state and the register $\RegC_{k-i-1}$ underlying $(\RegD_{k-i}, \RegC_{k-i}')$. Then the challenger returns $\RegD_{k-i}$ to the reduction.
            \end{enumerate}
            \item The reduction returns the decommitment $\RegD = (\RegD_1, \dots, \RegD_k)$ to the adversary.
        \end{enumerate}
    \end{itemize}
    The view of the adversary in the $b=0$ world of $R_i^{(3-2)}$ is exactly the same as in $H_i^{(3-2)}$, and the view of the adversary in the $b=1$ world is exactly the same as in $H_{i+1}^{(3-2)}$.
    \end{proof}

\subsection{Formalizing the folklore: QSCs from QBCs}
\label{sec:folklore-construction}

In this section, we formalize the security of the ``folklore'' construction of \emph{hiding} QSCs: to commit to an $n$-qubit quantum state $\ket{\psi}$, sample two random $n$-bit classical strings $r,s$, and send $X^rZ^s\ket{\psi}$ together with commitments to the classical string $(r,s)$. The opening is just the opening to the commitment to $(r,s)$, which enables the receiver to recover $\ket{\psi}$ by applying $Z^sX^r$. We instantiate the commitment to $(r,s)$ with any quantum-communication commitment to classical messages satisfying hiding and binding.

\paragraph{Defining quantum bit(string) commitments.}

The syntax for a quantum bit(string) commitment (QBC) is identical to our syntax for \emph{non-interactive} quantum state commitments (\cref{def:noninteractive_qsc}). 

Hiding and binding are defined as follows:

\begin{definition}[Hiding for QBCs] \label{def:qbc-hiding}
For a quantum bit(string) commitment $\Com_{\QBC}$, an interactive adversary $A$, a challenge bit $b \in \{0,1\}$, and a security parameter $\lambda$, define a security experiment $\emph{\texttt{HideExpt}}_{\QBC,A,b}(\lambda)$ as follows.

\indent $\emph{\texttt{HideExpt}}_{\QBC,A,b}(\lambda)$:
\begin{enumerate}
    \item\label[step]{step:qbc-hiding-expt-measurement} The adversary $A$ sends a two messages $x_0, x_1 \in \{0,1\}^{|\RegM|}$ to the challenger.
    \item The challenger (acting as the honest sender) initializes register $\RegM$ to the standard basis state $\ket{x_b}$ and register $\RegW$ to $\ket{0}$, applies $\Com_{\QBC}$, and sends register $\RegC$ to the adversary.
    \item The output of the experiment is $b' \gets A$. 
\end{enumerate}

$\Com_{\QBC}$ is computationally (resp. statistically) hiding if there exists a negligible function $\mu(\lambda)$ such that for all polynomial-time (resp. unbounded-time) quantum interactive adversaries $A$,
\[ \Pr_{b \gets \{0,1\}} [\emph{\texttt{HideExpt}}_{\QBC,A,b}(\lambda)=b] \leq \frac{1}{2} + \mu(\lambda). \]
\end{definition}

\begin{definition}[Binding for QBCs \cite{EC:Unruh16}] \label{def:qbc-binding}

For a quantum bit(string) commitment $\Com_{\QBC}$, an interactive adversary $A$, a challenge bit $b \in \{0,1\}$, and a security parameter $\lambda$, define a binding security experiment $\emph{\texttt{ColBindExpt}}_{\QBC,A,b}(\lambda)$ as follows. This definition is an adaptation of Unruh's collapse-binding definition~\cite{EC:Unruh16} to the setting of quantum bit commitments.

\indent $\emph{\texttt{ColBindExpt}}_{\QBC,A,b}(\lambda)$:
\begin{enumerate}
    \item The adversary $A$ (acting as a malicious committer) sends registers $(\RegC, \RegD)$ to the challenger.    
    \item The challenger applies $\Com_{\QBC}^{\dagger}$ and measures $\{\ketbra{0}_{\RegW}, \Id - \ketbra{0}_{\RegW}\}$; if the decommitment is invalid (the measurement rejects), it aborts the experiment (i.e., outputs a random $b'$).
    \item\label[step]{step:qbc-bind-expt-challenge} The challenger does the following: 
    \begin{itemize}
        \item if $b = 0$, the challenger simply applies $\Com_{\QBC}$, and sends the $\RegD$ register to $A$;
        \item if $b = 1$, the challenger measures $\RegM$ in the standard basis. It then applies $\Com_{\QBC}$ and sends the $\RegD$ register to $A$.
    \end{itemize}
    \item The output of the experiment is $b' \gets A$.
\end{enumerate}

$\Com_{\QBC}$ is computationally (resp. statistically) collapse-binding if there exists a negligible function $\mu(\lambda)$ such that for all polynomial-time (resp. unbounded-time) quantum interactive adversaries $A$,
\[ \Pr_{b \gets \{0,1\}} [\emph{\texttt{ColBindExpt}}_{\QBC,A,b}(\lambda)=b] \leq \frac{1}{2} + \mu(\lambda). \]
\end{definition}

\paragraph{Non-interactive form of the folkore construction.} Before proving security of the commitment to an $n$-qubit state, we write the folklore construction in non-interactive form for a commitment to one qubit. Let $\Com_{\QBC}$ be a quantum bitstring commitment to \emph{two} classical bits. That is, it maps two bits in registers $\RegM' \coloneqq (\RegM'_1, \RegM'_2)$ and ancilla register $\RegW'$ to commitment register $\RegC'$ and decommitment register $\RegD'$. Recall that in the folklore construction, the sender commits to two randomly samples bits, $r$ and $s$. To purify this step, we initialize two additional single-qubit registers $\RegK \coloneqq (\RegK_1, \RegK_2)$, and make $(\RegK_1, \RegM'_1)$ and $(\RegK_2, \RegM'_2)$ each hold one EPR pair. Formally, the construction is as follows:

To commit to a qubit in message register $\RegM$, a sender will initialize ancillary registers $\RegW \coloneqq (\RegK_1, \RegK_2, \RegM'_1, \RegM'_2, \RegW')$ to the $\ket{0}$ state. Then the sender applies the following commitment unitary:
\begin{align} \label{eq:qsc_from_qbc}
    \Com_{\QSC} = \Com_{\QBC} (\ctl_{\RegM'_2} \mh Z_{\RegM}) (\ctl_{\RegM'_1} \mh X_{\RegM}) \left( (\ctl_{\RegK_1} \mh X_{\RegM'_1}) H_{K_1}  \otimes (\ctl_{\RegK_2} \mh X_{\RegM'_2}) H_{\RegK_2} \right) \; .
\end{align}
The commitment register is $\RegC \coloneqq (\RegM,\RegC')$ and decommitment register is $ \RegD \coloneqq (\RegK_1, \RegK_2, \RegD')$. Thus, for any quantum message $\ket{\psi}_{\RegM}$,
\[
    \Com_{\QSC} \left( \ket{\psi}_{\RegM} \otimes \ket{0}_{\RegK \RegM' \RegW'} \right) = \frac{1}{2} \sum_{r,s \in \{0,1\}} \ket{r,s}_{\RegK} \otimes (\Com_{\QBC} \ket{r,s}_{\RegM'} \otimes \ket{0}_{\RegW'}) \otimes Z^sX^r \ket{\psi}_{\RegM} .
\]

\begin{theorem} \label{theorem:folklore-construction}
    The quantum state commitment scheme in \cref{eq:qsc_from_qbc}, $\Com_{\QSC}$, is computationally (resp. statistical) hiding and statistically (resp. computational) binding if the underlying quantum bit commitment scheme, $\Com_{\QBC}$, is computationally (resp. statistical) hiding and statistically (resp. computational) binding. 
\end{theorem}
\begin{proof}
    Hiding of $\Com_{\QSC}$ follows by a standard hybrid argument. To simplify notation, for $(r,s) \in \{0,1\}^2$, let $\tau_{r,s}$ denote the density matrix corresponding to the commitment to $(r,s)$ under $\Com_{\QBC}$.
    
    \begin{itemize}
        \item $H_0$: In this hybrid, the adversary $\Adv$ sends a single qubit state $\brho_{\RegM}$ to the challenger. The challenger returns the mixed state
        \[
            \frac{1}{4} \sum_{r,s} \tau_{r,s} \otimes (X^rZ^s) \brho_{\RegM} (Z^sX^r),
        \]
        corresponding to an honest commitment of $\brho_{\RegM}$ under $\Com_{\QSC}$. This hybrid is the view of the adversary in the $b = 0$ case of the $\QSC$ hiding experiment (\cref{def:quantum-hiding}).

        \item $H_1$: This hybrid is the same as $H_0$, except the commitment to $(r,s)$ is replaced by a commitment to $(0,0)$.
        \[
            \tau_{0,0} \otimes \frac{1}{4} \sum_{r,s} (X^rZ^s) \brho_{\RegM} (Z^sX^r).
        \]
        This is indistinguishable from $H_1$ by the hiding security of $\Com_{\QBC}$: an adversary in the hiding experiment can sample random $r,s \gets \{0,1\}^2$, send $(r,s)$ and $(0,0)$ to the challenger, receive the state $\tilde{\tau}$ (which is either $\tau_{r,s}$ or $\tau_{0,0}$), and prepare the state $\frac{1}{4}\sum_{r,s} \tilde{\tau} \otimes (X^rZ^s) \brho_{\RegM} (Z^sX^r)$.
        \item $H_2$: This hybrid is the same as $H_1$, except the challenger uses $\ketbra{0}$ instead of $\brho_{\RegM}$:
        \[
            \tau_{0,0} \otimes \frac{1}{4} \sum_{r,s} (X^rZ^s) \ketbra{0}_{\RegM} (Z^sX^r).
        \]
        This is identical to $H_1$ since the state in both hybrids is $\tau_{0,0} \otimes (\Id/2)$.
        \item $H_3$: This hybrid is the same as $H_2$, except the commitment to $(0,0)$ is replaced by a commitment to $(r,s)$: 
        \[
            \frac{1}{4} \sum_{r,s} \tau_{r,s} \otimes (X^rZ^s) \ketbra{0}_{\RegM} (Z^sX^r).
        \]
        This follows by the hiding property of $\Com_{\QBC}$. This completes the proof of hiding for $\Com_{\QSC}$ since $H_3$ is the view of the adversary in the $b = 1$ case of the $\Com_{\QSC}$ hiding experiment (\cref{def:quantum-hiding}).
    \end{itemize}

    We now show that $\Com_{\QSC}$ is swap binding (\cref{def:swap-binding}) assuming $\Com_{\QBC}$ is binding (\cref{def:qbc-binding}). 
    
    \begin{itemize}
        \item Hybrid $H_0$: 
        \begin{enumerate}
            \item The adversary $\Adv$ sends a quantum state on registers $\RegC, \RegD$ to the challenger. 
            \item The challenger then does the following:
            \begin{enumerate}
                \item Apply $\Com_{\QSC}^{\dagger}$ to $(\RegC, \RegD)$.
                \item Measure $\{\ketbra{0}_{\RegW}, \Id - \ketbra{0}_\RegW\}$. If the measurement rejects, abort.
                \item Apply $\Com_{\QSC}$. 
                \item Send the decommitment register $\RegD$ to $\Adv$.
            \end{enumerate}
        \end{enumerate}
        The view of $\Adv$ corresponds to the view of the adversary in the $b=0$ world of the swap-binding experiment for $\Com_{\QSC}$ (\cref{def:swap-binding}). 
        \item Hybrid $H_1$: 
        \begin{enumerate}
            \item The adversary $\Adv$ sends a quantum state on registers $\RegC, \RegD$ to the challenger. 
            \item The challenger then does the following:
            \begin{enumerate}
                \item\label[step]{step:h1-open} Apply $\Com_{\QSC}^{\dagger}$ to $(\RegC, \RegD)$.
                \item\label[step]{step:h1-measurement} Measure $\{\ketbra{0}_{\RegW}, \Id - \ketbra{0}_\RegW\}$. If the measurement rejects, abort.
                \item\label[step]{step:h1-recommit} Apply $\Com_{\QSC}$. 
                {\color{red}
                \item Apply $\Com_{\QBC}^\dagger$ to reveal the two-bit committed classical message on register $\RegM'$.
                \item Measure $\RegM'$ in the standard basis.
                \item Apply $\Com_{\QBC}$.}
                \item Send the decommitment register $\RegD$ to $\Adv$.
            \end{enumerate}
        \end{enumerate}
        The measurement in \cref{step:h1-measurement} ensures that after applying $\Com_{\QSC}$ in \cref{step:h1-recommit}, the resulting state includes a valid commitment-decommitment pair for the message in $\RegM'$ under $\Com_{\QBC}$. By the binding security of $\Com_{\QBC}$, the act of measuring the contents of $\RegM'$ in the standard basis is undetectable without access to the commitment $\RegC'$. Thus, hybrid $H_1$ is indistinguishable from $H_0$; otherwise, an adversary in the $\Com_{\QBC}$ binding experiment receives $(\RegC, \RegD)$ from $\Adv$, performs \cref{step:h1-open,step:h1-measurement,step:h1-recommit}, sends $(\RegC', \RegD')$ to the challenger, receives $\RegD'$ from the challenger, sends $\RegD$ to $\Adv$, and uses $\Adv$'s output to win the $\Com_{\QBC}$ binding experiment. 
        \item Hybrid $H_2$: 
        \begin{enumerate}
            \item The adversary $\Adv$ sends registers $(\RegC, \RegD)$ to the challenger. 
            \item The challenger then does the following:
            \begin{enumerate}
                \item Apply $\Com_{\QSC}^{\dagger}$ to $(\RegC, \RegD)$ to obtain $(\RegM,\RegW)$.
                \item\label[step]{step:h2-measurement} Measure $\{\ketbra{0}_{\RegW}, \Id - \ketbra{0}_\RegW\}$. If the measurement rejects, abort.
                \item {\color{red} Replace the contents of $\RegM$ with the state $\ketbra{0}_{\RegM}$}. 
                \item Apply $\Com_{\QSC}$. 
                \item Apply $\Com_{\QBC}^\dagger$ to reveal the two-bit committed classical message on register $\RegM'$.
                \item Measure $\RegM'$ in the standard basis.
                \item Apply $\Com_{\QBC}$.
                \item Send the decommitment register $\RegD$ to $\Adv$.
            \end{enumerate}
        \end{enumerate}
        For $(r,s) \in \{0,1\}^2$, let $\gamma_{r,s}$ denote the density matrix corresponding to the decommitment to $(r,s)$ under the quantum commitment to classical messages $\Com_{\QBC}$. By a straightforward computation, regardless of the the contents of $\RegM$ after \cref{step:h2-measurement}, the state that the challenger sends $\Adv$ is 
        \[
            \frac{1}{4} \sum_{r,s} \ketbra{r,s}_{\RegK} \otimes (\gamma_{r,s})_{\RegD'} .
        \]
        Thus $H_2$ is perfectly indistinguishable from $H_1$. 
        \item Hybrid $H_3$: 
        \begin{enumerate}
            \item The adversary $\Adv$ sends registers $(\RegC, \RegD)$ to the challenger. 
            \item The challenger then does the following:
            \begin{enumerate}
                \item Apply $\Com_{\QSC}^{\dagger}$ to the $(\RegC, \RegD)$ registers.
                \item Measure $\{\ketbra{0}_{\RegW}, \Id - \ketbra{0}_\RegW\}$. If the measurement rejects, abort.
                \item Replace the contents of $\RegM$ with the state $\ketbra{0}_{\RegM}$. 
                \item Apply $\Com_{\QSC}$. 
                {\color{red}
                \item[] \sout{Apply $\Com_{\QBC}^\dagger$ to reveal the two-bit committed classical message on register $\RegM'$.}
                \item[] \sout{Measure $\RegM'$ in the standard basis.}
                \item[] \sout{Apply $\Com_{\QBC}$.}}
                \item[(e)] Send the decommitment register $\RegD$ to $\Adv$.
            \end{enumerate}
        \end{enumerate}
        The difference between hybrids $H_2$ and $H_3$ is that challenger measures $\RegM'$ in the former but not in the latter. Again, by the binding security of $\Com_{\QBC}$, the two hybrids are indistinguishable. This completes the proof of binding for $\Com_{\QSC}$ because $H_3$ is the view of the adversary in the $b=1$ world of the binding experiment (\cref{def:swap-binding}). 
    \end{itemize}
\end{proof}

The proof of security above shows that the folklore construction applied to one qubit is a hiding and binding QSC. By appealing to the security of the parallel composition of $n$ such commitment schemes using \cref{theorem:parallel_composition}, we get a full QSC for $n$-qubit quantum messages. 

\begin{corollary}
    Let $\Com_{\QSC}$ be the QSC defined by \cref{eq:qsc_from_qbc}. Then for any $n=\poly(\lambda)$, the $n$-fold parallel composition $\bigotimes_{i=1}^n \Com_{\QSC}$ is a hiding and binding QSC if the underlying quantum commitment to classical messages $\Com_{\QBC}$ is hiding and binding.  
\end{corollary}

\newpage

\section{The oracle security lemma}
\label{sec:oracle-security}

In this section we prove that swap binding is equivalent to a seemingly stronger definition that we call oracle swap binding. As discussed in \Cref{subsubsec:to-rewinding-difficulties}, oracle swap binding plays a crucial role in our rewinding-based security proofs (\cref{sec:rewinding}).

% We refer the reader to \Cref{subsubsec:to-rewinding-difficulties} to see why oracle swap binding is a necessary component in proofs of security using swap binding commitment schemes. For example, this equivalence will be a crucial part of our rewinding-based security proofs in~\cref{sec:rewinding}.

% Roughly speaking, a commitment is oracle swap binding if the swap binding game remains hard even when the adversary is given an oracle that can perform \emph{arbitrary efficient operations on the originally committed message}. 

Recall that in the swap-binding security experiment, the challenger's bit $b$ determines whether the adversary's originally committed message is left on the $\RegM$ register (when $b = 0$) or swapped into the challenger's $\RegM'$ register (when $b = 1$). Of course, the standard swap-binding adversary is unable to access either of these registers: the state on $\RegM$ can only be revealed by applying $\Com^\dagger$ to $(\RegC,\RegD)$ (which the adversary cannot do since the challenger does not return the $\RegC$ register) and the contents of the $\RegM'$ register remain entirely out of the adversary's view. Roughly speaking, the oracle swap binding experiment gives the adversary an oracle that can perform \emph{arbitrary efficient operations on the originally committed message}, i.e., $\RegM$ if $b = 0$ or $\RegM'$ if $b = 1$.

To state the definition formally, we use the following notation:
\begin{itemize}
    \item For any operator $O$ and commitment unitary $\Com$, we write $\widehat{O} \coloneqq (\Com) O (\Com^{\dagger})$. 
    
    % operations $\RegM$ us to easily specify operations acting on $\RegM$  particular, applying $\widehat{O}$ to $(\RegC,\RegD)$ is equivalent to applying $\cO$ to apply an operation $\RegO$ on $\RegM$ to a commitment-decommitment pair, we msut perform $\widehat{O}$. is useful for appyling operations on the $\RegM$ register operations to be performed on $\RegM$
    \item Let $\Pi \coloneqq \Com (\Id_\RegM \otimes \ketbra{0}_\RegW) \Com^{\dagger}$ denote the projection onto valid commitment-decommitment states.
    \item Recall that for a pair of equal-size registers $\RegM, \RegM'$, the unitary $\SWAP[\RegM,\RegM']$ maps $\ket{\psi}_{\RegM}\ket{\phi}_{\RegM'}$ to $\ket{\phi}_{\RegM}\ket{\psi}_{\RegM'}$. 
\end{itemize}

In the oracle swap binding experiment, the adversary specifies an arbitrary unitary $\cO$, which acts on $\RegM$ and any number of additional ancilla qubits $\RegR$. After the adversary sends $(\RegC,\RegD)$ to its challenger and gets back a state on $\RegD$, it is allowed to perform the following oracle call as many times as desired:
% a state $\RegD$ to the adversary,  

% between oracle-swap binding experiment and the standard swap-binding experiment i fr differs  the standard swap-binding experiment, the oracle swap-binding adversary submits $(\RegC,\RegD)$ to the challenger and then gets back a state on $\RegD$. The oracle swap-binding adversary specifies an arbitrary efficient unitary $\cO$, which acts on $\RegM$ and any number of additional ancilla qubits $\RegR$. The adversary is given the ability to interact with the following oracle interface:
% who returns $\RegD$. But now, before the adversary specifies n the adversary   After the challenger returns $\RegD$ to the adversary, the adversary can make queries to the oracle as follows:
\begin{enumerate}
    \item The adversary sends registers $(\RegR,\RegD)$ to the oracle.
    \item If $b = 0$, the oracle applies the unitary $\widehat{\cO} \Pi + (\Id - \Pi)$ to the registers $(\RegC,\RegD,\RegR)$. If $b = 1$, the oracle applies the unitary $\widehat{\SWAP[\RegM,\RegM']} \cdot \widehat{\cO} \cdot \widehat{\SWAP[\RegM,\RegM']} \Pi + (\Id - \Pi)$.
    \item The oracle returns $(\RegR,\RegD)$ to the adversary.
\end{enumerate}

We can write this oracle more compactly as
\[
G_b \coloneqq  (\widehat{\SWAP[\RegM, \RegM']})^b \cdot \widehat{ \cO} \cdot  (\widehat{\SWAP[\RegM, \RegM']})^b \Pi + (\Id-\Pi).
\]
When $(\RegC,\RegD)$ contains a valid commitment (i.e., the state is in the image of $\Pi$), the oracle behaves as expected: in the $\Com^\dagger$ basis, it applies $\cO$ to $(\RegM,\RegR)$ if $b = 0$, or to $(\RegM',\RegR)$ if $b = 1$. To ensure that the oracle does not enable trivial attacks, we will restrict the oracle's behavior to be independent of $b$ whenever $(\RegC,\RegD)$ is an invalid commitment. In particular, we require that the oracle behave as identity on any state in $(\Id - \Pi)$.

\begin{definition}[Oracle swap-binding for quantum state commitments] \label{def:oracle-binding}
% Consider a QSC specified by a unitary $\Com$. An efficient adversary for oracle swap binding is initialized with a state on registers $(\RegC,\RegD,\RegR)$ where $\RegR$ denotes the adversary's internal register. The adversary additionally specifies an efficient unitary $\cO$ that acts on registers $(\RegR,\RegM)$, where $\RegM$ is the message register obtained by applying $\Com^\dagger$ to $(\RegC,\RegD)$. The adversary plays the standard swap-binding security experiment, but is augmented with the ability to make oracle calls to the unitary $G_b$ defined as
% \[
% G_b \coloneqq  (\widehat{\SWAP[\RegM, \RegM']})^b \cdot \widehat{ \cO} \cdot  (\widehat{\SWAP[\RegM, \RegM']})^b \Pi + (\Id-\Pi),
% \]
% where $b$ is the challenger's bit.

$\Com$ is computationally (resp. statistically) oracle swap-binding if there exists a negligible function $\mu(\lambda)$ such that for all polynomial-time (resp. unbounded-time) oracle swap-binding adversaries specified by an efficient (resp. inefficient) interactive adversary $A$ and any efficient (resp. inefficient) unitary $\cO$, 
\[
\Pr_{b \gets \{0,1\}} [\emph{\texttt{SwapBindExpt}}_{\QSC,A^{G_b},b}(\lambda)=b] \leq \frac{1}{2} + \mu(\lambda).
\]

% consists of a quantum interactive algorithm specified by a quantum interactive adversary $A$ and unitary operation $\mathcal{O}$ acting on $\RegR, \RegM$.

% Let $\mathcal{O}$ be a unitary operation acting on $\RegR, \RegM$, where $\RegR$ is some register held by the adversary and $\RegM$ is the message register for $\Com$. Let $\RegM'$ be a register held by the challenger and initialized to $\ket{0}$. We define
% \[
% G_b \coloneqq  (\widehat{\SWAP[\RegM, \RegM']})^b \cdot \widehat{ \cO} \cdot  (\widehat{\SWAP[\RegM, \RegM']})^b \Pi + (\Id-\Pi).
% \]
% For challenge bit $b \in \{0,1\}$ and security parameter $\lambda$, we say that $\Com$ is computationally (resp. statistically) oracle swap-binding if there exists a negligible function $\mu(\lambda)$ such that for all polynomial-time (resp. unbounded-time) quantum interactive adversaries $A$,
% \[
% \Pr_{b \gets \{0,1\}} [\emph{\texttt{SwapBindExpt}}_{\QSC,A^{G_b},b}(\lambda)=b] \leq \frac{1}{2} + \mu(\lambda).
% \]
\end{definition}

\begin{theorem} \label{theorem:oracle-binding-security}
If $\Com$ is computationally (resp. statistically) swap-binding (\cref{def:swap-binding}), then $\Com$ is computationally (resp. statistically) oracle swap-binding (\cref{def:oracle-binding}).
\end{theorem}

The rest of this section is devoted to proving~\cref{theorem:oracle-binding-security}. 

% In fact, we will prove a slightly more general lemma that we call the \emph{admissible oracle lemma}, which we show implies~\cref{theorem:oracle-binding-security} in~\cref{subsec:swap-binding-implies-oracle}.

% Roughly speaking, the admissible oracle lemma shows that for a 

% The only difference between oracle swap-binding (\cref{def:oracle-binding}) and regular swap-binding (\cref{def:swap-binding}) is that in oracle swap-binding, the adversary is allowed to query the oracle $G_b$. We will prove that swap binding schemes is equivalent to oracle swap-binding. Our main technique in proving this theorem is what we call the \emph{admissible oracle lemma}.

\subsection{The admissible oracle lemma}

Instead of proving~\cref{theorem:oracle-binding-security} directly, we will consider a more abstract distinguishing task that we call the $(W,\Pi)$-distinguishing game, parameterized by a binary observable $W$ and a projection $\Pi$ that commutes with $W$. We state and prove a lemma we call the \emph{admissible oracle lemma} (\cref{lemma:admissible-oracles}); in~\cref{subsec:swap-binding-implies-oracle}, we show how~\cref{theorem:oracle-binding-security} follows from a special case of the admissible oracle lemma.

% \fermi{add better transition ``We will prove that swap binding implies oracle swap binding by proving a more general result about an abstract distinguishing task (namely, one that generalizes the swap-binding game) that we call the $(W,\Pi)$-distinguishing game.''}
% We now define an abstract distinguishing task called the ``$(W,\Pi)$-distinguishing game'' parameterized by a binary observable $W$ and a projection $\Pi$ that commutes with $W$. In this game, the adversary provides a state in the image of $\Pi$, and must determine whether a challenger applies the operator $W$ to its state. 

\begin{definition}[$(W,\Pi)$-distinguishing game]
    Let $(\RegA, \RegB)$ be two quantum registers. Let $W$ be a binary observable and $\Pi$ be a projector on $(\RegA, \RegB)$ such that $\Pi$ commutes with $W$. Consider the following distinguishing game:
    \begin{enumerate}
        \item The adversary sends a quantum state on registers $(\RegA, \RegB)$ to the challenger.
        \item\label[step]{step:dist-chal} The challenger chooses a random bit $b \leftarrow \{0,1\}$. Next, it measures measures $\{\Pi,\Id-\Pi\}$; if the measurement rejects, abort and output a random bit $b' \from \{0,1\}$. Otherwise, the challenger applies $W^b$ to $(\RegA, \RegB)$, and returns $\RegB$ to the adversary.
        \item\label[step]{step:dist-adv} The adversary outputs a guess $b'$.
    \end{enumerate}
    We define the distinguishing advantage of the adversary to be $\abs{\Pr[b'=b] - 1/2}$.
\end{definition}

% When it is clear from context, we will drop the $\Pi$ and simply refer to this as the $(W,\Pi)$-distinguishing game. 

The main result of the section is that, if the $(W,\Pi)$-distinguishing game is hard, i.e., the best possible distinguishing advantage is negligible, then it remains hard when the adversary is given oracle access to any \emph{admissible} unitary $G$, defined as follows.

% certain operations $G$ at the end of the game. Importantly, $G$ may act on the register $\RegA$, which is completely inaccessible in the standard $(W,\Pi)$-distinguishing game. We define our class of \emph{admissible} operations below.

\begin{definition}
We say that a unitary $G$ is \emph{admissible} if:
\begin{itemize}
    \item $G$ commutes with both $W$ and $\Pi$, and
    \item $G$ acts identically on $\Id-\Pi$, i.e., $G(\Id-\Pi) = \Id - \Pi$.
\end{itemize}
\end{definition}

Informally, the admissible oracle lemma (\cref{lemma:admissible-oracles}) says the following. Suppose that a quantum adversary with an initial state $\ket{\psi}$ wins the $(W,\Pi)$-distinguishing game with advantage $\varepsilon$, making $t$ oracle calls to $G$. Then there exists a quantum adversary that wins the $(W,\Pi)$-distinguishing game with advantage $\varepsilon^2/8t^2$, but does not make any calls to $G$. The latter adversary requires an initial state $\brho$ which can be prepared using a circuit $\Prep^G$ that runs $G$ as a subroutine. We emphasize that $\Prep^G$ does not require special oracle access and can be implemented directly, since it is run before the adversary sends any registers to the challenger.

\begin{lemma}[Admissible oracle lemma]
\label{lemma:admissible-oracles}
    Let $G$ be an admissible unitary. Suppose that for some quantum state $\ket{\psi}_{\RegA \RegB \RegR}$, there is a quantum algorithm $D^G$ (acting only on registers $(\RegB, \RegR)$) with oracle access to $G$ that achieves a distinguishing advantage $\varepsilon$ in the $(W,\Pi)$-distinguishing game with initial state $\ket{\psi}$. Then there exists another quantum state $\brho$ and an algorithm $E$ making no queries to $G$ that wins the $(W,\Pi)$-distinguishing game with advantage $\frac{\varepsilon^2}{8t^2}$ using initial state $\brho$, where $t$ is an upper bound on the number of queries $D$ makes to $G$. Moreover,
    \begin{itemize}
        \item if $D^G$ is implemented with a size-$T$ oracle circuit (counting oracle queries as size $1$), then $\brho \leftarrow \Prep^{G, M} (\ket{\psi})$, where $\mathsf{Prep}^{G, M}$ is a $\poly(T)$ size oracle circuit with (controlled) queries to $G$, $G^{\dagger}$, and $M \coloneqq (W^+ \Pi) \otimes X + (\Id - W^+ \Pi) \otimes \Id$, and
        \item $E$ is $\poly(T)$ size. 
    \end{itemize}
\end{lemma}

We will prove \cref{lemma:admissible-oracles} in two steps:
\begin{enumerate}
    \item In \cref{lemma:duality}, we show that the distinguishing task is equivalent to the following \emph{mapping task}: given any state in $\image(W^+\Pi)$, output any state in $\image(W^-\Pi)$.
    \item In \cref{lemma:main}, we show that for the \emph{mapping task}, admissible oracles only provide limited help.
\end{enumerate}
For a distinguisher $D$ implemented as a binary projective measurement $D = \{\Pi_D, \Id - \Pi_D\}$, the distinguishing advantage on (potentially sub-normalized) $\ket{\psi}$ can be written as
\[
\frac{1}{2} \abs{ \bra{\psi}\Pi \Pi_D \Pi \ket{\psi} - \bra{\psi} \Pi W \Pi_D W \Pi \ket{\psi} }.
\]
For a unitary $U$ and (potentially sub-normalized) state $\ket{\psi}$, we define the \emph{mapping advantage} of $U$ on $\ket{\psi}$ to be
\[ \norm{(W^- \Pi) U (W^+ \Pi)\ket{\psi}}^2\]
The following lemma states that the mapping and distinguishing tasks are equivalent in a setting where the adversary supplies the initial state (up to a quadratic loss in the advantage).

\begin{lemma}
\label{lemma:duality}
Let $W$ be a binary observable on a register $\RegH$, let $\Pi$ be a projection on $\RegH$ that commutes with $W$, and let $\RegB$ be a fresh single-qubit register. Fix a state $\ket{\psi}$ on $\RegH$ and let $\ket{\psi_+} \coloneqq W^+ \Pi\ket{\psi}$. Then:
\begin{enumerate}[(i)]
    \item\label[item]{item:map-to-dist} If $U$ has mapping advantage $\varepsilon$ on $\ket{\psi}$, then $\{\widetilde{\Pi}_D, \Id - \widetilde{\Pi}_D\}$ has distinguishing advantage $\varepsilon/2$ on $ \mathsf{ctl}_{\RegB} U \ket{+}_\RegB \ket{\psi_+}$, where $\widetilde{\Pi}_D \coloneqq  \mathsf{ctl}_{\RegB} U \ketbra{+}_\RegB \mathsf{ctl}_{\RegB} U^\dagger$.
    \item\label[item]{item:dist-to-map} If $\Pi_D$ has distinguishing advantage $\varepsilon$ on $\ket{\psi}$, then the unitary $\widetilde{U} = \Id - 2 \Pi_D$ has mapping advantage at least $\varepsilon^2$ on $\ket{\psi}$.
\end{enumerate}
\end{lemma}
\begin{proof}[Proof of \cref{item:map-to-dist}]
Let $\ket{\phi} = \mathsf{ctl}_{\RegB} U \ket{+}_\RegB \ket{\psi_+}$. The proof follows from directly computing the distinguishing advantage
\begin{align*}
    \frac{1}{2} \abs{ \bra{\phi} \Pi \widetilde{\Pi}_D \Pi \ket{\phi} - \bra{\phi} \Pi W \widetilde{\Pi}_D W \Pi \ket{\phi} }.
\end{align*}
For $b \in \{0,1\}$, we compute
\begin{align*}
    \bra{\phi}\Pi W^b \widetilde{\Pi}_D W^b \Pi \ket{\phi} &= \norm{\widetilde{\Pi}_D W^b \Pi \mathsf{ctl}_{\RegB}U \ket{+}_\RegB \ket{\psi_{+}} }^2 \\
    &= \frac{1}{2} \norm{\ketbra{+}_\RegB \left(\ket{0}_\RegB W^b \ket{\psi_+} + \ket{1}_\RegB U^\dagger W^b \Pi U \ket{\psi_+}\right)}^2 \\
    &= \frac{1}{4} \norm{\ket{+}_\RegB \left(\ket{\psi_+} + U^\dagger W^b \Pi U \ket{\psi_+}\right)}^2 \\
    &= \frac{1}{4} \left( \braket{\psi_+} + \bra{\psi_+}U^\dagger \Pi U \ket{\psi_+} + 2 \Re \bra{\psi_+} U^\dagger W^b \Pi U \ket{\psi_+} \right).
\end{align*}

Then taking the difference,
\begin{align*}
    & \bra{\phi}\Pi  \widetilde{\Pi}_D \Pi \ket{\phi} - \bra{\phi}\Pi W \widetilde{\Pi}_D W \Pi\ket{\phi} \\
    &= \frac{1}{2} \left(\Re \bra{\psi_+} U^\dagger \Pi U \ket{\psi_+} - \Re \bra{\psi_+} U^\dagger W \Pi U \ket{\psi_+}\right) \\
    &=  \frac{1}{2}\Re \bra{\psi_+} U^\dagger (I-W) \Pi U \ket{\psi_+} \\
    &= \bra{\psi_+} U^\dagger W^- \Pi U \ket{\psi_+} \\
    &= \varepsilon. \qedhere
\end{align*}
\end{proof}
\begin{proof}[Proof of \cref{item:dist-to-map}]
We compute 
\begin{align*}
    \norm{W^- \Pi \widetilde{U} \ket{\psi_+}}^2 &= \norm{\Pi\left(\frac{I-W}{2}\right)\widetilde{U}\left(\frac{I+W}{2}\right)\Pi \ket{\psi}}^2\\
    &\geq \abs{\bra{\psi}\Pi\left(\frac{I-W}{2}\right)\widetilde{U}\left(\frac{I+W}{2}\right)\Pi\ket{\psi}}^2 \\
    &= \frac{1}{4} \abs{(\bra{\psi} - \bra{\psi} W) \widetilde{U} (\ket{\psi} + W \ket{\psi})}^2 \\ 
    &= \frac{1}{4} \abs{\bra{\psi} \widetilde{U} \ket{\psi} - \bra{\psi} W\widetilde{U}W \ket{\psi} + 2i \cdot \Im(\bra{\psi} \widetilde{U}W \ket{\psi})}^2 \\
    &\geq \frac{1}{4} \abs{\bra{\psi} \widetilde{U} \ket{\psi} - \bra{\psi} W\widetilde{U}W \ket{\psi}}^2 \\
    &= \frac{1}{2} \cdot \abs{ \bra{\phi}\Pi \widetilde{\Pi}_D \Pi\ket{\phi} - \bra{\phi} \Pi W \widetilde{\Pi}_D W \Pi \ket{\phi} }^2 \\
    &= 2 \varepsilon^2. \qedhere
\end{align*}
\end{proof}

\begin{remark}
In \cite{AAS20}, a similar statement is proven about the equivalence between distinguishing and \emph{swapping} for the task of distinguishing between two \emph{fixed} states. They also showed that distinguishing and mapping are not in general equivalent for the fixed-state setting. \cref{lemma:duality} implies that distinguishing, swapping, and mapping are all equivalent when the task is to \emph{provide} a state and detect the application of a fixed operation.
\end{remark}

We now prove that an admissible unitary $G$ cannot help an adversary with the mapping task. Without loss of generality we assume that the adversary repeatedly applies a single unitary $U$, interleaved with $G$.\footnote{Any adversary can be converted into this form by introducing an additional clock register.} \Cref{lemma:main} bounds the mapping advantage of $(UG)^t$ in terms of the mapping advantage of \emph{just applying $U^q$} for any $q \le t$, with an initial state that can be prepared with access to a unitary very similar to $G$.

\begin{lemma}
\label{lemma:main}
    Fix a Hilbert space $\cH$. Let $U$ be a unitary, $\ket{\psi}$ be a state, and $\Pi_0,\Pi_1$ be a pair of orthogonal projectors, all on $\cH$. Let $\Pi = \Pi_0 + \Pi_1$. Let $G$ be a unitary on $\cH$ of the form $G = \Pi_0 G_0 + \Pi_1 G_1 + (\Id- \Pi)$ where $G_0,G_1$ are unitaries that commute with $\Pi_0,\Pi_1$ respectively. Define $\widetilde{G}_0 = \Pi_0 G_0 + (\Id-\Pi_0)$ and
    \[
    \varepsilon(t) \coloneqq \max_{q,r,s \le t} \norm{ \Pi_1 U^q \Pi_0 (\widetilde{G}_0)^r (U\widetilde{G}_0)^s \Pi_0 \ket{\psi}}
    \]
    for any integer $t \geq 0$. Then for all integers $t \geq 0$,
    \[
        \norm{\Pi_1 (U G)^t \Pi_0 \ket{\psi}} \leq 4 t^2 \cdot \varepsilon(t).
    \]
\end{lemma}

We divide the proof of~\cref{lemma:main} into~\cref{claim:main-expansion} and~\cref{claim:ignore-P1}.

\begin{claim}
\label{claim:main-expansion}
\[
\norm{\Pi_1 (U \widetilde{G}_0)^t \Pi_0 \ket{\psi}} \leq 2 t \cdot \varepsilon(t)
\]
\end{claim}
\begin{proof}
    Let $H \coloneqq \Pi_0 G_0 + \Pi_1 G_1 - \Id$ and write $\widetilde{G}_0 = \Pi_0 G_0 + \Pi_1 + (\Id - \Pi) = \Id + \Pi_0 H$. 
    We can then expand the state $\Pi_1 (U \widetilde{G}_0)^t \Pi_0 \ket{\psi}$ as
    \begin{align*}
        \Pi_1 (U \widetilde{G}_0)^t \Pi_0 \ket{\psi} =  \Pi_1 (U + U \Pi_0 H)^t \Pi_0 \ket{\psi} = \sum_{r=1}^{t} \Pi_1 U^r F_{t-r} \Pi_0\ket{\psi},
    \end{align*}
    where $F_0 \coloneqq \Id$, and $F_i \coloneqq U \Pi_0 H (U + U \Pi_0 H)^{i-1}$ for $i \geq 1$. The second equality above uses the fact that $U^r F_{t-r} $ is the sum of the terms in the binomial expansion of $(U + U \Pi_0 H)^t$ that, when going from left to right, consist of exactly $r$ $U$'s before the first $U \Pi_0 H$.
    
    By the triangle inequality, it suffices to show that for each $r \in [t]$,
    \begin{align*}
        \norm{\Pi_1 U^r F_{t-r} \Pi_0 \ket{\psi}} \leq 2\varepsilon(t).
    \end{align*}
         The $r=t$ case is immediate from assumption. For the $r \leq t-1$ case, rewrite the definition of  $F_{t-r}$ using $\Pi_0 H = \Pi_0G_0 - \Pi_0$ and $\widetilde{G}_0 = \Id + \Pi_0 H$:
    \begin{align*}
        F_{t-r} = U (\Pi_0G_0 - \Pi_0) (U\widetilde{G}_0)^{t-r-1}.
    \end{align*}
    Plugging in this expression for $F_{t-r}$ and invoking the triangle inequality yields
    \begin{align*}
        & \norm{\Pi_1 U^r F_{t-r} \Pi_0 \ket{\psi}}\\
        &= \norm{\Pi_1 U^{r+1} \Pi_0 G_0 (U\widetilde{G}_0)^{t-r-1}\Pi_0 \ket{\psi}-\Pi_1 U^{r+1} \Pi_0 (U\widetilde{G}_0)^{t-r-1} \Pi_0 \ket{\psi}}\\
        &\leq \norm{\Pi_1 U^{r+1} \Pi_0 G_0 (U\widetilde{G}_0)^{t-r-1}\Pi_0 \ket{\psi}} + \norm{\Pi_1 U^{r+1} \Pi_0 (U\widetilde{G}_0)^{t-r-1} \Pi_0 \ket{\psi}}.
    \end{align*}
    By our initial assumptions, this implies $\norm{\Pi_1 U^r F_{t-r} \Pi_0 \ket{\psi}} \leq 2\varepsilon(t)$.
\end{proof}

\begin{claim}
\label{claim:ignore-P1}
\[
\norm{ (U G)^t \Pi_0 \ket{\psi} - (U \widetilde{G}_0)^t \Pi_0 \ket{\psi}} \leq 2 t^2 \cdot \varepsilon(t)
\]
\end{claim}
\begin{proof}
We prove this claim by induction on $t$. The base case $t=0$ is immediate. For $t \geq 1$, we have by the inductive hypothesis that
\begin{align}
\label{eq:first-inequality}
    \norm{(UG)^t \Pi_0 \ket{\psi} -U G (U \widetilde{G}_0)^{t-1} \Pi_0 \ket{\psi}} \leq 2(t-1)^2\cdot\varepsilon(t-1).
\end{align}
Using $G = \Id + HP$, we can write $ U G (U \widetilde{G}_0)^{t-1} \Pi_0 \ket{\psi}$ and $(U \widetilde{G}_0)^t \Pi_0 \ket{\psi}$ as follows:
\begin{align}
    \label{eq-subtract1} U G (U \widetilde{G}_0)^{t-1} \Pi_0 \ket{\psi} &= U (U \widetilde{G}_0)^{t-1} \Pi_0 \ket{\psi} + U H \Pi (U \widetilde{G}_0)^{t-1} \Pi_0 \ket{\psi}\\
     \label{eq-subtract2} (U \widetilde{G}_0)^t \Pi_0 \ket{\psi} &= U (U \widetilde{G}_0)^{t-1} \Pi_0 \ket{\psi} + U H \Pi_0 (U \widetilde{G}_0)^{t-1} \Pi_0 \ket{\psi}.
\end{align}
By subtracting \cref{eq-subtract2} from \cref{eq-subtract1} and taking the norm, we obtain
\begin{align*}
    \norm{ U G (U \widetilde{G}_0)^{t-1} \Pi_0 \ket{\psi} -   (U \widetilde{G}_0)^t \Pi_0 \ket{\psi}} &= \norm{U H \Pi (U \widetilde{G}_0)^{t-1} \Pi_0 \ket{\psi} -  U H \Pi_0 (U \widetilde{G}_0)^{t-1} \Pi_0 \ket{\psi}}\\
    &\leq \norm{UH}_{\mathrm{op}}\norm{ \Pi (U \widetilde{G}_0)^{t-1} \Pi_0 \ket{\psi} - \Pi_0 (U \widetilde{G}_0)^{t-1} \Pi_0 \ket{\psi}}.
\end{align*}
By \cref{claim:main-expansion} (and writing $\Pi_1 = \Pi - \Pi_0$)
\begin{align*}
    \norm{\Pi (U \widetilde{G}_0)^{t-1} \Pi_0 \ket{\psi} - \Pi_0 (U \widetilde{G}_0)^{t-1} \Pi_0 \ket{\psi}} \leq 2 (t-1) \cdot \varepsilon(t-1).
\end{align*}
Since $\norm{UH}_{\mathrm{op}} = \norm{H}_{\mathrm{op}} \leq 2$, it follows that
\begin{align}
\label{eq:second-inequality}
    \norm{ U G (U \widetilde{G}_0)^{t-1} \Pi_0 \ket{\psi} - (U \widetilde{G}_0)^t \Pi_0 \ket{\psi}} \leq 4 (t-1) \cdot \varepsilon(t-1).
\end{align}

Combining \cref{eq:first-inequality,eq:second-inequality}, we obtain 
\begin{align*}
    \norm{ (U G)^t \Pi_0 \ket{\psi} - (U \widetilde{G}_0)^t \Pi_0 \ket{\psi}} &\leq 2(t-1)^2 \cdot \varepsilon(t-1) + 4 (t-1) \cdot \varepsilon(t-1) \\
    &\leq 2 t^2 \cdot \varepsilon(t). \qedhere
\end{align*}
\end{proof}

\begin{proof}[Proof of \cref{lemma:main}]
    % Define $\widetilde{G}_0 \coloneqq \Pi_0 G_0 + \Pi_1 + (\Id - \Pi)$, which is $G$ but with the action on $\Pi_1$ fixed to $\Id$.
    
    By \cref{claim:ignore-P1} and \cref{claim:main-expansion}, we have
    \begin{align*}
        \norm{\Pi_1 (U G)^t \Pi_0 \ket{\psi}} &\le \norm{\Pi_1 (U \widetilde{G}_0)^t \Pi_0 \ket{\psi}} + 2t^2 \cdot \varepsilon(t) \\
        &\le 2 t \cdot \varepsilon(t) + 2t^2 \cdot \varepsilon(t) \\
        &\le 4 t^2 \cdot \varepsilon(t). \qedhere
    \end{align*}
\end{proof}

\begin{proof}[Proof of the admissible oracle lemma (\cref{lemma:admissible-oracles})]
Let $\{\Pi_{\D^G}, \Id - \Pi_{D^G}\}$ be the binary projective measurement corresponding to running distinguisher $D^G$ and measuring the output bit that distinguishes $\Pi \ket{\psi}$ and $W \Pi \ket{\psi}$ with $\varepsilon$ advantage. By \cref{lemma:duality}, \cref{item:dist-to-map}, the unitary operator $\widetilde{U} \coloneqq \Id - 2 \Pi_{D^G}$ achieves a mapping advantage at least $\varepsilon^2$ on initial state $\ket{\psi}$. The unitary $\widetilde{U}$ can be implemented by running $D^G$, applying $Z$ to the output qubit, then running $(D^G)^{\dagger}$. This requires access to two oracles $G$ and $G^{\dagger}$. For convenience, we will introduce the new admissible oracle $\widetilde{G} \coloneqq \ketbra{0}_{\RegX} \otimes G + \ketbra{1}_{\RegX} \otimes G^{\dagger}$, which can simulate the action of $G$ and $G^{\dagger}$ using an ancilla qubit $\RegX$. For this reason, we write the unitary as $\widetilde{U}^{\widetilde{G}}$. 

The circuit $\widetilde{U}^{\widetilde{G}}$, which makes $t+1=O(T)$ calls to $\widetilde{G}$, can be converted to the form $(U \widetilde{G})^{t}$ for some fixed unitary $U$ making no calls to $\widetilde{G}$ by appending a clock register $\RegY$ with $\poly(T)$ overhead in overall circuit size and ancilla registers. Then:
\[
\norm{(W^- \Pi) (U\widetilde{G})^t (W^+ \Pi) (\ket{0}_{\RegX} \ket{0}_{\RegY} \ket{\psi})}^2 \ge \varepsilon^2.
\]
Then by \cref{lemma:main}, there exist $q,r,s \le t$ such that
\[
\norm{(W^- \Pi) U^q (W^+ \Pi) (\widetilde{G}_0)^r (U\widetilde{G}_0)^s (W^+ \Pi) (\ket{0}_{\RegX} \ket{0}_{\RegY} \ket{\psi})}^2 \ge \varepsilon^2/4t^2
\]
where $\widetilde{G}_0 = W^+\Pi\widetilde{G} + (\Id-W^+\Pi)$. The unitary $\widetilde{G}_0$ can be implemented using $O(1)$ applications of controlled $G$, $G^{\dagger}$ and the unitary $M \coloneqq (W^+ \Pi) \otimes X + (\Id - W^+ \Pi) \otimes \Id$.

Let $\Prep^{G,M}$ be the quantum algorithm that on input $\ket{\psi}$ creates the sub-normalized state\footnote{By creating a sub-normalized state, we mean that if a measurement $\{W^+ \Pi, \Id - W^+ \Pi\}$ rejects (i.e., apply $M$ on the state and an ancilla qubit and measure the ancilla qubit), abort and do not produce any output. Otherwise, output the resulting state $\ket{\gamma}$.}
\[
    \ket{\gamma} \coloneqq (W^+ \Pi) (\widetilde{G}_0)^r (U\widetilde{G}_0)^s (W^+ \Pi) (\ket{0}_{\RegX} \ket{0}_{\RegY} \ket{\psi}),
\]
and then uses $\ket{\gamma}$ to prepare the state 
\[
    \ket{\phi} = \ctl_{\RegZ}\mh U^q \left( \ket{+}_{\RegZ} \ket{\gamma} \right).
\]
By \cref{lemma:duality}, \cref{item:map-to-dist}, if we define the projector $\widetilde{\Pi}_{D} \coloneqq \ctl_{\RegZ} \mh U^q  \ketbra{+}_{\RegZ} \ctl_{\RegZ} \mh (U^q)^{\dagger}$, then the binary projective measurement $\{ \widetilde{\Pi}_{D}, \Id - \widetilde{\Pi}_{D} \}$ distinguishes $\Pi \ket{\phi}$ and $W \Pi \ket{\phi}$ with advantage $\varepsilon^2/8t^2$. Thus, the algorithm $E$ that measures the binary projective measurement $\{\widetilde{\Pi}_D, \Id - \widetilde{\Pi}_D\}$ wins the $(W,\Pi)$-distinguishing game with advantage $\varepsilon^2/8t^2$. 
\end{proof}

\subsection{Swap binding implies oracle swap binding}
\label{subsec:swap-binding-implies-oracle}

We now prove the main theorem of this section, \cref{theorem:oracle-binding-security}, which says that any swap binding commitment scheme is oracle swap binding.

\begin{proof}[Proof of \Cref{theorem:oracle-binding-security}]
    Suppose that a size-$T$ adversary $A^{G_b}$ for oracle swap-binding (counting queries to $G_b$ as unit cost) achieves distinguishing advantage $\varepsilon$ with initial state $\ket{\psi}_{\RegC \RegD \RegR}$, where $\RegR$ is the adversary's auxiliary registers. We will use \cref{lemma:admissible-oracles} to translate this adversary into one that achieves distinguishing advantage $\varepsilon/16T^2$ in the standard swap-binding game.

    Let $\ctl_{\RegB}\mh S \coloneqq \ketbra{0}_{\RegB} \otimes \Id + \ketbra{1}_{\RegB} \otimes \widehat{\SWAP[\RegM, \RegM']}$, where $\RegB$ is a single qubit register and $\RegM'$ is a register of the same dimension as $\RegM$. Then define the oracle
    \[
        G = (\ctl_{\RegB} \mh S) \widehat{\cO} (\ctl_{\RegB} \mh S) \Pi  + (\Id - \Pi)
    \]
    and operation
    \begin{align} \label{def:swap-W}
        W = (\ctl_{\RegB}\mh S) X_{\RegB} (\ctl_{\RegB} \mh S) = X_{\RegB} \otimes  \widehat{\SWAP[\RegM, \RegM']}.
    \end{align}
    
    Observe that $A^G$ is the same quantum algorithm as $A^{G_b}$, except the calls to $G_b$ are replaced by calls to $G$. Since $A^{G_b}$ achieves distinguishing advantage $\varepsilon$ in oracle swap-binding with initial state $\ket{\psi}_{\RegC \RegD \RegR}$, $A^G$ achieves distinguishing advantage $\varepsilon$ in the $(W,\Pi)$-distinguishing game with registers $\RegA \coloneqq (\RegB, \RegM', \RegC)$, $\RegB \coloneqq \RegD$ and initial state $\ket{0}_{\RegB} \otimes \ket{0}_{\RegM'} \otimes \ket{\psi}_{\RegC \RegD \RegR}$.
    
    By \cref{lemma:admissible-oracles}, there exists another quantum state $\brho$ and efficient algorithm $E$ \emph{without query access to $G$} such that the following hold:
    \begin{itemize}
        \item there is a size-$\poly(T)$ oracle circuit $\Prep^{G,M}$ such that $\brho \leftarrow \Prep^{G,M} ( \ket{0}_{\RegB \RegM'} \ket{\psi}_{\RegC \RegD \RegR} \ket{0}_{\anc})$ and $\Prep$ uses at most $T$ (controlled) queries to $G$, $G^{\dagger}$, and $M \coloneqq (W^+ \Pi) \otimes X + (\Id - W^+ \Pi) \otimes \Id$,
        \item $E$ has size $\poly(T)$, and
        \item $E$ wins the $(W,\Pi)$-distinguishing game with advantage $\varepsilon^2 / 8T^2$ using initial state $\brho$.
    \end{itemize}
    If $\cO$ is efficiently implementable, then (controlled) $G$ and $G^{\dagger}$ are also efficiently implementable. In addition, with our given definition of $W$, the binary projective measurement $\{W^+ \Pi, \Id - W^+ \Pi\}$ is efficiently implementable. This is because  
    \[
        W^+ \Pi = \left(\ketbra{+}_{\RegB} \otimes \widehat{\Pi_{\Sym}} + \ketbra{-}_{\RegB} \otimes\widehat{\Pi_{\Antisym}} \right) \Pi,
    \]
    where $\Pi_{\Sym}$ and $\Pi_{\Antisym}$ are the projectors onto the the symmetric and antisymmetric subspaces of $\cM \otimes \cM'$, which can be measured using a $\SWAP$ test. Thus $M$ is efficiently implementable. For this reason, we now write $\Prep \coloneqq \Prep^{G,M}$ because $\Prep$ can implement the functionality of controlled $G$, $G^{\dagger}$, and $M$ without any oracle access.
    
    Thus $E$ can win the $(W,\Pi)$-distinguishing game using initial state $\brho \leftarrow \Prep( \ket{0}_{\RegB \RegM'} \ket{\psi}_{\RegC \RegD \RegR} \ket{0}_{\anc})$ with advantage $\varepsilon^2 / 8T^2$. The conclusion follows by an application of \cref{claim:swap-binding-to-swap-binding-with-registers} (which we state and prove below), which gives a distinguisher for swap-binding with advantage $\varepsilon^2 / 16T^2$.
\end{proof}

\begin{claim} \label{claim:swap-binding-to-swap-binding-with-registers}
    Suppose that an efficient distinguisher $D$ wins the $(W,\Pi)$-distinguishing game with advantage $\delta$, where $W$ and $\Pi$ are defined as 
    \[
        W = X_{\RegB} \otimes \widehat{\SWAP[\RegM, \RegM']} \text{ and } \Pi = \Com (\Id_\RegM \otimes \ketbra{0}_\RegW) \Com^{\dagger}.
    \]
    Then there is an efficient distinguisher for the swap-binding game
    (\cref{def:swap-binding}) with advantage $\delta / 2$.
\end{claim}
\begin{proof}
    Suppose that $D$ wins the $(W,\Pi)$-distinguishing game with advantage $\delta$ using the state $\brho_{\RegB \RegM' \RegC \RegD \RegR}$. This means that $D$ can distinguish the states $\Tr_{\RegB \RegM' \RegC}(\Pi \brho \Pi)$ and $\Tr_{\RegB \RegM' \RegC} (W \Pi \brho \Pi W)$ with advantage $\delta$. Because the $\RegB$ register is traced out in the $(W,\Pi)$-distinguishing game, the operation $X_{\RegB}$ in $W$ has no effect on the distinguisher's view, so we may rewrite the latter state as
    \begin{align} \label{eq:world1-W-application}
        \Tr_{\RegB \RegM' \RegC} \left( W \Pi \brho \Pi W \right) = \Tr_{\RegB \RegM' \RegC} \left( \left( \widehat{\SWAP[\RegM, \RegM']} \right) \cdot \Pi \brho \Pi \cdot \left( \widehat{\SWAP[\RegM, \RegM']} \right)  \right) . 
    \end{align}
    
    We use a hybrid argument to show that the distinguisher $D$ yields an adversary for swap-binding with advantage $\delta/2$. We define hybrid games $H_0$, $H_1$, and $H_2$, where the adversary's view in hybrid $H_0$ corresponds to $\Tr_{\RegB \RegM' \RegC}(\Pi \brho \Pi)$ and the adversary's view in hybrid $H_2$ corresponds to $\Tr_{\RegB \RegM' \RegC} (W \Pi \brho \Pi W)$.
    
    \begin{itemize}
        \item $H_0$:
            \begin{enumerate}
                \item The adversary sends a quantum state $\brho_{\RegB \RegM' \RegC \RegD \RegR}$ to the challenger. 
                \item The challenger then does the following: 
                    \begin{enumerate}
                        \item Measure $\{\Pi, \Id - \Pi\}$ and abort if the measurement rejects. 
                        \item Send the $(\RegD, \RegR)$ register to the adversary.
                    \end{enumerate}
            \end{enumerate}
        \item $H_1$:
            \begin{enumerate}
                \item The adversary sends a quantum state $\brho_{\RegB \RegM' \RegC \RegD \RegR}$ to the challenger. 
                \item The challenger then does the following: 
                    \begin{enumerate}
                        \item Measure $\{\Pi, \Id - \Pi\}$ and abort if the measurement rejects. 
                        {\color{red} \item Initialize another register $\RegE$ (of the same dimension as $\RegM$) to the $\ket{0}$ state.
                        \item Apply $\widehat{\SWAP[\RegM, \RegE]}$. }
                        \item Send the $(\RegD, \RegR)$ register to the adversary.
                    \end{enumerate}
            \end{enumerate}
            This is immediately indistinguishable from $H_0$ by swap-binding.
        \item $H_2$: 
            \begin{enumerate}
                \item The adversary sends a quantum state $\brho_{\RegB \RegM' \RegC \RegD \RegR}$ to the challenger. 
                \item The challenger then does the following: 
                    \begin{enumerate}
                        \item Measure $\{\Pi, \Id - \Pi\}$ and abort if the measurement rejects. 
                        {\color{red}\item Apply $\widehat{\SWAP[\RegM, \RegM']}$.}
                        \item Send the $(\RegD, \RegR)$ register to the adversary.
                    \end{enumerate}
            \end{enumerate}
            We show that $H_2$ is indistinguishable from $H_1$ by invoking swap-binding. Given an adversary that distinguishes $H_1$ from $H_2$ with advantage $\delta'$, we construct the following reduction that distinguishes the $b=0$ and $b=1$ worlds of the swap binding game with advantage $\delta'$:
             \begin{enumerate}
                 \item Receive $\brho_{\RegB \RegM' \RegC \RegD \RegR}$ from the adversary.
                 \item Measure $\{\Pi, \Id - \Pi\}$ and abort if the measurement rejects.
                 \item Apply $\widehat{\SWAP[\RegM, \RegM']}$.
                 \item Send registers $(\RegC, \RegD)$ to the swap-binding challenger.
                 \item Receive $\RegD$ from the swap-binding challenger. 
                 \item Forward $(\RegD, \RegR)$ to the adversary.
                 \item Use the adversary's output as the guess for the swap-binding game.
             \end{enumerate}
             Because the reduction performs the measurement $\{\Pi, \Id-\Pi\}$, the challenger's measurement has no effect. In the $b=0$ world of the swap-binding game, the challenger will do nothing, which corresponds exactly to $H_2$. In the $b=1$ world, the challenger will apply $\widehat{\SWAP[\RegM, \RegE]}$ where $\RegE$ is initialized to $\ket{0}$. Thus, the adversary receives the decommitment after the operation $\widehat{\SWAP[\RegM, \RegE]} \cdot \widehat{\SWAP[\RegM, \RegM']} = \SWAP[\RegM', \RegE] \cdot \widehat{\SWAP[\RegM, \RegE]}$ has been applied. Because the final $\SWAP[\RegM', \RegE]$ has no effect on the adversary's view, this corresponds exactly to $H_1$.    
    \end{itemize}
\end{proof}

\newpage

\section{A succinct quantum argument protocol}
\label{sec:quantum-kilian}

In this section, we construct an interactive succinct quantum argument for any language with quantum PCPs with constant soundness error and polylogarithmic query complexity.

\subsection{Preliminaries}

\subsubsection{Probabilistically checkable proofs}
\label{subsubsec:pcps}

\begin{definition}[Quantum Probabilistically Checkable Proofs]
\label{def:quantum-pcp}
A quantum PCP for a language $L$ is parameterized by a completeness parameter $c$, soundness $s$, proof length $m$, randomness complexity $\ell$, and query complexity $q$. We require the following properties: 
\begin{itemize}
    \item (Efficient verification) There is a classical $\poly(n)$ time procedure that takes as input $x \in \{0,1\}^n$ and $r \in \{0,1\}^{\ell}$ and outputs the description of circuit for implementing a $q$-qubit projective measurement  $\{\Pi^{\PCP}_{x,r},\Id-\Pi^{\PCP}_{x,r}\}$, which acts on a state of size $m$.
    \item (Completeness) If $x \in L$, there exists an $m$-qubit state $\bpi$ such that
    \[ \displaystyle \mathop{\mathbb{E}}_{r \gets \{0,1\}^{\ell}} \Tr(\Pi_{x,r}^{\PCP} \bpi) \geq c.\]
    \item (Soundness) If $x \not\in L$, then for any $m$-qubit state $\bpi$,
    \[ \displaystyle \mathop{\mathbb{E}}_{r \gets \{0,1\}^{\ell}} \Tr(\Pi_{x,r}^{\PCP} \bpi) \leq s.\]
\end{itemize}
\end{definition}

\paragraph{Which languages have good quantum PCPs?} It is easy to check that any \emph{classical} PCP is captured by this definition. Define the measurements $\{\Pi^{\PCP}_{x,r},\Id-\Pi^{\PCP}_{x,r}\}$ corresponding to the predicate that the classical PCP verifier checks. Completeness is trivially preserved since any classical proof string $\pi$ is also a quantum state. Soundness is preserved because the measurements  $\{\Pi^{\PCP}_{x,r},\Id-\Pi^{\PCP}_{x,r}\}$ correspond to classical predicates and thus commute with standard basis measurements. If $x \not\in L$, the probability that a quantum state $\ket{\psi} = \sum_{\pi \in \{0,1\}^m} \alpha_\pi \ket{\pi}$ is accepted is \[\sum_{\pi \in \{0,1\}^m} \abs{\alpha_\pi}^2 \Pr[\pi \text{ is accepted}] \leq \sum_{\pi \in \{0,1\}^m} \abs{\alpha_\pi}^2 s = s.\]
The statement for mixed states follows by convexity.

\begin{theorem}[PCP theorem~\cite{STOC:BFLS91,FOCS:FGLSS91,AroraS98,AroraLMSS98}]
\label{theorem:classical-pcp-theorem}
    There exist constants $0 \leq s < c \leq 1$ and $k = O(1)$ such that every $\NP$ language has a PCP with completeness $c$, soundness $s$, proof length $m = \poly(n)$, randomness complexity $\ell = O(\log n)$, and query complexity $q = O(\log n)$.
\end{theorem}

\begin{conjecture}[Quantum PCP conjecture~\cite{AALV, AAV13}]
\label{conj:quantum-pcp}
    There exist constants $0 \leq s < c \leq 1$ and $k = O(1)$ such that every $\QMA$ language has a quantum PCP with completeness $c$, soundness $s$, proof length $m = \poly(n)$, randomness complexity $\ell = O(\log n)$, and query complexity $q = O(\log n)$.
\end{conjecture}

Any quantum PCP with constant completeness-soundness gap can be repeated in parallel $\log^2(\lambda)$ times to achieve $1-\negl(\lambda)$ completeness soundness error. The proof of soundness amplification is the same as the standard proof that $\QMA$ can be amplified to exponentially small error by parallel repetition \cite{KSV02}.

\begin{claim}[Soundness amplification]
Suppose that a language $L$ has a quantum PCP with constant completeness-soundness gap $c-s$, proof length $m$, randomness complexity $\ell$, and query complexity $q$. Then $L$ has a quantum PCP with completeness $1-\negl(\lambda)$, soundness $\negl(\lambda)$, proof length $m \cdot \polylog(\lambda)$, randomness complexity $\ell \cdot \polylog(\lambda)$, and query complexity $q \cdot  \polylog(\lambda)$.
\end{claim}

\subsubsection{Interactive argument preliminaries}

\paragraph{Quantum interactive protocol syntax.} We first establish a syntax for general quantum interactive protocols between a prover and a verifier. We state the definition below for any $(2r-1)$-round protocol, though in this work we will focus on $3$-round protocols.
\begin{definition}
    A $(2r-1)$-message quantum interactive protocol $\langle P, V \rangle$ between a prover $P$ and verifier $V$ is specified by two sequences of unitaries $U_{P,1},\dots,U_{P,r}$ and $U_{V,1},\dots,U_{V,r}$. The prover and verifier have internal registers $\RegP$ and $\RegV$ respectively and also share a message register $\RegZ$. The interaction proceeds as follows:
    \begin{enumerate}
        \item[] For $i = 1,\dots,r$, the prover and verifier do the following:
        \begin{itemize}
            \item The prover applies $U_{P,i}$ to $(\RegP, \RegZ)$ and sends $\RegZ$ to the verifier.
        \item The verifier applies $U_{V,i}$ to $(\RegV, \RegZ)$. If $i < r$, it sends $\RegZ$ back to the prover. If $i = r$, it measures the first qubit of $\RegV$ and accepts if it is $1$ and rejects if it is $0$.
        \end{itemize}
    \end{enumerate}
    
    We write $\langle P(\brho_{\RegP}),V(\btau_{\RegV}) \rangle_{\mathrm{OUT}}$ to denote the $0/1$-random variable corresponding to the verifier's decision when the protocol is run with the $\RegP$ register initialized to $\brho_{\RegP}$ and the $\RegV$ register initialized to $\btau_{\RegV}$.
    
\end{definition}

The following definition of argument of knowledge is based on \cite[Definition~3.6]{FOCS:CMSZ21} with some modifications to handle extraction of quantum proofs.

\begin{definition}
Let $\mathsf{Ver}$ be a quantum algorithm that takes as input a state $\brho \in \bfD(\cH)$ and outputs an accept/reject decision. Define $\mathsf{Wit}[\mathsf{Ver},p] \subseteq \bfD(\cH)$ to be the set of all states $\brho \in \bfD(\cH)$ such that $\Pr[\mathsf{Ver} \text{ accepts } \brho] \geq p$.
\end{definition}

\begin{definition}[Argument of knowledge]
\label{def:argument-of-knowledge}
    Let $L$ be a $\QMA$ language and let $\mathsf{Ver}_{L,x}$ be the corresponding $\QMA$ verification procedure for statement $x$. An interactive argument $\langle P,V \rangle$ is an argument of knowledge for $L$ with respect to $\mathsf{Ver}_L$ if there exists a quantum algorithm $\KnowledgeExt$ such that for any $x$, any polynomial-time quantum interactive prover $\tilde{P}$ that convinces the argument verifier with probability $p' \coloneqq \Pr[\langle \tilde{P}(x), V(x)\rangle_{\mathrm{OUT}} = 1]$, any $\varepsilon \geq 1/\poly(\lambda)$, and any $p \leq p'$,\footnote{Following~\cite{FOCS:CMSZ21}, we give the extractor a lower bound $p$ on the success probability $q$ as classical advice. It is plausible that this requirement could be removed using an extractor based on techniques of~\cite{FOCS:LomMaSpo22}, but we did not attempt to work out the details.} 
    \[
        \Pr[ \brho \in \mathsf{Wit}[\mathsf{Ver}_{L,x},p-\varepsilon] :  \brho \leftarrow \KnowledgeExt^{\tilde{P}}(x,1^{\lceil 1/p\rceil})] \geq \Omega(\varepsilon).
    \]
    The runtime of $\KnowledgeExt^{\tilde{P}}(x,1^{\lceil1/p\rceil})$ is $\poly(\lambda,1/\varepsilon)$, counting calls to $\tilde{P}$ as unit time.
\end{definition}

At an intuitive level, this definition states that a protocol is an argument of knowledge with respect to a proof verifier $\mathsf{Ver}$ (which will typically correspond to a PCP verifier) if, given access to any efficient malicious prover that convinces the argument verifier with probability $p$, the extractor can (with noticeable probability) produce a proof string that convinces $\mathsf{Ver}$ with probability $\approx p$.

\subsection{Quantum tree commitments}
\label{subsec:quantum-merkle}
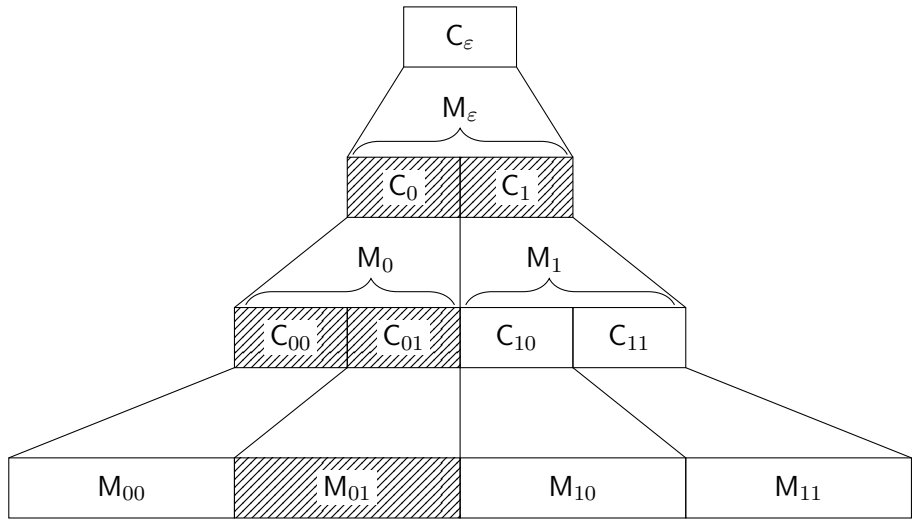
\begin{figure}[h!]
    \centering
    \begin{tikzpicture}
        \def\width{1.5}
        \def\height{0.8}
        \def\vspacing{2}
        %% level 0 / root
        \draw (3.5*\width,3*\vspacing) rectangle (4.5*\width,3*\vspacing+\height);
        \node at (4*\width,3*\vspacing+0.5*\height) {$\RegC_\varepsilon$};
        \draw (3.5*\width,3*\vspacing) -- (3*\width,2*\vspacing+\height);
        \draw (4.5*\width,3*\vspacing) -- (5*\width,2*\vspacing+\height);
        %% level 1
        \draw [decorate,decoration = {brace,raise=1pt,amplitude=10pt}] (3.05*\width,2*\vspacing+\height) --  (4.95*\width,2*\vspacing+\height);
        \node at (4*\width, 2*\vspacing+1.8*\height) {$\RegM_\varepsilon$};
        \draw[pattern=north east lines] (3*\width,2*\vspacing) rectangle (4*\width,2*\vspacing+\height);
        \node[fill=white,inner sep=2pt] at (3.5*\width,2*\vspacing+0.5*\height) {$\RegC_0$};
        \draw (3*\width,2*\vspacing) -- (2*\width,\vspacing+\height);
        \draw (4*\width,2*\vspacing) -- (4*\width,\vspacing+\height);
        \draw[pattern=north east lines] (4*\width,2*\vspacing) rectangle (5*\width,2*\vspacing+\height);
        \node[fill=white,inner sep=2pt] at (4.5*\width,2*\vspacing+0.5*\height) {$\RegC_1$};
        \draw (5*\width,2*\vspacing) -- (6*\width,\vspacing+\height);
        % level 2 M2
        \draw [decorate,decoration = {brace,raise=1pt,amplitude=10pt,aspect=0.6}] (2.1*\width,\vspacing+\height) --  (3.95*\width,\vspacing+\height);
        \node at (3.25*\width, \vspacing+1.8*\height) {$\RegM_0$};
        %% level 2 part 1
        \draw[pattern=north east lines] (2*\width,\vspacing) rectangle (3*\width,\height+\vspacing);
        \node[fill=white,inner sep=2pt] at (2.5*\width,\vspacing+0.5*\height) {$\RegC_{00}$};
        \draw (2*\width,\vspacing) -- (0,\height);
        \draw[pattern=north east lines] (3*\width,\vspacing) rectangle (4*\width,\height+\vspacing);
        \node[fill=white,inner sep=2pt] at (\width+2.5*\width,\vspacing+0.5*\height) {$\RegC_{01}$};
        \draw (\width+2*\width,\vspacing) -- (2*\width,\height);
        \draw [decorate,decoration = {brace,raise=1pt,amplitude=10pt,aspect=0.4}] (4.05*\width,\vspacing+\height) --  (5.9*\width,\vspacing+\height);
        \node at (4.75*\width, \vspacing+1.8*\height) {$\RegM_1$};
        %% level 2 part 2
        \draw (2*\width+2*\width,\vspacing) rectangle (2*\width+3*\width,\height+\vspacing);
        \node at (2*\width+2.5*\width,\vspacing+0.5*\height) {$\RegC_{10}$};
        \draw (2*\width+2*\width,\vspacing) -- (4*\width,\height);
        \draw (3*\width+2*\width,\vspacing) rectangle (3*\width+3*\width,\height+\vspacing);
        \node at (3*\width+2.5*\width,\vspacing+0.5*\height) {$\RegC_{11}$};
        \draw (3*\width+2*\width,\vspacing) -- (6*\width,\height);
        %% level 3 part 1
        \draw (0,0) rectangle (2*\width,\height);
        \node at (\width,0.5*\height) {$\RegM_{00}$};
        \draw[pattern=north east lines] (2*\width,0) rectangle (2*\width+2*\width,\height);
        \node[fill=white,inner sep=2pt] at (2*\width+\width,0.5*\height) {$\RegM_{01}$};
        %% level 3 part 2
        \draw (4*\width,0) rectangle (4*\width+2*\width,\height);
        \node at (4*\width+\width,0.5*\height) {$\RegM_{10}$};
        \draw (6*\width,0) rectangle (6*\width+2*\width,\height);
        \node at (6*\width+\width,0.5*\height) {$\RegM_{11}$};
        \draw (6*\width,\vspacing) -- (8*\width,\height);
    \end{tikzpicture}
    \caption{The Merkle tree structure. For simplicity, we assume that $\RegD_\ell = \RegM_\ell$ in this picture (as is the case in, for example, the scheme in \cref{theorem:pru-construction}). The shaded boxes correspond to decommitments $\{\RegD_\varepsilon,\RegD_0,\RegD_{01}\}$ for the $01$ block.}
    \label{fig:merkle-tree}
\end{figure}

We now describe the syntax for quantum tree commitments. This syntax is a generalization of a construction that appeared in prior work of Chen and Movassagh~\cite{ARXIV:CheMov21}; in our notation, their construction corresponds to a setting where the ancilla register $\RegW$ and the commitment $\RegC$ must be the same size. 

Let $\SQSC$ be a succinct QSC with $s(\lambda)$-sized messages and $s(\lambda)/2$-sized commitments and let $\Com$ be the unitary implementing the commitment. For any positive integer $\beta$, a quantum tree commitment for $(s \cdot 2^\beta)$-size messages $\RegM$ supports the following functionality:
\begin{itemize}
    \item The tree commitment algorithm $\mathsf{QTreeCom}_\SQSC$ does the following:
    \begin{enumerate}
        \item Divide $\RegM$ into $2^\beta$ registers $\RegM = \{\RegM_{\ell}\}_{\ell \in \{0,1\}^{\beta}}$ where each $\RegM_\ell$ is a block of $s$ qubits.
        \item For $j = \beta,\beta-1,\dots,1$:
        \begin{enumerate}
            \item For each $\ell \in \{0,1\}^{j}$, initialize $\RegW_\ell$ to $\ket{0}$ and apply $\Com$ to $(\RegM, \RegW)$ to obtain $(\RegC_\ell, \RegD_\ell)$. 
            \item For each $\ell \in \{0,1\}^{j-1}$, define $\RegM_{\ell} \coloneqq (\RegC_{\ell,0}, \RegC_{\ell,1})$ (when $j = 1$, we define $\RegM_{\varepsilon} \coloneqq (\RegC_0, \RegC_1)$ where $\varepsilon$ is the empty string).
        \end{enumerate}
        \item Initialize $\RegW_\varepsilon$ to $\ket{0}$ and apply $\Com$ to $(\RegM_{\varepsilon}, \RegW)$ to obtain $(\RegC_{\varepsilon}, \RegD_{\varepsilon})$. The commitment is $\RegC_{\varepsilon}$ and $\{\RegD_{\ell}\}_{\ell \in \{0,1\}^{\leq \beta}}$ is a set of decommitments which will be used later to form the local decommitments.
    \end{enumerate}
    \item For any subset $S \subseteq \{0,1\}^{\le \beta}$, we define the local decommitment $\decom_S$ as follows:
    \begin{itemize}
        \item For any $\ell \in \{0,1\}^{\le \beta}$, let $\Path(\ell)$ be the set of all $\abs{\ell}+1$ prefixes of $\ell$ (including $\ell$ itself and the empty string $\varepsilon$).
        \item Let $\Path(S) \coloneqq \bigcup_{\ell \in S} \Path(\ell)$.
        \item Finally, define $\decom_S \coloneqq  \{\RegD_{\ell}\}_{\ell \in \Path(S)}$.
    \end{itemize}
    \item The receiver can verify $\decom_S$ as follows:
    \begin{enumerate}
        \item Parse $\decom_S$ as $\{\RegD_{\ell}\}_{\ell \in \Path(S)}$.
        \item For $j = 0,\dots,\beta$:
        \begin{enumerate}
            \item For each $\ell \in \Path(S) \cap \{0,1\}^j$:
            \begin{enumerate}
                \item Apply $\Com^{\dagger}$ to $(\RegC_{\ell}, \RegD_{\ell})$ to obtain  $(\RegM_{\ell},  \RegW_{\ell})$. 
                \item Measure $\RegW_\ell$ with $(\ketbra{0}, \Id - \ketbra{0})$ and abort if the outcome is not $\ket{0}$.
                \item If $j < \beta$, partition $\RegM_\ell$ into two $(s/2)$-qubit registers $(\RegC_{\ell,0}, \RegC_{\ell,1})$.
            \end{enumerate}
        \end{enumerate}
        \item Output $\bigotimes_{\ell \in S} \RegM_\ell$.
    \end{enumerate}
\end{itemize}

\subsection{The protocol description}
\label{sec:kilian-construction}
In this section we construct a succinct quantum argument from any quantum PCP and succinct QSC. The protocol is parameterized by a choice of the underlying $\PCP$ as well as the succinct QSC $\SQSC$, where:
\begin{itemize}
    \item $\PCP$ is a classical or quantum PCP system for a language $L$.
    \item $\SQSC$ is a succinct QSC with $s(\lambda)$-sized messages and $s(\lambda)/2$-sized commitments.
\end{itemize}
For any setting of $\PCP$ and $\SQSC$ satisfying these requirements, we denote the resulting succinct quantum argument as $\SQUARG[\PCP,\SQSC]$ (for Succinct QUantum ARGument). The protocol is specified as follows.

\begin{itemize}
    \item[] $\SQUARG[\PCP,\SQSC]$:
    \begin{enumerate}
    \item[] \textbf{Prover input:} $x$ and a corresponding PCP $\ket{\pi}$. We assume the PCP has length $m = s \cdot 2^\beta$ qubits for some integer $\beta$, which is without loss of generality (by padding with $0$'s).
    \item[] \textbf{Verifier input:} $x$.
    \item The prover applies $\mathsf{QTreeCom}_\SQSC$ to $\ket{\pi}$ and obtains $\RegC_\varepsilon,\{\RegD_{\ell}\}_{\ell \in \{0,1\}^{\leq \beta}}$. It sends the root commitment $\RegC_{\varepsilon}$ to the verifier. 
    \item The verifier sends random coins $r \leftarrow \{0,1\}^{\mathsf{rc}}$ to the prover.
    \item The prover sends $\decom_{S_r}$ where $S_r \subseteq \{0,1\}^\beta$ is the set of block labels corresponding to the positions that $\{\Pi_r, \Id - \Pi_r\}$ checks (i.e., those that contain $Q_r$). 
    \item Finally, the verifier checks $\decom_{S_r}$. If the decommitment is valid, it also obtains the qubits $Q_r$ of the PCP corresponding to $r$ and measures $\{\Pi_r, \Id - \Pi_r\}$. It accepts if the measurement accepts.
    \end{enumerate}
\end{itemize}

The following is immediate by construction.

\begin{claim}[Succinctness]
\label{claim:succinctness}
Suppose that for an instance of size $n$, we have a $\PCP$ with proofs of length $m = \poly(n,\lambda)$, query complexity $\poly(\lambda)$, randomness complexity $\poly(\lambda)$, and completeness-soundness gap $1-\negl(\lambda)$. Moreover, assume that the PCP verifier runs in time $t(n,\lambda)$. Then the protocol $\SQUARG[\PCP,\SQSC]$ satisfies the following properties:
\begin{itemize}
    \item The quantum prover runs in time $\poly(m,\lambda)$.
    \item The quantum verifier runs in time $t(n,\lambda) + \poly(\lambda,\log(m))$.
    \item The total quantum communication is $\poly(\lambda,\log(m))$ qubits.
\end{itemize}
\end{claim}

\newpage

\section{Security of the quantum succinct argument protocol}
\label{sec:rewinding}

In this section we prove security of the quantum succinct argument protocol described in~\cref{sec:kilian-construction}.

\begin{theorem}
\label{theorem:kilian-security}
Let $\PCP$ be a probabilistically checkable proof for a language $L$ and let $\SQSC$ be a succinct QSC with $s$-qubit messages and $s/2$-qubit commitments. If $\SQSC$ is swap binding, then $\SQUARG[\PCP,\SQSC]$ is an argument of knowledge for $L$ with respect to the $\PCP$ verifier (\cref{def:argument-of-knowledge}).
\end{theorem}

By~\cref{claim:succinctness}, we obtain the following corollaries of~\cref{theorem:kilian-security}.

\begin{corollary}
\label{corollary:succinct-np-body}
Assuming the existence of succinct QSCs, there is a three-message quantum succinct argument of knowledge for $\NP$.\footnote{This follows by combining~\cref{theorem:intro-kilian} with the classical PCP theorem~\cite{STOC:BFLS91,FOCS:FGLSS91,AroraS98,AroraLMSS98}, since quantum PCPs encompass classical PCPs.}
\end{corollary}

\begin{corollary}
\label{corollary:succinct-qma-body}
If the quantum PCP conjecture holds~\cite{AALV, AAV13}, then assuming the existence of succinct QSCs, there is a three-message quantum succinct argument (of knowledge) for $\QMA$.
\end{corollary}

\begin{remark}
Technically, our definition of argument of knowledge (\cref{def:argument-of-knowledge}) states that if the argument prover convinces the argument verifier with probability $p$, then the knowledge extractor outputs a PCP (with noticeable probability) that would convince the PCP verifier with probability $\approx p$. In the setting of arguments/proofs of knowledge for $\NP$, our notion is slightly different from the traditional notion, which requires that the extractor output an $\NP$ witness. (PCPs are not $\NP$ witnesses since they are not guaranteed to be deterministically verifiable.)

However, this issue can easily be resolved assuming the PCP is ``witness-extractable'' (see~\cite{TCC:Valiant08} for a discussion). Moreover, this distinction does not arise in the $\QMA$ setting, since the definition of $\QMA$ only requires probabilistic verification, so quantum PCPs are already valid $\QMA$ witnesses.
\end{remark}

The rest of this section is organized as follows:
\begin{itemize}
    \item In~\cref{sec:one-bit}, we recall the ``one-bit'' extraction procedure from~\cite{FOCS:CMSZ21}.
    \item In~\cref{subsec:extractor-description}, we describe a knowledge extraction procedure that outputs a quantum PCP given black-box access to any malicious prover for the quantum succinct argument.
    \item In~\cref{sec:extractor-success}, we prove that our knowledge extractor satisfies the conditions of~\cref{def:argument-of-knowledge}, establishing~\cref{theorem:kilian-security}.
\end{itemize}

\subsection{CMSZ preliminaries}
\label{sec:one-bit}

\paragraph{Recap: estimating success probability.} A key component of the~\cite{FOCS:CMSZ21} extraction procedure is a subroutine called $\Est$ defined with respect to a set of projectors $\{\Pi_r\}_r$. At a high level, $\Est$ is a quantum measurement on a state $\brho$ that estimates the probability that, for a random $r$, the measurement $\{\Pi_r,\Id-\Pi_r\}$ accepts. In~\cite{FOCS:CMSZ21}, this $\Est$ procedure is implemented using the alternating projectors trick of~\cite{MarWat05} and incorporates additional modifications proposed by~\cite{TCC:Zhandry20}. 

To make this concrete, fix a Hilbert space $\cH$, a corresponding register $\RegH$, and a state $\brho \in \bfD(\cH)$. Let $\sfPi \coloneqq \{\Pi_r\}_{r \in R}$ denote a set of projections acting on $\RegH$. In our setting, each $\Pi_r$ will correspond to the projection that checks whether the prover outputs an accepting response on the challenge $r \in R$. Thus, the ``success probability'' of $\brho$ is defined to be $\mathbb{E}_{r \from R}[\Tr(\Pi_r \brho)]$.  

The $\Est$ algorithm is also parameterized by an additive error bound $0< \varepsilon < 1$ and a failure probability $0 < \delta < 1$.\footnote{In~\cite{TCC:Zhandry20,FOCS:CMSZ21}, this is captured by the notion of $(\varepsilon,\delta)$-almost-projectivity defined in~\cite{TCC:Zhandry20}: for any initial state $\brho$, if $\Est_{\sfPi,\varepsilon,\delta}$ is applied \emph{twice in a row}, the probability that the resulting outcomes $p,p'$ are more than $\varepsilon$ apart is at most $\delta$.} We give a formal description below, following~\cite{FOCS:CMSZ21}:

\begin{itemize}
    \item[] $\underline{\Est^{\sfPi}_{\varepsilon,\delta}}$:
    \item[] \textbf{Input:} A state $\brho$ on register $\RegH$.
    \item[] \textbf{Setup:} Define $t  =\lceil 2\log(2/\delta)/\varepsilon^2 \rceil$.
    \begin{enumerate}[itemsep=0pt]
        \item Initialize a register $\RegQ$ to the state $\ket{+_R} \coloneqq \frac{1}{2} \ket{\top} + \frac{1}{2} \ket{\bot} + \frac{1}{\sqrt{2\abs{R}}} \sum_{r \in R} \ket{r}$.
        \item\label[step]{step:valest-estimate} Define $M_{\unif} \coloneqq \{\ketbra{+_R} \otimes \Id_{\RegH}, \Id-\ketbra{+_R}\otimes \Id_{\RegH}\}$ and $M_{\win} \coloneqq \{\Pi_{\win}, \Id-\Pi_{\win}\}$ where
        \[
            \Pi_{\win} = \sum_{r \in R} \ketbra{r}_\RegQ \otimes \Pi_r + \ketbra{\top}_\RegQ \otimes \Id.
        \]
        For $i = 1, \ldots, t$:
        \begin{enumerate}[nolistsep]
        \item Measure $M_{\unif}$, obtaining outcome $b_{2i-1} \in \{0,1\}$.
        \item Measure $M_{\win}$, obtaining outcome $b_{2i} \in \{0,1\}$.
        \end{enumerate}
        \item \label[step]{step:valest-disentangle} If $b_{2t} = 1$, skip to  \cref{step:valest-output}. Otherwise, apply $M_\win, M_\unif$ in an alternating fashion until $M_\unif \to 1$, or a further $2t$ measurements have been applied.
        \item \label[step]{step:valest-output} Discard $\RegQ$ and output 
        \[\tilde{p} \coloneqq  2\left(\frac{\sum_{i \in [t]} (1-(b_i-b_{i-1})^2)}{t}\right)-\frac{1}{2}.\]
    \end{enumerate}
\end{itemize}

\begin{remark}
Following~\cite{FOCS:CMSZ21}, the additional symbols $\bot$ and $\top$ are used to ensure that $\Est$ runs in strict (rather than expected) polynomial time. In particular, the outcomes obtained in~\cref{step:valest-estimate} correspond to the success probability of the prover in the following game:
\begin{itemize}
    \item with $1/2$ probability, we sample a random $r \gets R$ and test whether the $\{\Pi_r,\Id-\Pi_r\}$ accepts and
    \item with $1/2$ probability, we sample a random ``challenge'' from $\{\top,\bot\}$, and output success if the challenge is $\top$ and failure if the challenge is $\bot$.
\end{itemize}
If the prover has ``real'' success probability $p = \mathbb{E}_{r \from R}[\Tr(\Pi_r \brho)]$, its success probability in this modified game is $1/4 + p/2 \in [1/4,3/4]$. This rescaling enables running in strict polynomial time and is corrected in the last step of $\Est$ (\cref{step:valest-output}). 
\end{remark}

\paragraph{Recap: repairing the state after measurement.} \cite{FOCS:CMSZ21} use $\Est$ as a subroutine for their ``state repair'' procedure. To describe state repair, we introduce some additional definitions:
\begin{itemize}
    \item Without loss of generality, the procedure $\Est^{\sfPi}_{\varepsilon,\delta}$ can be efficiently implemented by (1) initializing an ancilla register $\RegV$ to $\ket{0}$, (2) applying some unitary to $(\RegH,\RegV)$, (3) performing some projective measurement $\{\Pi'_p\}_p$ on $\RegV$ to obtain the estimate $\widetilde{p}$, and (4) tracing out $\RegV$. We define $\mathsf{CoherentEst}^{\sfPi}_{\varepsilon,\delta}$ to be the unitary applied in step (2).\footnote{In~\cite[Section 4.3]{FOCS:CMSZ21}, $\mathsf{CoherentEst}^{\sfPi}_{\varepsilon,\delta}$ corresponds to the unitary $U_{\mathsf{M}}$.} 
    \item For any $p \in [0,1]$, we define a projection $\Pi_p$ on $(\RegH,\RegV)$ as
    \[ \Pi_p \coloneqq \sum_{p' \geq p} (\mathsf{CoherentEst}^{\sfPi}_{\varepsilon,\delta})^\dagger \Pi'_{p'} (\mathsf{CoherentEst}^{\sfPi}_{\varepsilon,\delta}).\]
    In other words, $\Pi_p$ is the projection that corresponds to running $\mathsf{CoherentEst}^{\sfPi}_{\varepsilon,\delta}$ and obtaining an estimate $p' \geq p$.
\end{itemize}

With these definitions in hand, we now describe the~\cite{FOCS:CMSZ21} repair procedure. State repair is applied to a (normalized) state $\brho \in \bfD(\cH)$ that was just ``disturbed'' by some projective measurement $\{D,\Id-D\}$; we will assume the outcome of the measurement was $D$, i.e., $\brho$ satisfies $\Tr(D \brho) = 1$. Given a target success probability $p$, the procedure will output a ``repaired'' state whose success probability is $p$, assuming that certain conditions on $\brho$ hold.

\begin{itemize}
    \item[] $\underline{\Repair^{\sfPi,D}_{\varepsilon,\delta,p}}$:
    \item[] \textbf{Input:} A state $\brho$ on register $\RegH$ such that $\Tr(D \brho) = 1$.
    \item[] 
    \begin{enumerate}[itemsep=0pt]
        \item Initialize the ancilla register $\RegV$ to $\ket{0}$.
        \item Perform the measurement $A \coloneqq \{\Pi_p,\Id-\Pi_p\}$. If the measurement accepts, skip to~\cref{step:repair-output}.
        \item Define $B \coloneqq \{D_{\RegH} \otimes \ketbra{0}_{\RegV}, \Id_{\RegH,\RegV} - D_{\RegH} \otimes \ketbra{0}_{\RegV}\}$. Perform the measurements $B,A,B,A,\dots$ in alternating fashion until either an $A$ measurement accepts, or $\lceil 1/\sqrt{\delta} \rceil$ total applications of $(B,A)$ have occurred.
        \item\label[step]{step:repair-output} Apply $\mathsf{CoherentEst}^{\sfPi}_{\varepsilon,\delta}$ and discard the $\RegV$ registers.
    \end{enumerate}
\end{itemize}

We refer the reader to~\cite[Section 4.3]{FOCS:CMSZ21} for additional details and intuition behind the procedure.

\paragraph{A one-bit extractor.} Using $\Est$ and $\Repair$, we define a procedure $\BitExtract$ that takes a single copy of a state $\brho$ and tries to obtain many accepting outcomes for many $\{\Pi_r,\Id-\Pi_r\}$ measurements. At a very high level, $\BitExtract$ works by (1) running $\Est$ to estimate the initial success probability of $\brho$ and (2) repeatedly applying $\{\Pi_r,\Id-\Pi_r\}$ for random $r \gets R$ followed by an application of $\Repair$ to restore the success probability.

Formally, the $\BitExtract$ procedure is parameterized by an error tolerance $\gamma>0$, a lower bound $q$ on the success probability, and a runtime parameter $T$.
\begin{itemize}
    \item[] $\underline{\BitExtract_{\gamma,q,T}^{\sfPi}}$:
    \item[] \textbf{Input:} A state $\brho$ on register $\RegH$.
    \item[] \textbf{Setup:} Define $\varepsilon = \gamma/4T$ and $\delta = (\gamma/16T)^2$.
\begin{enumerate}
    \item For $i = 0, \dots, T-1$:
    \begin{enumerate}
        \item Measure $\Est_{\varepsilon,\delta}^{\sfPi} \to p_i$. If $p_i < q$, abort and output failure.
        \item Sample a random $r \from R$ and apply the measurement $\{\Pi_r, \Id-\Pi_r\}$.
        \item If the measurement accepts, set $D = \Pi_{r}$. If the measurement rejects, set $D = \Id - \Pi_{r}$.
        \item Apply $\Repair_{\varepsilon,\delta,p_i}^{\sfPi,D}$.
    \end{enumerate}
    \item Measure $\Est^{\sfPi}_{\varepsilon,\delta} \rightarrow p_T$.
    \item Output success if $p_T \ge q$; otherwise output failure.
\end{enumerate}
\end{itemize}

We require the following guarantee on $\BitExtract$ which we prove using~\cite[Claim 4.14]{FOCS:CMSZ21}.

\begin{claim}
\label{claim:bit-extractor}
Let $p = \mathbb{E}_{r \from R}[\Tr(\Pi_r \brho)]$. Then
\[ \Pr[\BitExtract_{\gamma,q,T}^{\sfPi}(\brho) \rightarrow \emph{success}] \geq p- q-\gamma.\]
\end{claim}
\begin{proof}
If $p_i \ge q + \gamma/2 - i\gamma/2T$ for all $i \in \{0,\dots,T\}$, then the extractor will succeed. We therefore take a union bound over the probability that $p_0 < q + \gamma/2$ and the probability that $p_i < p_{i-1} - \gamma/2T$ for some $i \in [T]$:
\begin{itemize}
    \item By Markov's inequality, the probability that $p_0 < q+\gamma/2$  is at most $1-(p-q-\gamma/2)$.
    \item By \cite[Claim 4.14]{FOCS:CMSZ21}, the probability that $p_i < p_{i-1} - \gamma/2T$ for some $i \in [T]$ is at most $8T\sqrt{\delta} \le \gamma/2$.
\end{itemize}
The claim follows by a union bound over these probabilities.
\end{proof}

\subsection{Description of the knowledge extractor}
\label{subsec:extractor-description}

In this subsection, we describe our knowledge extractor $\KnowledgeExt$. We first establish some useful notation:
\begin{itemize}
    \item Recall the following notation for the tree commitments to PCP proof strings introduced in~\cref{subsec:quantum-merkle,sec:kilian-construction}. The number of blocks (leaves) of the tree commitment is $2^{\beta}$. Each node of the tree is indexed by a string $\ell \in \{0,1\}^{\leq \beta}$, where node $\ell$ is distance $\abs{\ell}$ away from the root. For a challenge $r$, the set $S_r \subseteq \{0,1\}^{\beta}$ is the set of (block) indices that contain the positions the PCP verifier checks on randomness $r$. For any tree node $\ell \in \{0,1\}^{\leq \beta}$, $\Path(\ell)$ is the set of all $\abs{\ell}+1$ prefixes of $\ell$ (including $\ell$ itself and the empty string $\varepsilon$). For $S \subseteq \{0,1\}^{\leq \beta}$, $\Path(S) \coloneqq \bigcup_{\ell \in S} \Path(\ell)$. The tree commitment is the root $\RegC_{\varepsilon}$. For any $\ell \in \{0,1\}^{\le\beta}$, the commitment-decommitment pair $(\RegC_\ell,\RegD_{\ell})$ can be opened by applying $\Com^\dagger$ to obtain the corresponding message-ancilla registers $(\RegM_\ell,\RegW_\ell)$. If $|\ell|< \beta$, $\RegM_{\ell}$ can be split into two $s/2$ qubit registers $(\RegC_{\ell,0},\RegC_{\ell,1})$ (see~\cref{fig:merkle-tree}).
    \item Let $\RegH$ denote the register containing the malicious prover's state after it sends the tree commitment $\RegC_{\varepsilon}$. 
    \item Let $\RegM' \coloneqq (\RegM'_{\ell})_{\ell \in \{0,1\}^{\leq \beta}}$ denote a set of $s$-qubit registers. Looking ahead, these will eventually contain the extracted messages from the tree commitment and the extracted PCP will be the state on the registers $(\RegM'_{\ell})_{\ell \in \{0,1\}^{\beta}}$ corresponding to the leaves.\fermi{this is also wrong} Similarly to $\RegM$, if $|\ell|< \beta$, then $\RegM_{\ell}'$ can also be split into two $s/2$ qubit registers $(\RegC_{\ell,0}',\RegC_{\ell,1}')$.
    \item Let $U_r$ denote the unitary that the malicious prover applies to $\RegH$ to generate its response. We will abuse notation slightly and write $U$ as shorthand for $\{U_r\}_{r \in R}$.\footnote{Even when $R$ is an exponential-size set, the collection of unitaries $\{U_r\}$ can be compactly described by a single unitary $U$ acting on $\ketbra{r} \otimes \brho_{\RegH}$.}
\end{itemize}

Next, we define subroutines $\SwapRecover$ and $\SwapDiff$, which will simplify the description of our extractor. Both of the subroutines are unitaries that can be implemented given oracle access to the prover unitary $U_r$ for some challenge $r$.
\begin{itemize}
    \item $\SwapRecover^{U_r}_{S,E}$ is implemented with oracle access to the prover unitary $U_r$ for some $r$ and is parameterized by two sets of tree indices $E \subseteq \{0,1\}^{\le \beta}$ and $S \subseteq \Path(S_r)$, where $S_r$ is the set of PCP (block) indices for challenge $r$. At any stage of the extraction procedure, $E$ will denote the set of tree indices whose messages have already been extracted (i.e., swapped out) into $\RegM'$. Roughly speaking, $\SwapRecover^{U_r}_{S,E}$ is the unitary that (1) runs the prover unitary $U_r$ to generate its response to challenge $r$, (2) opens the decommitments to reveal the messages corresponding to indices in $S$, and (3) swaps in any messages in $E \cap S$.
    
    In our extraction procedure, calls to $\SwapRecover^{U_r}_{S,E}$ will use $S = \Path(S_r)$, where $S_r$ is the set of indices of PCP blocks containing all the indices that the PCP verifier checks on challenge $r$. However, in order to prove correctness of our extractor, we will need to define hybrid extractors that run $\SwapRecover$ for subsets $S \subseteq \Path(S_r)$.
    
    Formally, $\SwapRecover^{U_r}_{S,E}$ acts on the prover's response register $(\RegD_{\ell})_{\ell \in S}$, the tree commitment root $\RegC_\varepsilon$, and the extracted registers $(\RegM_\ell')_{\ell \in \{0,1\}^{\leq \beta}}$.
    \begin{enumerate}
        \item[] $\underline{\SwapRecover^{U_r}_{S,E}}$:
        \item Apply $U_r$.
        \item For $j = 0,\dots,\beta$:
        \begin{itemize}
            \item[] For each $\ell \in S \cap \{0,1\}^j$:
            \begin{enumerate}
                \item Apply $\Com^{\dagger}$ to $(\RegC_{\ell}, \RegD_{\ell})$ to obtain $(\RegM_\ell,\RegW_\ell)$.\footnote{Technically, there are choices of $S$ for which $\SwapRecover^{U_r}_{S,E}$ is not well-defined (in particular, this may occur if this $\RegC_{\ell}$ has not been revealed in a previous loop). However, we will avoid such choices of $S$.}
                \item\label[step]{step:b2-ESwap} If $\ell \in E$, apply $\SWAP[\RegM_\ell,\RegM'_\ell]$.
                \item If $j < \beta$, split $\RegM_{\ell}$ into two $s/2$-qubit registers $(\RegC_{\ell,0},\RegC_{\ell,1})$.
            \end{enumerate}
        \end{itemize}
    \end{enumerate}
    We emphasize that the outer loop must be performed in the order $j = 0,\dots,\beta$. However, the operations in the inner loop commute for all $\ell \in S \cap \{0,1\}^j$.

    \item We additionally define a unitary $\SwapDiff^{U_r}_{S,E}$ (for the same $U_r,S,E$ as in $\SwapRecover$) as follows:
    \begin{enumerate}
        \item[] $\underline{\SwapDiff^{U_r}_{S,E}}$: 
        \item Apply $\SwapRecover^{U_r}_{S,E}$.
        \item Apply $(\SwapRecover^{U_r}_{S,E \cup S})^\dagger$.
    \end{enumerate}
    In our extraction procedure, if the prover answers correctly on challenge $r$, we will apply $\SwapDiff^{U_r}_{S,E}$ with $S = \Path(S_r)$ in order to swap out any messages at tree indices in $\Path(S_r) \setminus E$. These indices correspond to tree nodes that (a) the prover opens on challenge $r$ and (b) were not extracted in an earlier query.
\end{itemize}
For any set of extracted nodes $E \subseteq \{0,1\}^{\leq \beta}$, we define a set of projections $\sfPi^{U}_E \coloneqq \{\Pi^{U_r}_{r,E}\}_{r \in R}$ as
\[
        \Pi^{U_r}_{r,E} \coloneqq (\SwapRecover^{U_r}_{\Path(S_r),E})^\dagger \left(\Pi_r^{\PCP} \otimes \bigotimes_{\ell\in\Path(S_r)} \ketbra{0}_{\RegW_\ell}\right) \SwapRecover^{U_r}_{\Path(S_r),E},
\]
where $\Pi^{\PCP}_r$ is the projection the verifier applies to $\bigotimes_{j \in S_r} \RegM_j$ to check the proof on random coins $r$ (see \Cref{def:quantum-pcp}; here we leave the dependence on the instance $x$ implicit). In words, the projective measurement $\{\Pi^{U_r}_{r,E}, \Id-\Pi^{U_r}_{r,E}\}$ measures whether the prover gives an accepting response on challenge $r$, where the response is computed by swapping in any previously-extracted messages from the registers $(\RegM'_i)_{i \in E}$ (whenever necessary).

We now define $\Extract^{U}_{\gamma,q,T}$, the primary subroutine in our full extractor $\KnowledgeExt$. $\Extract$ essentially behaves identically to $\BitExtract$, except that whenever it sends a challenge $r \gets R$ and receives openings for messages in $\Path(S_r)$, it swaps out the messages in $\Path(S_r)$ that it has not yet extracted. $\gamma > 0$ is an error tolerance, $q$ is a lower bound on the allowable success probability, and $T$ is a runtime parameter (these parameters will be set when we define $\KnowledgeExt$).

\begin{itemize}
    \item[] $\underline{\Extract^{U}_{\gamma,q,T}}$:
    \item[] \textbf{Input:} A prover state on register $\RegH$ and a tree commitment register $\RegC_{\varepsilon}$.
    \item[] \textbf{Setup:} Define $\varepsilon = \gamma/4T$ and $\delta = (\gamma/16T)^2$. Initialize registers $\RegM' \coloneqq (\RegM'_i)_{i \in [2^{\beta+1}-1]}$.\footnote{We identify each tree node $\ell \in \{0,1\}^{\leq \beta}$ with an integer $i \in [2^{\beta+1}-1]$ following a level-order traversal, i.e., the root is $i = 1$, its children are $2,3$, etc.} Initialize $E = \emptyset$.
\begin{enumerate}
    \item For $t = 0, \dots, T-1$:
    \begin{enumerate}
        \item Measure $\Est_{\varepsilon,\delta}^{\sfPi^U_{E}} \to p_t$. If $p_t < q$, abort and output failure.
        \item Sample a random $r \from R$ and apply the measurement $\{\Pi^{U_r}_{r,E}, \Id-\Pi^{U_r}_{r,E}\}$.
        \item If the measurement accepts, set $D \coloneqq \Pi^{U_r}_{r,E}$. If the measurement rejects, set $D \coloneqq \Id - \Pi^{U_r}_{r,E}$.  Perform the following steps:
        \begin{enumerate}
            \item Apply $\SwapDiff^{U_r}_{\Path(S_r),E}$.
            \item Update $E \leftarrow E \cup \Path(S_r)$. 
        \end{enumerate}
        \item Apply $\Repair_{\varepsilon,\delta,p_t}^{\sfPi^U_E,D}$.
    \end{enumerate}
    \item Measure $\Est_{\varepsilon,\delta}^{\sfPi^U_{E}} \to p_T$. Return ``success'' if $p_T \geq q$ and ``failure'' otherwise. Additionally, output the registers $(\RegM_i')_{i \in [2^{\beta+1}-1] \backslash [2^{\beta}-1]}$.
\end{enumerate}
\end{itemize}

We now define $\KnowledgeExt$.
\begin{enumerate}
        \item[] $\underline{\KnowledgeExt^{\widetilde{P}}(x,1^{\lceil 1/p \rceil},1^{\lceil 1/\gamma \rceil})}$:
        \item Run the malicious prover $\widetilde{P}$ to obtain a tree commitment $\RegC$. 
        \item Let $\btau_{\RegH,\RegC}$ denote the global state on the prover's registers $\RegH$ and the commitment register $\RegC$. Sample a random $T \gets \{1,2,\dots,\lceil 2^{\beta+3}/\gamma^2\rceil \}$ and run $\Extract^{U}_{\gamma/4,q+\gamma/4,T}(\btau_{\RegH,\RegC}) \rightarrow (b,\bpi)$ where $b$ is a bit indicating success/failure. Discard $b$ and output $\bpi$.
        
        % If $\Extract$ succeeded, output $\bpi$, otherwise output $\bot$.
\end{enumerate}

\subsection{Correctness of our extractor}
\label{sec:extractor-success}

In this subsection, we will complete the proof of our main theorem (\cref{theorem:kilian-security}) by showing that $\KnowledgeExt$ satisfies the conditions of~\cref{def:argument-of-knowledge}.

The main technical steps are captured by the following claim about the $\Extract$ subroutine.

\begin{claim}[Success probability of the $\Extract$ subroutine]
\label{claim:extractor-success}
For any state $\btau_{\RegH,\RegC}$ and efficient prover unitary $U$, let $p$ denote the initial success probability, i.e., $p \coloneqq \mathbb{E}_{r \from R}[\Tr(\Pi_{r} \btau)]$ where $\Pi_r \coloneqq \Pi^{U_r}_{r,E =\emptyset}$. If $\SQSC$ is swap-binding, then for any $\gamma = 1/\poly(\lambda)$, any $q < p - \gamma - 1/\poly(\lambda)$, and any $T = \poly(\lambda)$, the subroutine $\Extract^{U}_{\gamma,q,T}$ succeeds with probability at least $p - q - \gamma - \negl(\lambda) = 1/\poly(\lambda)$.
\end{claim}

\paragraph{Proof outline.} At a high level, our proof of~\cref{claim:extractor-success} consists of the following steps:
\begin{enumerate}
    \item First, we write down a sequence of hybrid extraction procedures $\HybExtract[j]$, one for each tree node $[2^{\beta+1}-1]$. $\HybExtract[j]$ is defined similarly to $\Extract$, except that the extractor only swaps out nodes in $[j]$. Thus, $\HybExtract[0]$ will correspond to $\BitExtract$ for which the~\cite{FOCS:CMSZ21} statistical guarantees apply, and $\HybExtract[2^{\beta+1}-1]$ corresponds to $\Extract$.
    \item Next, we show that for all $j$, if an adversary distinguishes between $\HybExtract[j-1]$ and $\HybExtract[j]$, this implies a reduction that violates the \emph{oracle swap binding} security (defined in~\cref{def:oracle-binding}) of $\SQSC$. We refer the reader to~\cref{subsubsec:to-rewinding-difficulties} for discussion on the necessity of oracle swap binding. Crucially, by~\cref{theorem:oracle-binding-security}, oracle swap binding security is \emph{implied} by standard swap binding.
\end{enumerate}

\begin{proof}[Proof of~\cref{claim:extractor-success}]

For $j \in \{0,\dots,2^{\beta+1}-1\}$, define $\HybExtract[j]$ as follows.

\begin{itemize}
    \item[] $\underline{\HybExtract^{U}_{\gamma,q,T}[j]}$ (differences from $\Extract$ highlighted in {\color{red}{red}}):
    \item[] \textbf{Input:} A prover state on register $\RegH$ and a tree commitment register $\RegC_{\varepsilon}$.
    \item[] \textbf{Setup:} Define $\varepsilon = \gamma/4T$ and $\delta = (\gamma/16T)^2$. Initialize registers $\RegM' \coloneqq (\RegM'_i)_{i \in [2^{\beta+1}-1]}$. Initialize $E = \emptyset$.
\begin{enumerate}
    \item For $t = 0, \dots, T-1$:
    \begin{enumerate}
        \item Measure $\Est_{\varepsilon,\delta}^{\sfPi^U_{E}} \to p_t$. If $p_t < q$, abort and output failure.
        \item Sample a random $r \from R$ and apply the measurement $\{\Pi^{U_r}_{r,E}, \Id-\Pi^{U_r}_{r,E}\}$.
        \item If the measurement accepts, set $D \coloneqq \Pi^{U_r}_{r,E}$. If the measurement rejects, set $D \coloneqq \Id - \Pi^{U_r}_{r,E}$.  Perform the following steps:
        \begin{enumerate}
            \item Apply $\SwapDiff^{U_r}_{\Path(S_r) {\color{red}{\cap [j]}}, E}$.
            \item Update $E \gets E \cup (\Path(S_r) {\color{red}{\cap [j]}})$. 
        \end{enumerate}
        \item Apply $\Repair_{\varepsilon,\delta,p_t}^{\sfPi^U_E,D}$.
    \end{enumerate}
    \item Measure $\Est_{\varepsilon,\delta}^{\sfPi^U_{E}} \to p_T$. Return ``success'' if $p_T \geq q$ and ``failure'' otherwise. Additionally, output the registers $(\RegM_i')_{i \in [2^{\beta+1}-1] \backslash [2^{\beta}-1]}$.
\end{enumerate}
\end{itemize}
Observe that:
\begin{itemize}
    \item $\HybExtract^{U}_{\gamma,q,T}[0]$ is equivalent to running $\BitExtract_{\gamma,q,T}^{\sfPi}$ for $\sfPi = \{\Pi_r\}$ (where $\Pi_r \coloneqq \Pi^{U_r}_{r,E =\emptyset}$) and additionally outputting registers $(\RegM_i')_{i \in [2^{\beta+1}-1] \backslash [2^{\beta}-1]}$ set to $0$. 
    \item $\HybExtract^{U}_{\gamma,q,T}[2^{\beta+1}-1]$ is identical to  $\mathsf{Extract}^{U}_{\gamma,q,T}$. 
\end{itemize}
We will prove that for all $j \in [2^{\beta+1}-1]$ and any efficient adversary (specified by some state $\btau_{\RegH,\RegC}$ and a unitary $U$),
\[ \abs{\Pr[\HybExtract^U_{\gamma,q,T}[j-1](\btau) \rightarrow \text{success}] - \Pr[\HybExtract^U_{\gamma,q,T}[j](\btau) \rightarrow \text{success}] } = \negl(\lambda).\]

Suppose otherwise. We will give a reduction $\Reduction[j]$ that wins the oracle swap binding game (\cref{def:oracle-binding}) using the commitment at index $j$. Before we state $\Reduction[j]$, we define the oracle $G_b$ that we provide the reduction in the oracle swap binding game:
\begin{itemize}
    \item Recall from~\cref{def:oracle-binding} that the adversary in the oracle swap binding game has access to a unitary $G_b$ with the form \[
G_b =  (\widehat{\SWAP[\RegM_j, \RegM_j']})^b \cdot \widehat{ \cO} \cdot  (\widehat{\SWAP[\RegM_j, \RegM_j']})^b \Pi + (\Id-\Pi),
\]
    for some unitary $\cO$ chosen by the adversary and register $\RegM_j''$ that the challenger initializes to $\ket{0}$. Recall that $\Pi \coloneqq (\Com) (\Id_{\RegM_j} \otimes \ketbra{0}_{\RegW_j}) (\Com^{\dagger})$ denotes the projection onto valid commitment-decommitment pairs, and that for any operation $U$, we write $\widehat{U} \coloneqq (\Com) U (\Com^\dagger)$.
    \item Let $\RegB$ be a one-qubit ancilla register. We will instantiate $G_b$ for the following choice of $\cO \coloneqq \cO_j$ for $\cO_j$ defined as
    \[
    \cO_j = \sum_{\substack{E \subseteq [2^{\beta+1}-1] \\ r \in R}} \ketbra{E}_\RegE \otimes \ketbra{r}_\RegQ \otimes \mathcal{O}_{r,E,j},\] 
    where the unitary $\mathcal{O}_{r,E,j}$ is defined as
    \[\mathcal{O}_{r,E,j} \coloneqq \SwapRecover_{\Path(S_r) \cap [j],E}^{U_r} \cdot (X_{\RegB} \otimes \Pi^{U_r}_{r,E,j} + \Id_{\RegB} \otimes (\Id-\Pi^{U_r}_{r,E,j})) \cdot (\SwapRecover_{\Path(S_r) \cap [j],E}^{U_r})^\dagger\]
    where
    \[
        \Pi^{U_r}_{r,E,j} \coloneqq (\SwapRecover^{U_r}_{\Path(S_r),E})^\dagger \left(\Pi_r^{\PCP} \otimes \bigotimes_{\ell\in\Path(S_r) \backslash \{j\}} \ketbra{0}_{\RegW_\ell}\right) \SwapRecover^{U_r}_{\Path(S_r),E}.
    \]
\end{itemize}

This operation $\cO$ is carefully defined to enable the reduction in a swap-binding security game to (a) run the (swap-augmented) prover on any classical challenge $r$ and (b) run $\Est$ and $\Repair$ on the prover (which requires running the swap-augmented prover on a superposition of challenges $r$), \emph{even after the challenger is given the $\RegC_j$ register}. Formally, $\cO$ allows the adversary to pick any choice of $r$ and $E$ and apply a unitary we call $\cO_{r,E,j}$. $\cO_{r,E,j}$ is defined so that both of the following are true:
\begin{itemize}
    \item The unitary \[ (\SwapRecover_{\Path(S_r) \cap [j-1],E}^{U_r})^\dagger \cdot (\widehat{\cO_{r,E,j}} \Pi + (\Id-\Pi)) \cdot (\SwapRecover_{\Path(S_r) \cap [j-1],E}^{U_r})\]
is equivalent to checking (and recording the outcome onto $\RegB$) whether the swap-augmented prover in $\HybExtract[j-1]$ would make the verifier accept when run on challenge $r$ when the extracted set is $E$.
    \item The unitary 
    \begin{align*}
        (\SwapRecover_{\Path(S_r) \cap [j-1],E}^{U_r})^\dagger \cdot (\widehat{\SWAP[\RegM_j, \RegM_j']}) \cdot (\widehat{\cO_{r,E,j}} \Pi + (\Id - \Pi)) \\
        \cdot (\widehat{\SWAP[\RegM_j, \RegM_j']}) \cdot (\SwapRecover_{\Path(S_r) \cap [j-1],E}^{U_r})
    \end{align*}
is equivalent to checking (and recording the outcome onto $\RegB$) whether the swap-augmented prover in $\HybExtract[j]$ would make the verifier accept when run on challenge $r$ when the extracted set is $E$.
\end{itemize}

For $\cO$ to be a valid oracle for the oracle swap binding game for the commitment at index $j$, it must be implemented as a circuit that acts only on $\RegM_j$ and other registers independent of the $j$th commitment (i.e., $\cO$ should not be defined with respect to $\RegC_j,\RegD_j$, or $\RegW_j$). An implementation of $\cO$ that naively implements all the unitaries that appear in the description of $\cO_{r,E,j}$ and $\Pi^{U_r}_{r,E,j}$ would \emph{not} satisfy this guarantee, since, e.g., $(\SwapRecover_{\Path(S_r) \cap [j],E}^{U_r})^\dagger$ in the definition $\cO_{r,E,j}$ acts on registers $\RegC_j,\RegD_j$ if $j \in \Path(S_r)$. However, we observe that to implement $\cO_{r,E,j}$, it suffices to implement $(\SwapRecover_{\Path(S_r),E}^{U_r}) \cdot (\SwapRecover_{\Path(S_r) \cap [j],E}^{U_r})^\dagger$ and its inverse, and this operation can be implemented to satisfy the syntactic requirement. In particular, $(\SwapRecover_{\Path(S_r),E}^{U_r}) \cdot (\SwapRecover_{\Path(S_r) \cap [j],E}^{U_r})^\dagger$ is equivalent to the following procedure:

\begin{enumerate}
        \item[] For each $\ell \in \Path(S_r) \backslash [j]$ in increasing order:
        \begin{enumerate}
            \item Apply $\Com^{\dagger}$ to $(\RegC_{\ell}, \RegD_{\ell})$ to obtain $(\RegM_\ell,\RegW_\ell)$.
            \item If $\ell \in E$, apply $\SWAP[\RegM_\ell,\RegM'_\ell]$.
        \end{enumerate}
    \end{enumerate}

% Roughly speaking, the unitary operation $\cO_{r,E,j}$ coherently performs the measurement $\{\Pi_{r,E}^{U_r},\Id-\Pi_{r,E}^{U_r}\}$ and records the outcome onto $\RegB$, except it expects the commitments along $\Path(S_r) \cap [j]$ to already be opened via $\SwapRecover_{\Path(S_r) \cap [j],E}$ and the commitment at index $j$ verified, so it only opens the remaining commitments along $\Path(S_r) \backslash [j]$ while swapping extracted messages back in. Then, it (coherently) measures whether the PCP check passes and all commitments along $\Path(S_r)$, except for the commitment at index $j$ (if along the path), are valid. Finally, it uncomputes. Because the operation $\cO$ can only act nontrivially on register $\RegM_j$ and other registers independent of the commitment scheme at index $j$, it is a valid oracle for the oracle swap binding game.  

% The operation $\cO$ allows the reduction in the oracle swap binding game to apply $\cO_{r,E,j}$ for any choice of $r$ and $E$. This enables the reduction to (a) run the (swap-augmented) prover on any classical challenge $r$ and (b) run $\Est$ and $\Repair$ on the prover, which requires running the (swap-augmented) prover on a superposition of challenges $r$.

At a high level, $\Reduction[j]$ works as follows:
\begin{itemize}
    \item It runs $\HybExtract[j-1]$ as usual until the first time the adversary gives an accepting response containing a decommitment to node $j$ (up to this point, there is no difference between running $\HybExtract[j-1]$ or $\HybExtract[j]$). In particular, it simply implements operations involving the projections $\Pi^{U_r}_{r,E}$ as normal.
    \item Once the adversary gives an accepting response containing a decommitment to node $j$, the reduction sends the valid commitment-decommitment pair to the oracle swap binding challenger corresponding to the commitment at node $j$. Since the challenger does not return the commitment register, the reduction can no longer implement operations involving the projections $\Pi^{U_r}_{r,E}$ properly. However, using the oracle $G_b$ defined as above, the reduction can implement operations involving a projection $\Pi^{U_r,G_b}_{r,E}$ that we define as
    \[\Pi_{r,E}^{U_r,G_b} := \bra{1}_\RegB \bra{E}_\RegE \bra{r}_\RegQ \widetilde{G}_b \ket{0}_\RegB \ket{E}_\RegE \ket{r}_\RegQ, \]
    where
    \[\widetilde{G}_b := (\SwapRecover_{\Path(S_r) \cap [j-1],E}^{U_r})^\dagger \cdot G_b \cdot (\SwapRecover_{\Path(S_r) \cap [j-1],E}^{U_r}).\]
    By following the definitions, it can be verified that $\Pi_{r,E}^{U_r,G_b}$ has exactly the same behavior as $\Pi_{r,E}^{U_r}$, except that operations requiring the committed message in the node $j$ commitment are implemented with the oracle $G_b$.
\end{itemize}

To simplify the description of $\Reduction[j]$, we introduce a ``flag'' bit $\flag \in \{0,1\}$ that is initialized to $0$ and flips to $1$ after the reduction interacts with the oracle swap binding challenger. Define
\begin{equation}
\Pi_{r,E,\flag}^{U_r,G_b} \coloneqq
    \begin{cases}
        \Pi^{U_r}_{r,E} & \text{if } \flag = 0\\
        \Pi^{U_r,G_b}_{r,E} & \text{if } \flag = 1,
    \end{cases}
\end{equation}
and let $\sfPi^{U,G_b}_{E,\flag} \coloneqq \{\Pi_{r,E,\flag}^{U_r,G_b}\}_{r \in R}$.

We also need to modify the definitions of $\SwapRecover$ and $\SwapDiff$, since the original syntax of these subroutines will not be applicable after the reduction sends commitment $\RegC_j$ to the oracle swap binding challenger. 
\begin{itemize}
    \item We define $\SwapRecover^{U_r}_{S,E,j}$, to behave the same as $\SwapRecover^{U_r}_{S,E}$ except that it does not swap out $\RegC_j$. Formally, $\SwapRecover^{U_r}_{S,E,j}$ does the following:
    \begin{itemize}
        \item[] $\underline{\SwapRecover^{U_r}_{S,E,j}}$:
            \item If $j =1$ (i.e., $j$ is the root), apply $U_r$.
            \item Else:
            \begin{enumerate}
            \item Apply $\SwapRecover_{S \cap [j-1] \setminus \{\sibling(j),\parent(j)\},E}$, where $\sibling(j)$ is the sibling and $\parent(j)$ is the parent of node $j$. 
            \item Apply $\Com_{\parent(j)}^\dagger$. If $\parent(j) \in E$, apply $\SWAP[\RegC_{\sibling(j)},\RegC_{\sibling(j)}']$. 
            \item If $\sibling(j) \in S \cap [j-1]$, apply $\Com_{\sibling(j)}^\dagger$. If $\sibling(j) \in E$, apply $\SWAP[\RegM_{\sibling(j)},\RegM_{\sibling(j)}']$.
        \end{enumerate}
    \end{itemize}
    \item Next, we define
    \begin{enumerate}
    \item[] $\underline{\SwapDiff^{U_r}_{S,E,j}}$:
    \item Apply $\SwapRecover^{U_r}_{S,E,j}$.
    \item Apply $(\SwapRecover^{U_r}_{S,E \cup S,j})^\dagger$. 
\end{enumerate}
    \item Finally, we define $\SwapDiff^{U_r}_{S,E,j,\flag}$ as
    \begin{equation}
\SwapDiff^{U_r}_{S,E,j,\flag} \coloneqq
    \begin{cases}
        \SwapDiff^{U_r}_{S,E} & \text{if } \flag = 0\\
        \SwapDiff^{U_r}_{S,E,j} & \text{if } \flag = 1.
    \end{cases}
\end{equation}
\end{itemize}

We now state the reduction $\Reduction[j]$ which plays the oracle swap-binding game and outputs a bit $b'$. The reduction has the same input and setup as $\HybExtract[j-1]$ and $\HybExtract[j]$ and uses the same values of $\varepsilon$ and $\delta$. 

\begin{itemize}
    \item[] $\underline{\Reduction[j]}$:
\begin{enumerate}
    \item Initialize $\flag = 0$.
    \item For $t = 0, \dots, T-1$:
    \begin{enumerate}
        \item Measure $\Est_{\varepsilon,\delta}^{\sfPi_{E,\flag}^{U,G_b}} \to p_t$. If $p_t < q$, abort. 
        \item Sample a random $r \from R$ and apply the measurement $\{\Pi^{U_r,G_b}_{r,E,\flag}, \Id-\Pi^{U_r,G_b}_{r,E,\flag}\}$.
        \item If the measurement rejects, set $D \coloneqq \Id - \Pi^{U_r,G_b}_{r,E,\flag}$. If the measurement accepts, set $D \coloneqq \Pi^{U_r,G_b}_{r,E,\flag}$ and perform the following steps:
        \begin{enumerate}
            \item If $j \in \Path(S_r)$ and $\flag=0$:
            \begin{enumerate}
                \item Apply $\SwapRecover^{U_r}_{\Path(S_r) \cap [j-1],E}$.
                \item Apply $\SWAP[\RegC_{j},\RegC_{j}']$ and send $\RegD_{j}, \RegC_{j}'$ to the oracle swap binding challenger. 
                \item The challenger returns $\RegD_{j}$.
                \item Update $E \gets E \cup (\Path(S_r) \cap [j])$.
                \item Apply $(\SwapRecover^{U_r}_{\Path(S_r) \cap [j-1],E,j})^\dagger$.
                \item Set $\flag = 1$.
            \end{enumerate}
            \item Otherwise:
            \begin{enumerate}
                \item Apply $\SwapDiff^{U_r}_{\Path(S_r)\cap [j],E,j,\flag}$.
                \item Update $E \gets E \cup (\Path(S_r)\cap [j])$.
            \end{enumerate}
        \end{enumerate}
        \item Apply $\Repair_{\varepsilon,\delta,p_t}^{\sfPi_{E,\flag}^{U,G_b},D}$.
    \end{enumerate}
    \item Measure $\Est_{\varepsilon,\delta}^{\sfPi_{E,\flag}^{U,G_b}} \to p_T$. Guess $b' = 0$ if $p_T \geq q$ and guess $b ' = 1$ otherwise.
\end{enumerate}
\end{itemize}

If the oracle swap binding challenger's bit is $b = 0$, then $\Pr[\Reduction[j] \rightarrow 0]$ is equal to $\Pr[\HybExtract[j-1] \text { succeeds}]$. If the challenger's bit is $b =1$, then $\Pr[\Reduction[j] \rightarrow 0]$ is equal to $\Pr[\HybExtract[j] \text { succeeds}]$. Thus, by oracle swap binding security, we have
\[
\Pr[\HybExtract[j] \text{ succeeds}] \ge \Pr[\HybExtract[j-1] \text{ succeeds}] - \negl(\lambda).
\]
Since $2^\beta = \poly(\lambda)$, it follows from~\cref{claim:bit-extractor} that
\begin{align*}
\Pr[\mathsf{Extract}^{U}_{\gamma,q,T}]
&\geq p-q-\gamma-\negl(\lambda). \qedhere
\end{align*}
\end{proof}

\paragraph{Why does $\KnowledgeExt$ produce a valid PCP?}
\cref{claim:extractor-success} shows that after running $\Extract$, the swap-augmented prover still has high success probability. To complete the proof of~\cref{theorem:kilian-security}, it remains to show that when $T$ is sampled uniformly at random from $\{1,2,\dots,\lceil 2^{\beta+3}/\gamma^2\rceil\}$ in the $\KnowledgeExt$ procedure, with noticeable probability the subroutine $\Extract^{U}_{\gamma/4,q+\gamma/4,T}$ outputs a good PCP string.\footnote{A similar claim is proved in~\cite{FOCS:CMSZ21,FOCS:LomMaSpo22}, but their argument is specialized for extracting classical proof strings and works by taking a union bound over all possible proofs. We therefore require a different proof for this final step.}

\begin{claim}
\label{claim:random-stop}
For any state $\btau_{\RegH,\RegC}$ and efficient prover unitary $U$, let $p$ denote the initial success probability, i.e., $p \coloneqq \mathbb{E}_{r \from R}[\Tr(\Pi_{r} \btau)]$ where $\Pi_r \coloneqq \Pi^{U_r}_{r,E =\emptyset}$. If $\SQSC$ is swap-binding, then for any $\gamma = 1/\poly(\lambda)$ and any $q < p - \gamma - 1/\poly(\lambda)$, with probability at least $p-q-\gamma-\negl(\lambda) \ge 1/\poly(\lambda)$ the procedure $\KnowledgeExt^{\widetilde{P}}(x,1^{\lceil 1/p \rceil},1^{\lceil 1/\gamma \rceil})$ outputs an extracted proof $\bpi$ such that
\[
\mathbb{E}_{r \from R}[\Tr(\Pi_r^{\PCP} \bpi)] \ge q-\negl(\lambda).
\]
\end{claim}
\begin{proof}
Let $\brho_{\RegM',\RegC,\RegH}$ be the joint state of both the prover and extractor after running $\KnowledgeExt$ (not to be confused with the state $\bpi$ output by $\KnowledgeExt$). By \cref{claim:extractor-success}, with probability at least $p-q-\gamma/2-\negl(\lambda)$ the subroutine $\Extract^{U}_{\gamma/4,q+\gamma/4,T}$ outputs success and
\[
\mathbb{E}_{r \from R}[\Tr(\Pi_{r,E}^{U_r} \brho_{\RegM',\RegC,\RegH})] \ge q+\gamma/4.
\]
We divide the $\RegM'$ register into $\RegM' = (\RegM'_{\mathrm{internal}},\RegM'_{\mathrm{leaves}})$ so that the extracted proof corresponds to $\RegM'_{\mathrm{leaves}}$.
For any state $\brho_{\RegM',\RegC,\RegH}$, the extracted proof $\bpi = \Tr_{\RegC,\RegH,\RegM'_{\mathrm{internal}}}(\brho_{\RegM',\RegC,\RegH})$ satisfies
\begin{align*}
    \mathbb{E}_{r \from R}[\Tr(\Pi_r^{\PCP} \bpi)] &= \Pr_{r \from R}[\Pi_r^{\PCP} \text{ accepts } \bpi] \\
    &\ge \Pr_{r \from R}[\Pi_r^{\PCP} \text{ accepts } \bpi \text{ and } S_r \subseteq E] \\
    &\ge \Pr_{r \from R}[\Pi_{r,E}^{U_r} \text{ accepts } \brho_{\RegM',\RegC,\RegH} \text{ and } S_r \subseteq E] \\
    &= \Pr_{r \from R}[\Pi_{r,E}^{U_r} \text{ accepts } \brho_{\RegM',\RegC,\RegH}] - \Pr_{r \from R}[\Pi_{r,E}^{U_r} \text{ accepts } \brho_{\RegM',\RegC,\RegH} \text{ and } S_r \not\subseteq E] \\
    &= \mathbb{E}_{r \from R}[\Tr(\Pi_{r,E}^{U_r} \brho_{\RegM',\RegC,\RegH})] - \Pr_{r \from R}[\{\Pi_{r,E}^{U_r},\Id-\Pi_{r,E}^{U_r}\} \text{ accepts on } \brho_{\RegM',\RegC,\RegH} \text{ and } S_r \not\subseteq E] \\
    &\ge q + \gamma/4 - \Pr_{r \from R}[\Pi_{r,E}^{U_r} \text{ accepts } \brho_{\RegM',\RegC,\RegH} \text{ and } S_r \not\subseteq E].
\end{align*}
The second inequality follows because the measurement $\{\Pi_{r,E}^{U_r},\Id-\Pi_{r,E}^{U_r}\}$ checks that both the measurement $\{\Pi_r^{\PCP}, \Id - \Pi_r^{\PCP}\}$ accepts $\bpi$ and all commitment-decommitment pairs along the root to leaf path are valid.

For any state $\brho$ on registers $(\RegM',\RegC,\RegH)$, define
\[
x(\brho) = \Pr_{r \from R}[\Pi_{r,E}^{U_r} \text{ accepts } \brho \text{ and } S_r \not\subseteq E].
\]
We will show that $\Pr_{\brho \gets \KnowledgeExt}[x(\brho) > \gamma/4] \le \gamma/2$, which will complete the proof.

Recall that $\KnowledgeExt^{\widetilde{P}}(x,1^{\lceil 1/p \rceil},1^{\lceil 1/\gamma \rceil})$ runs $\Extract_{q+\gamma/4,\gamma/4,T}$, stopping at a random time $t \from [T]$. Consider an execution of $\Extract_{q+\gamma/4,\gamma/4,T+1}$, and let $e_t$ denote the event that both $b_t=1$ (i.e., the measurement $\{\Pi_{r,E_t}^{U_r},\Id-\Pi_{r,E_t}^{U_r}\}$ accepts for the randomly chosen $r$ at step $t$) and $S_t \not\subseteq E_t$, where $S_t$ and $E_t$ are the sets $S_r$ and $E$, respectively, at step $t$. Observe that
\begin{align} \label{eq:kilian-corollary-key-observation}
\mathop{\mathbb{E}}_{\brho \gets \KnowledgeExt}[x(\brho)] &= \mathop{\mathbb{E}}_{t \from [T]} \mathop{\mathbb{E}}_{\brho \gets \Extract_{q+\gamma/4,\gamma/4,t}}[x(\brho)] \\
&= \mathop{\mathbb{E}}_{t \from [T]} \Pr_{\substack{\brho \gets \Extract_{q+\gamma/4,\gamma/4,t}\\r \from R}}[\Pi_{r,E}^{U_r} \text{ accepts } \brho \text{ and } S_r \not\subseteq E] \\
&= \mathop{\mathbb{E}}_{t \from \{2,\dots,T+1\}} \Pr[e_t]
\end{align}
% because for all $t \in [T]$, the following two events occur with equal probability
% \begin{itemize}
%     \item $\KnowledgeExt^{\widetilde{P}}(x,1^{\lceil 1/p \rceil},1^{\lceil 1/\gamma \rceil})$ stops at time $t \in [T]$, leaving a joint state $\brho$, and moreover a subsequent measurement of $\{\Pi_{r,E}^{U_r},\Id-\Pi_{r,E}^{U_r}\}$ for a random $r \gets R$ accepts.
%     \item 
% \end{itemize}

% $\KnowledgeExt^{\widetilde{P}}(x,1^{\lceil 1/p \rceil},1^{\lceil 1/\gamma \rceil})$ runs $\Extract_{q+\gamma/4,\gamma/4,T}$, stopping at a random time $t \from [T]$, and thus the event that $[\Pi_{r,E}^{U_r} \text{ accepts } \brho \text{ and } S_r \not\subseteq E \text{ for }r \gets R]$ is e be interpreted as the probability that 

% . 
Since there are only $2^\beta$ blocks in the proof and each time $e_t$ happens the extractor will swap out a new block, we can have at most $2^\beta$ events $e_t$, i.e.,
\[
\frac{1}{T} \sum_{t=2}^{T+1} \mathbf{1}_{e_t} \le \frac{2^\beta}{T},
\]
where $\mathbf{1}_{e_t}$ is the indicator for the event $e_t$.

Taking expectations and applying \cref{eq:kilian-corollary-key-observation},
\[
\mathop{\mathbb{E}}_{\brho}[x(\brho)] \le \frac{2^\beta}{T}.
\]
Finally, by Markov's inequality we have
\[
\Pr_{\brho}[x(\brho) > \gamma/4] \le \frac{2^{\beta+2}}{T \cdot \gamma} \le \frac{\gamma}{2}. \qedhere
\]
\end{proof}

\newpage

\section{Computationally sound quantum sigma protocols}
\label{sec:quantum-sigma}

In this section, we use our techniques to prove computational soundness for \emph{quantum sigma protocols} instantiated with our hiding and binding QSCs. This gives a quantum analogue of a standard template for building classical sigma protocols from any \emph{zero-knowledge PCP} (zk-PCP) and any hiding-binding commitment. While quantum sigma protocols were previously considered by Broadbent and Grilo~\cite{FOCS:BroGri20}, the results in this section differ from~\cite{FOCS:BroGri20} in several respects:
\begin{itemize}
    \item As discussed in~\cref{subsec:related}, the \cite{FOCS:BroGri20} analysis of computational soundness was incomplete because it was missing the rewinding analysis.
    \item As an additional contribution, we also place their protocol in a general framework using our new abstraction of hiding and binding QSCs. We also consider a slightly more general definition of quantum zero-knowledge PCPs than~\cite{FOCS:BroGri20}.
\end{itemize}

We begin by defining quantum zk-PCPs, a quantum analogue of classical zk-PCPs.\footnote{We refer the reader to Ishai's blog post \cite{Ishai20} for background on classical zk-PCPs.} Since we are not concerned with succinctness in this section, it suffices to consider quantum PCPs with soundness $1-1/\poly(n)$, which are known for all of $\QMA$~\cite{FOCS:BroGri20}. 
\begin{remark}
Our definition of quantum zk-PCPs differs slightly from the the definition of \emph{locally simulatable proofs} due to \cite{FOCS:BroGri20}. In particular, we require the simulator generates the mixed state corresponding to the honest PCP verifier's view, rather than (a classical description of) the reduced density matrix \emph{for all subsets} of qubits up to a certain size as in~\cite{FOCS:BroGri20}. Our definition captures both the locally simulatable proofs of~\cite{FOCS:BroGri20} as well as classical zk-PCPs such as~\cite{FOCS:GolMicWig86} 3-coloring.\footnote{The 3-coloring zk-PCP does not satisfy the~\cite{FOCS:BroGri20} definition of a locally simulatable proof, since it is not possible to generate the verifier's view for pairs of vertices that do not share an edge.}
\end{remark}

\begin{definition}[Quantum zero-knowledge PCPs]
\label{def:quantum-zk-pcp}
A quantum zero-knowledge PCP for a language $L$ is parameterized by a completeness parameter $c$, soundness $s$, proof length $m$, randomness complexity $\ell$, and query complexity $q$. We require the following properties: 
\begin{itemize}
    \item (Efficient verification) There is a classical $\poly(n)$ time procedure that takes as input $x \in \{0,1\}^n$ and $r \in \{0,1\}^{\ell}$ and outputs the description of circuit for implementing a $q$-qubit projective measurement  $\{\Pi^{\PCP}_{x,r},\Id-\Pi^{\PCP}_{x,r}\}$, which acts on a state of size $m$. 
    \item (Honest-verifier zero knowledge) Let $Q_r \subset [m]$ denote the size-$q$ subset of indices that $\{\Pi^{\PCP}_{x,r},\Id-\Pi^{\PCP}_{x,r}\}$ checks. There exists an efficient quantum algorithm $\mathsf{PCPSim}(x)$ that outputs a mixed state of the form \[\frac{1}{2^\ell}\sum_{r \in \{0,1\}^\ell} \ketbra{r} \otimes \brho_{Q_r},\]
    such that for any $x \in L$, there exists an $m$-qubit locally simulatable proof $\bpi$ satisfying completeness, i.e.,
    \[ \displaystyle \mathop{\mathbb{E}}_{r \gets \{0,1\}^{\ell}} \Tr(\Pi_{x,r}^{\PCP} \bpi) \geq c,\]
    and the following zero-knowledge property:
    \[
        \norm{ \frac{1}{2^{\ell}} \sum_{r \in \{0,1\}^{\ell}} \ketbra{r} \otimes \Tr_{\overline{Q_r}}(\bpi) - \frac{1}{2^{\ell}} \sum_{r \in \{0,1\}^{\ell}} \ketbra{r} \otimes \brho_{Q_r} }_1 = \negl(n).
    \]
    \item (Soundness) If $x \not\in L$, then for any $m$-qubit state $\bpi$,
    \[ \displaystyle \mathop{\mathbb{E}}_{r \gets \{0,1\}^{\ell}} \Tr(\Pi_{x,r}^{\PCP} \bpi) \leq s.\]
\end{itemize}
\end{definition}

\begin{lemma}[\cite{FOCS:BroGri20}]  \label{lemma:BG-simulatable-qma}
    Every language in $\QMA$ has a zero-knowledge quantum PCP with completeness $c(n) \geq 1-\negl(n)$ and soundness $s(n) \leq 1-1/\poly(n)$.  
\end{lemma}

We remark that because our definition allows mixed state PCPs, any classical zk-PCP is also a quantum zk-PCP. In particular, this means the zk-PCPs underlying the~\cite{FOCS:GolMicWig86} 3-coloring protocol and the~\cite{STOC:IKOS07} MPC-in-the-head protocols fall into our framework.

\paragraph{A generic quantum sigma protocol.} We now describe a generic template for quantum sigma protocols. The protocol is parameterized by a choice of zk-PCP, $\ZKPCP$, and quantum commitment scheme, $\QSC$, where:
\begin{itemize}
    \item $\ZKPCP$ is quantum zk-PCP for a language $L$, and
    \item $\QSC$ is a hiding-binding QSC for one-qubit messages.
\end{itemize}
For any $\ZKPCP$ and $\QSC$ satisfying the above, we denote the resulting quantum sigma protocol by $\QSigma[\ZKPCP,\QSC]$.\footnote{Broadbent and Grilo~\cite{FOCS:BroGri20} presented this protocol in the special case that the QSC was defined by the folklore construction (see \cref{sec:folklore-construction}).}

\begin{itemize}
    \item[] $\QSigma[\ZKPCP,\QSC]$:
    \begin{enumerate}
        \item[] \textbf{Prover input:} $x$ and a corresponding zk-PCP $\brho$ on $m=\poly(n)$ qubits $\bigotimes_{i=1}^{m} \RegM_i$.
        \item[] \textbf{Verifier input:} $x$.
        \item The prover applies $\Com_{\QSC}(\RegM_i) \rightarrow (\RegC_i, \RegD_i)$ and sends $\RegC_i$ for all $i \in [m]$ to the verifier.
        \item The verifier sends random coins $r$ to the prover.
        \item The prover sends $\RegD_i$ for all $i \in Q_r$ to the verifier.
        \item The verifier applies $\Com_{\QSC}^{\dagger}(\RegC_i, \RegD_i) \rightarrow (\RegM_i, \RegW_i)$ and measures $\RegW_i$ with $\{\ketbra{0},  \Id - \ketbra{0}\}$ for all $i \in Q_r$; it aborts if the outcome is not $\ketbra{0}$. Finally, the verifier measures $\{\Pi_r^{\PCP}, \Id - \Pi_r^{\PCP}\}$. It accepts if the measurement accepts.
    \end{enumerate}
\end{itemize}

\subsection{Argument of knowledge}

\begin{theorem} \label{theorem:qsigma-soundness}
    Let $\ZKPCP$ be a probabilistically checkable proof for a language $L$. Then $\QSigma[\ZKPCP, \QSC]$ is a computational (resp. statistical) proof of knowledge for $L$ with respect to the $\ZKPCP$ verifier (\cref{def:argument-of-knowledge}) if $\QSC$ satisfies computational (resp. statistical) binding. 
\end{theorem}

\begin{proof}
    The proof strategy is nearly identical to the proof of~\cref{theorem:kilian-security} and we will avoid repeating it in full. The only difference arises from the fact that instead of a Merkle tree of commitments to the $\PCP$, we have independent commitments for every block of the $\ZKPCP$ proof string. Thus, we can use the same proof as~\cref{theorem:kilian-security} with the following (syntactic) simplifications:
    
    \begin{itemize}
        \item There are now $m$ explicit commitments instead of $2^{\beta+1}-1$ (implicit) commitments. Correspondingly, the $\HybExtract_j$ procedures only need to be defined for $j \in \{0,1,\dots,m\}$.
        \item For any $S \subseteq [m]$, replace $\Path(S)$ with $S$.
        \item For any $S \subseteq [m]$ and $E \subseteq [m]$, replace $\SwapRecover^{U_r}_{S,E}$ with the algorithm that applies $\Com_{\QSC}^{\dagger}$ to $(\RegC_i, \RegD_i)$ to obtain $(\RegM_i, \RegW_i)$ for all $i \in S$. If $i \in E$, apply $\SWAP[\RegM_i, \RegM'_i]$ as well. 
        \item For any $S \subseteq [m]$, $E \subseteq [m]$, $j \in [m]$, replace $\SwapRecover^{U_r}_{S,E,j}$ with $\SwapRecover^{U_r}_{S \cap [j-1], E}$ (for $\SwapRecover^{U_r}_{S,E}$ as defined in the previous bullet).\qedhere
    \end{itemize}
\end{proof}

\subsection{Zero knowledge}
\label{subsec:zk}

We now prove that our quantum sigma protocol is zero knowledge against malicious verifiers whenever the size of the challenge space $\{0,1\}^\ell$ for the underlying $\ZKPCP$ is at most $\poly(\lambda)$. This is a standard application of Watrous's zero-knowledge rewinding lemma, which we recall below. We refer the reader to \cite[Section~3]{Watrous09} for a definition of zero-knowledge against quantum attacks.\footnote{Even though Watrous's paper focuses on post-quantum security of classical protocols, Watrous's zero-knowledge definitions only refer to the admissible super-operators (CPTP maps) induced by the interaction with the prover (or the simulator) and thus apply equally well to quantum-communication protocols.} 

\begin{lemma}[\cite{Watrous09}]
\label{lemma:watrous}
Let $Q$ be a quantum circuit that acts on $n+k$ qubits, where the first $n$ qubits may take an arbitrary state $\ket{\psi}$ as input and the remaining $k$ qubits are initially set to the state $\ket*{0^k}$. For any $n$-qubit state $\ket{\psi}$, define $p(\psi) \in [0,1]$ and (normalized) states $\ket{\phi_0(\psi)},\ket{\phi_1(\psi)}$ so that
\begin{align}
    Q \ket{\psi}\ket*{0^k} = \sqrt{p(\psi)} \ket{0}\ket{\phi_0(\psi)} + \sqrt{1-p(\psi)} \ket{1}\ket{\phi_1(\psi)}\label{eq:watrous}.
\end{align}

Suppose that there exists $p_0,q \in (0,1)$ and $\varepsilon \in (0,1/2)$ such that for all $n$-qubit states $\ket{\psi}$, 
\begin{enumerate}
    \item $\abs{p(\psi) - q} < \varepsilon$,
    \item $p_0(1-p_0) \leq q(1-q)$, and
    \item $p_0 \leq p(\psi).$
\end{enumerate}
Then there exists a general quantum circuit $R$ of size
\[\abs{R} = O\left(\frac{\log(1/\varepsilon)\abs{Q}}{p_0(1-p_0)}\right) \]
such that, for every $n$-qubit state $\ket{\psi}$, the output $\brho(\psi)$ of $R$ satisfies
\[ \bra{\phi_0(\psi)} \brho(\psi) \ket{\phi_0(\psi)} \geq 1- 16\varepsilon \frac{\log^2(1/\varepsilon)}{p_0^2(1-p_0)^2}.\]
\end{lemma}

\begin{theorem}
    Suppose $\ZKPCP$ has randomness complexity $\ell = O(\log \lambda)$. Then if $\QSC$ satisfies computational (resp. statistical) hiding, $\QSigma[\ZKPCP, \QSC]$ satisfies computational (resp. statistical) zero knowledge.
\end{theorem}

\begin{proof}
For any instance $x$ and security parameter $\lambda$, we will define a unitary $U_{x,\lambda}$ acting on the following registers:
\begin{itemize}
    \item An $m$-qubit register $\RegM = (\RegM_1,\dots,\RegM_m)$ where each $\RegM_i$ is a one-qubit register.
    \item A $dm$-qubit register $\RegW = (\RegW_1,\dots,\RegW_m)$ where each $\RegW_i$ is a $d(\lambda)$-qubit register.
    \item A pair of $\ell$-qubit registers $\RegR, \RegR'$. When we run the simulator, $\RegR$ will be a register containing uniform PCP verifier randomness $r$ and $\RegR'$ will contain the string output by the malicious verifier $\widetilde{V}$.
    \item An ancilla register $\RegS$ corresponding to the workspace of $\mathsf{PCPSim}_x$.
    \item A register $\RegV$ corresponding to the auxiliary quantum input for the malicious verifier $\widetilde{V}$.
    \item A one-qubit register $\RegB$.
\end{itemize}
All registers are initialized to $\ket{0}$ except $\RegV$, which is initialized to the auxiliary quantum input of $\widetilde{V}$. 

For a size-$q$ subset $Q_r \subset [m]$ and $q$-qubit state $\brho_{Q_r}$, let $\brho_{Q_r,\overline{Q}_r}$ denote the $m$-qubit state whose reduced density matrix on indices $Q_r$ is $\brho_{Q_r}$ and is all-zero everywhere else. We will assume that $\mathsf{PCPSim}_x$ is a unitary acting on $\RegR,\RegM,\RegS$ that, when applied to the all-zero state, produces a state whose reduced density matrix on $(\RegR,\RegM)$ has the form \[ \frac{1}{2^\ell}\sum_{r \in \{0,1\}^\ell} \ketbra{r}_{\RegR} \otimes \brho_{Q_r,\overline{Q}_r}. \]

We now define the unitary $U_{x,\lambda}$ to perform the following steps:
\begin{enumerate}
    \item Apply the unitary $\mathsf{PCPSim}_x$ to $(\RegR,\RegM,\RegS)$. 
    \item For each index $i \in [m]$, apply $\Com_{\QSC}$ to $(\RegM_i,\RegW_i)$ to obtain $(\RegC_i,\RegD_i)$.
    \item Apply the malicious verifier unitary $\widetilde{V}$ on $(\RegC_1,\dots,\RegC_m,\RegV,\RegR')$.
    \item Apply the unitary 
    \[\Id_{\RegB} \otimes \sum_{r \in \{0,1\}^\ell} \ketbra{r}_{\RegR} \otimes  \ketbra{r}_{\RegR'} + X_{\RegB} \otimes \sum_{r,s \in \{0,1\}^\ell, r\neq s} \ketbra{r}_{\RegR} \otimes \ketbra{s}_{\RegR'}, \]
    which applies the bit-flip operator $X$ to $\RegB$ controlled on the contents of $\RegR$ and $\RegR'$ being different (in the standard basis).
\end{enumerate}

Fix any $x \in L$. We will show that there exists a negligible function $\mu(\lambda)$ and a constant $\lambda_0$ such that for all $\lambda \geq \lambda_0$, the unitary $U_{x,\lambda}$ satisfies the conditions of~\cref{lemma:watrous} for $p_0 = 1/2^{\ell+1}$, $q = 1/2^\ell$, and $\varepsilon = \mu(\lambda)$ where:
\begin{itemize}
    \item $\ket{\psi}$ in~\cref{lemma:watrous} corresponds to the auxiliary input state on $\RegV$,
    \item $\ket{0^k}$ in~\cref{lemma:watrous} corresponds to the initial state on registers $(\RegM,\RegW,\RegR,\RegR',\RegS,\RegB)$,
    \item and the left-most qubit on the right-hand side of~\cref{eq:watrous} corresponds to $\RegB$.
\end{itemize}
Suppose otherwise. Then there exists a family of states $\{\ket{\psi_\lambda}_{\RegV}\}_{\lambda \in \mathbb{N}}$ and a constant $c>0$ such that for infinitely many $\lambda$, we have $\abs{p(\psi_\lambda) - 1/2^\ell} > 1/\lambda^c$. We will give a reduction breaking hiding (\cref{def:quantum-hiding}) of the $m$-qubit QSC corresponding to the $m$ different $1$-qubit QSCs in the protocol; this can be turned into a reduction breaking the hiding of one of the $1$-qubit QSCs by a standard hybrid argument. The reduction works as follows:
\begin{enumerate}
    \item Run $\mathsf{PCPSim}_x$ to obtain the mixed state \[\frac{1}{2^\ell}\sum_{r \in \{0,1\}^\ell} \ketbra{r}_{\RegR} \otimes (\brho_{Q_r,\overline{Q}_r})_{\RegM}.\]
    \item Send $(\RegM_1,\dots,\RegM_m)$ to the QSC hiding challenger.
    \item The challenger returns commitments $(\RegC_1,\dots,\RegC_m)$. Run $\widetilde{V}$ on $\RegC$ to obtain a challenge register $\RegR'$.
    \item Measure $\RegR$ and $\RegR'$ in the standard basis. If the two outcomes are different, output a random bit $b'$. If the outcome is the same, then:
    \begin{itemize}
        \item if $p(\psi_\lambda) > 1/2^\ell + 1/\lambda^c$, guess $b' = 0$, and
        \item if $p(\psi_\lambda) < 1/2^\ell - 1/\lambda^c$, guess $b' = 1$.
    \end{itemize}
\end{enumerate}
If the challenger's bit is $b = 1$ in the hiding experiment (\cref{def:quantum-hiding}), then the challenger returns commitments $(\RegC_1,\dots,\RegC_m)$ where the underlying messages are all-zero states. In this case, the probability that the reduction obtains the same challenge when it measures $\RegR$ and $\RegR'$ is $1/2^\ell$, since $\RegR$ is a uniform mixture over $r \in \{0,1\}^\ell$ independent of the malicious verifier's view.

If the challenger's bit is $b = 0$, then the probability that the reduction obtains the same challenge when it measures $\RegR$ and $\RegR'$ is $p(\psi_\lambda)$. An elementary calculation shows that the reduction guesses $b$ with probability $1/2 + O(1/\lambda^c)$, which violates the hiding of the QSC.

Thus, by applying~\cref{lemma:watrous}, we can efficiently generate a state $\brho(\psi_\lambda)$ such that $\brho(\psi_\lambda)$ is within $1-\negl(\lambda)$ trace distance from $\phi_0(\psi_\lambda)$, the (normalized) state obtained by applying $U_{x,\lambda}$ on $\ket{\psi_\lambda}_{\RegV} \ket{0}_{\RegM,\RegW,\RegR,\RegR',\RegS,\RegB}$ and post-selecting on $\RegB = 0$. With one additional step, we can generate the simulated view: controlled on the value $r$ in $\RegR'$, generate a response register $\RegZ$ containing $\RegD_i$ for each $i \in Q_r$. The resulting simulated malicious verifier view is $(\RegC_1,\dots,\RegC_m,\RegR',\RegZ,\RegV)$. We omit the formal details of the rest of the proof, which are a direct analogue of Watrous's proof~\cite{Watrous09}. In short, the hiding of the commitments together with the zero-knowledge property of the quantum zk-PCP implies that the simulated view is indistinguishable from a view in which the $m$ committed qubits are a valid zk-PCP state $\bpi$. The latter corresponds to the view of the malicious verifier in the interaction with the honest prover initialized with $\bpi$. 
\end{proof}

\newpage

\ifsubmission
\else
\section*{Acknowledgments}

We thank Prabhanjan Ananth, James Bartusek, Andrea Coladangelo, Yael Kalai, Willy Quach, Nicholas Spooner, Vinod Vaikuntanathan and Umesh Vazirani for helpful discussions. We thank Luowen Qian for his detailed feedback on an early draft of this manuscript. We thank Alex Lombardi for many insightful conversations throughout the course of this project and countless suggestions that greatly improved the paper.

\fi

\newpage

\bibliographystyle{alpha}
\bibliography{abbrev3,crypto,references}

\appendix

\newpage

\section{Removing interaction from any QSC} 
\label{sec:non-interactive}

In this section, we define the syntax for an interactive QSC and define two security definitions for interactive QSCs called ``honest swap binding'' and ``honest hiding.'' We show that interactive QSCs satisfying these security definitions can be compiled to \emph{non-interactive} QSCs that are swap-binding (\Cref{def:swap-binding}) and hiding (\Cref{def:quantum-hiding}) in \Cref{theorem:interactive-to-noninteractive}. This transformation is similar to the round-collapse compiler of~\cite{AC:Yan22}. 

We begin by defining syntax for interactive QSCs.
\begin{definition}[Interactive QSC syntax]
\label{def:interactive_qsc_syntax}
An interactive quantum state commitment $\QSC$ between a quantum sender $S$ and a quantum receiver $R$ consists of an interactive commit phase $\langle S, R \rangle_{\mathrm{Com}}$ and unitary $\Recover$.
\begin{itemize}
    \item (Commitment phase) The sender commits to a state on message register $\RegM$ by engaging in a quantum interactive protocol $\langle S, R \rangle_{\mathsf{Com}}$ with the receiver. 
    \item (Opening phase) To open the commitment, the sender engages in a quantum interactive protocol $\langle S, R \rangle_{\mathsf{Open}}$ with the receiver. At the end of the protocol, the receiver applies a binary outcome projective measurement to decide whether to accept or reject the opening. If it accepts, the receiver recovers the originally committed message on some register $\RegM$. 
\end{itemize}
\end{definition}

For completeness, we require that the map induced by performing an honest commitment followed by an honest opening is the identity map on the message space $\cM$. 

\begin{definition}[Completeness]
For any quantum state commitment $\QSC$, let $\Phi_{\QSC}: \bfS(\cM) \rightarrow \bfS(\cM)$ be the CPTP map induced by (1) performing an honest commitment to $\RegM$ (2) performing an honest opening, and (3) replacing $\RegM$ with $\ket{0}$ if the receiver rejects. We say that a QSC scheme $\QSC$ is complete if \[\norm{\Phi_{\QSC}-\Id_{\bfS(\cM)}}_{\diamond} = 0.\] 
\end{definition}
In particular, this definition implies that if the sender honestly commits to a state $\ket{\psi}$ and later opens honestly, the receiver will accept with probability $1$ and recovers $\ket{\psi}$. It is easy to check that the syntax for a non-interactive QSC (\Cref{def:noninteractive_qsc}) is a special case of an interactive QSC. 

Next, we define honest hiding and honest swap binding.

\paragraph{Honest hiding and honest swap binding.} We define a security experiment \\$\texttt{HonHideExpt}_{\QSC,A,b}(\lambda)$ parameterized by a quantum commitment scheme $\QSC$, an interactive adversary $A$, a challenge bit $b \in \{0,1\}$, and a security parameter $\lambda$.

\begin{definition}[Honest hiding] \label{def:honest-hiding}
For a quantum commitment scheme $\QSC$, an interactive adversary $A$, a challenge bit $b \in \{0,1\}$, and a security parameter $\lambda$, define a security experiment $\emph{\texttt{HonHideExpt}}_{\QSC,A,b}(\lambda)$ as follows.

\indent $\emph{\texttt{HonHideExpt}}_{\QSC,A,b}(\lambda)$:
\begin{enumerate}
    \item The adversary $A$ prepares a message $\RegM$ and sends it to the challenger. 
    \item The challenger locally executes an entire commit phase between an honest sender and an honest receiver, where the honest sender commits to $\RegM$ if $b = 0$, or $\ket{0}$ if $b = 1$. The challenger sends the internal state $\RegC$ of the honest receiver to $A$.
    \item The output of the experiment is $b' \gets A$. 
\end{enumerate}

$\QSC$ is computationally (resp. statistically) honest-hiding if there exists a negligible function $\mu(\lambda)$ such that for all polynomial-time (resp. unbounded-time) quantum interactive adversaries $A$,
\[ \Pr_{b \gets \{0,1\}} [\emph{\texttt{HonHideExpt}}_{\QSC,A,b}(\lambda)=b] \leq \frac{1}{2} + \mu(\lambda). \]
\end{definition}

\begin{definition}[Honest swap binding] \label{def:honest-binding}
For a quantum commitment scheme $\QSC$, an interactive adversary $A$, a challenge bit $b \in \{0,1\}$, and a security parameter $\lambda$, define the security experiment $\emph{\texttt{HonSwapBindExpt}}_{\QSC,A,b}(\lambda)$ as follows.

\indent $\emph{\texttt{HonSwapBindExpt}}_{\QSC,A,b}(\lambda)$:
\begin{enumerate}
    \item The adversary $A$ sends the challenger a register $\RegM$. 
    \item The challenger locally executes an entire commit phase between an honest sender and an honest receiver, where the honest sender commits to $\RegM$ if $b = 0$, or $\ket{0}$ if $b =1$. The challenger sends the internal state $\RegD$ of the honest sender to $A$.
    \item The output of the experiment is $b' \gets A$. 
\end{enumerate}

$\QSC$ is computationally (resp. statistically) honest swap binding if there exists a negligible function $\mu(\lambda)$ such that for all polynomial-time (resp. unbounded-time) quantum interactive adversaries $A$,
\[ \Pr_{b \gets \{0,1\}} [\emph{\texttt{HonSwapBindExpt}}_{\QSC,A,b}(\lambda)=b] \leq \frac{1}{2} + \mu(\lambda). \]
\end{definition}

\paragraph{Compiling an interactive QSC into a non-interactive QSC.}

Consider the following non-interactive commitment scheme, compiled from a (potentially) interactive quantum commitment scheme $\QSC$:
\begin{itemize}
    \item Suppose the commit phase of $\mathsf{QSC}$ is $k$ rounds, where in round $i \in [k]$, the receiver applies a unitary $R_i$, sends a register $\RegX_i$ to the sender, who applies a unitary $S_i$ and sends back a register $\RegY_i$.
    \item Set 
    \begin{align} 
        \label{eq:appendix-compiled-U_com}
        \Com = S_k R_k S_{k-1} R_{k-1} \cdots S_1 R_1.
    \end{align}
\end{itemize}

\begin{theorem} \label{theorem:interactive-to-noninteractive}
    Let $\QSC$ be computationally (resp. statistical) honest swap binding and statistically (resp. computationally) honest hiding. Then the non-interactive commitment scheme $\Com$ defined by \Cref{eq:appendix-compiled-U_com} is computationally (resp. statistical) swap binding and statistically (resp. computationally) hiding.
\end{theorem}
\begin{proof}
    The honest-hiding experiment for the interactive scheme is equivalent to the hiding experiment for the non-interactive scheme from the adversary's point of view because the challenger locally executes the entire commit phase. 
    
    Similarly, by \Cref{property:hiding-binding-duality}, the swap binding experiment for non-interactive commitments can be equivalently phrased by the adversary sending a message register $\RegM$, and receiving an honest decommitment to it in register $\RegD$. Thus, the honest-binding experiment for the interactive QSC is equivalent to the swap binding experiment for the non-interactive QSC. 
\end{proof}

\newpage

\section{Equivalences between binding definitions}
\label{sec:equivalences}

\subsection{Pauli binding}
\label{subsec:pauli-binding}

In this subsection, we present an alternative binding definition for QSCs we call \emph{Pauli binding}, which applies to non-interactive and interactive QSCs (\Cref{def:interactive_qsc_syntax}). Pauli binding is a generalization of sum-binding (see~\cite{EC:Unruh16}) in which the adversary picks an arbitrary Pauli operator $P$, receives a random bit $b$ from the challenger specifying either the $+1$ or $-1$ eigenspace of $P$, and wins if it can open to a message (that is measured to be) in that eigenspace. We will prove that this notion is equivalent to swap binding for non-interactive QSCs.

\begin{definition}[Pauli binding] \label{def:pauli-binding}
For a quantum commitment scheme $\QSC$, an interactive adversary $A$, and a security parameter $\lambda$, define a security experiment \\$\emph{\texttt{PauliBindExpt}}_{\QSC,A}(\lambda)$ as follows.

\indent $\emph{\texttt{PauliBindExpt}}_{\QSC,A}(\lambda)$:
\begin{enumerate}
    \item The adversary $A$ (acting as a malicious sender) engages in the commit phase of $\QSC$ with the challenger (acting as the receiver). Then, $A$ sends the description of an $n$-qubit Pauli operator $P$ to the challenger. 
    \item The challenger samples a random $b \gets \{-1,1\}$ and sends $b$ to the adversary. 
    \item $A$ engages in the opening phase of $\QSC$ with the challenger. If the challenger rejects the opening as invalid, then the experiment aborts and outputs $\fail$. 
    \item Next, the challenger measures $\{P^+, P^-\}$ on the opened message register $\RegM$. It outputs $\win$ if the outcome is $b$, and otherwise outputs $\fail$.
\end{enumerate}

$\QSC$ is computationally (resp. statistically) Pauli binding if there exists a negligible function $\mu(\lambda)$ such that for all polynomial-time (resp. unbounded-time) quantum interactive adversaries $A$,
\[ \Pr [\emph{\texttt{PauliBindExpt}}_{\QSC,A}(\lambda) \rightarrow \win] \leq \frac{1}{2} + \mu(\lambda). \]
\end{definition}

\begin{remark}
    Strictly speaking, Pauli binding assumes the messages have ``qubit structure,'' i.e., the dimension of the Hilbert space must be a power of $2$. It is possible to generalize this definition to qudits where $d$ is $\poly(\lambda)$, but it is not immediately obvious how to extend this definition to message spaces of arbitrary dimension. Thus, we only show an equivalence with swap-binding when the message space has qubit structure. We note that the swap binding definition is agnostic to the dimension of the message space.
\end{remark}

\begin{theorem}
    A non-interactive quantum state commitment scheme $\Com$ is swap binding if and only if it is Pauli binding. 
\end{theorem}
\begin{proof}
    We first show that if $\Com$ is Pauli binding, then it is also swap binding. Suppose that an adversary $A$ distinguishes between the $b=0$ and $b=1$ worlds of the swap binding game with advantage $\varepsilon$. Consider the following three hybrids:
    \begin{itemize}
        \item $H_0$: In this hybrid, an adversary $A$ sends a commitment-decommitment pair $(\RegC, \RegD)$ to the challenger. The challenger applies $\Com^{\dagger}$, measures $\{\ketbra{0}, \Id - \ketbra{0}\}$ on register $\RegW$ (and aborts if the measurement rejects), applies $\Com$, and sends $\RegD$ to the adversary. 
        \item $H_1$: This hybrid is the same as $H_0$, except before applying $\Com$, the challenger applies a uniformly random $n$-qubit Pauli $P$ to $\RegM$ (thus making it maximally mixed).
        \item $H_2$: This hybrid is the same as $H_0$, except before applying $\Com$, the challenger replaces the contents of $\RegM$ with $\ket{0}$, i.e., initializes ancillary register $\RegM'$ to $\ket{0}$ and applies $\SWAP[\RegM, \RegM']$. 
    \end{itemize}
    Hybrids $H_0$ and $H_2$ are exactly the $b=0$ and $b=1$ worlds of the swap binding experiment for $\Com$, so $A$ distinguishes either $H_0$ and $H_1$ or $H_1$ and $H_2$ with advantage $\varepsilon/2$. Let's assume that $A$ distinguishes $H_0$ and $H_1$ with $\varepsilon/2$ advantage as the proof for $H_1$ and $H_2$ follows almost identically. 
    
    We purify the actions of the adversary $A$, so that right after sending $(\RegC, \RegD)$ to the challenger, $A$ and the challenger jointly hold a pure quantum state $\ket{\psi}_{\RegC \RegD \RegR}$, where $\RegR$ is held by $A$. Then after receiving $\RegD$ back from the challenger, $A$ measures a projector $\widetilde{\Pi}$ on $(\RegD, \RegR)$ to determine its output for the swap binding game. Recall that for any operator $O$, we define the notation $\widehat{O} \coloneqq (\Com) O (\Com^{\dagger})$. For any Pauli $P$ acting on register $\RegM$ define the quantities
    \[
        \mathsf{DistAdv}_P \coloneqq \abs{ \norm{\widetilde{\Pi}  \widehat{P} \widehat{\ketbra{0}}_{\RegW} \ket{\psi}}^2 - \norm{\widetilde{\Pi} \widehat{\ketbra{0}}_{\RegW} \ket{\psi}}^2 } 
    \]
    and 
    \[
        \mathsf{MapAdv}_{P} \coloneqq \norm{\widehat{P^-} \widehat{\ketbra{0}}_{\RegW} (\Id - 2 \widetilde{\Pi}) \widehat{P^+} \widehat{\ketbra{0}}_{\RegW}  \ket{\psi}}^2 .
    \]
    Then
    \begin{align*}
        \mathbb{E}_{P \gets \{I,X,Y,Z\}^{\otimes n}}[ \mathsf{MapAdv}_P ] &\geq \mathbb{E}_{P} \left[ (\mathsf{DistAdv}_P)^2 \right] \\
        &\geq \left( \mathbb{E}_{P} \left[ \mathsf{DistAdv}_P \right] \right)^2 \\
        &= (\varepsilon/2)^2,
    \end{align*}
    where the first inequality is by \cref{lemma:duality}, \cref{item:dist-to-map}, the second inequality is Jensen's inequality, and the last equality is because $A$ distinguishes $H_0$ and $H_1$ with advantage $\varepsilon/2$. 
    
    Then consider the Pauli binding adversary $A'$ that does the following:
    \begin{enumerate}
        \item Receive commitment-decommitment registers $(\RegC, \RegD)$ from $A$.
        \item Sample a random $n$-qubit Pauli $P$ and measure $\{\widehat{P^+} \widehat{\ketbra{0}}_{\RegW}, \Id - \widehat{P^+} \widehat{\ketbra{0}}_{\RegW}\}$ on $(\RegC, \RegD)$. Then:
            \begin{itemize}
                \item If the measurement accepts, send register $\RegC$ and Pauli $P$ to the challenger. Receive $b \in \{-1,1\}$ from the challenger. If $b=-1$, then apply $(\Id - 2\widetilde{\Pi})$, which can be done by coherently applying $A$, applying $Z$ to its output, and applying $A$ in reverse. If $b=1$, do nothing to the state. Then, send register $\RegD$ to the challenger.
                \item If the measurement rejects, prepare another commitment to a $+1$ eigenstate of $P$, i.e., initialize $\RegM'$ to a $+1$ eigenstate of $P$ and $\RegW'$ to $\ket{0}$, apply $\Com$ to $(\RegM', \RegW')$ (resulting in $(\RegC', \RegD')$), and send $\RegC'$ and $P$ to the challenger. Receive $b$ from the challenger. Then, send decommitment register $\RegD'$ to the challenger. 
            \end{itemize}
    \end{enumerate}
    
    We analyze the success probability of $A'$ in the Pauli binding game. Let $\win$ indicate the event in which $A'$ wins the Pauli binding game and let $\accept$ indicate the event that the measurement $\{\widehat{P^+} \widehat{\ketbra{0}}_{\RegW}, \Id - \widehat{P^+} \widehat{\ketbra{0}}_{\RegW}\}$ accepts. Then
    \begin{align*}
        \Pr[\win] &= \Pr[\win \wedge \accept \wedge (b=1)] + \Pr[\win \wedge \accept \wedge (b=-1)] + \Pr[\win \wedge \bar{\accept}] \\
        &\geq \frac{1}{2} \Pr[\accept] + \frac{1}{2} \mathbb{E}[\mathsf{MapAdv}_P] +  \frac{1}{2} (1 - \Pr[\accept]) \\
        &\geq \frac{1}{2} + \frac{\varepsilon^2}{8} .
    \end{align*}
    Thus $A'$ wins the Pauli binding game with noticeable advantage.

    We now show that swap binding implies Pauli binding. Let $A$ be an adversary that wins the Pauli binding game for $\Com$ with probability $1/2 + \varepsilon$ using Pauli operator $P$. Let us purify its actions so that it sends the $\RegC$ register of a pure quantum state $\ket{\psi}_{\RegC \RegD \RegR}$ to the challenger, and applies unitary operations $U_{-}$ and $U_{+}$ upon receiving $b=-1$ and $b=1$, respectively. Then letting $U \coloneqq U_{-}U_{+}^{\dagger}$ and $\ket{\phi} \coloneqq U_+ \ket{\psi}$, adversary $A$ winning the Pauli binding game with probability $1/2 + \varepsilon$ is equivalent to:
    \begin{align}
        1 + 2\varepsilon &\leq \norm{\widehat{P^+} \widehat{\ketbra{0}}_{\RegW} U_{+} \ket{\psi}}^2 + \norm{\widehat{P^-} \widehat{\ketbra{0}}_{\RegW} U_{-} \ket{\psi}}^2 \\
        &= \norm{\widehat{P^+} \widehat{\ketbra{0}}_{\RegW} \ket{\phi}}^2 + \norm{\widehat{P^-} \widehat{\ketbra{0}}_{\RegW} U \ket{\phi}}^2 . \label{eq:bounding-alpha-a-b}
    \end{align}
    
    Let us define (sub-normalized) quantum states
    \begin{itemize}
        \item $\ket{\alpha} \coloneqq \widehat{P^+} \widehat{\ketbra{0}}_{\RegW} \ket{\phi}$.
        \item $\ket{a} \coloneqq \widehat{P^-} \widehat{\ketbra{0}}_{\RegW} U \widehat{P^+} \widehat{\ketbra{0}}_{\RegW}  \ket{\phi}$.
        \item $\ket{b} \coloneqq \widehat{P^-} \widehat{\ketbra{0}}_{\RegW} U (\Id - \widehat{P^+} \widehat{\ketbra{0}}_{\RegW})  \ket{\phi}$ .
        \item $\ket{c} \coloneqq (\Id - \widehat{P^+} \widehat{\ketbra{0}}_{\RegW})  \ket{\phi}$
    \end{itemize}
    Then 
    \begin{align*}
        1 + 2\varepsilon &\leq \norm{\ket{\alpha}}^2 + \norm{\ket{a} + \ket{b}}^2 \\
        &\leq \norm{\ket{\alpha}}^2 + (\norm{\ket{a}} + \norm{\ket{b}})^2 \\
        &= \norm{\ket{\alpha}}^2 + \norm{\ket{a}}^2 + 2 \norm{\ket{a}} \norm{\ket{b}} + \norm{\ket{b}}^2 \\ 
        &\leq \norm{\ket{\alpha}}^2 + \norm{\ket{a}}^2 + 2 \norm{\ket{a}} \norm{\ket{b}} + \norm{\ket{c}}^2 \\
        &\leq 1 + \norm{\ket{a}}^2 + 2 \norm{\ket{a}} \norm{\ket{b}} \\
        &\leq 1 + \norm{\ket{a}}^2 + 2 \norm{\ket{a}} \\
        &\leq 1 + 3 \norm{\ket{a}} ,
    \end{align*}
    where the first inequality is a restatement of \Cref{eq:bounding-alpha-a-b}, the second inequality is triangle inequality, the third inequality is because $\norm{\ket{b}} \leq \norm{\ket{c}}$, the fourth inequality is because $\ket{\alpha} \perp \ket{c}$ and $\norm{\ket{\alpha} + \ket{c}}^2 \leq 1$, the fifth inequality is because $\norm{\ket{b}} \leq 1$, and the last inequality is because $\norm{\ket{a}} \leq 1$. 
    
    Thus
    \begin{align*}
        (2\varepsilon/3)^2 \leq \norm{\ket{a}}^2 = \norm{ \widehat{P^-} \widehat{\ketbra{0}}_{\RegW} U \widehat{P^+} \widehat{\ketbra{0}}_{\RegW}  \ket{\phi} }^2.
    \end{align*}
    By \cref{lemma:duality}, \cref{item:map-to-dist}, if we define $\Pi \coloneqq \ctl_{\RegB} \mh U \ketbra{+}_{\RegB} \ctl_{\RegB} \mh U^{\dagger}$ and $\ket{\gamma} \coloneqq \ctl_{\RegB} \mh U (\ket{+}_{\RegB} \otimes (\widehat{P^+} \widehat{\ketbra{0}}_{\RegW}) \ket{\phi} )$, then
    \begin{align}
        \label{eq:dist-advantage-P}
        \abs{\norm{\Pi \widehat{\ketbra{0}}_{\RegW} \ket{\gamma}}^2 - \norm{\Pi \widehat{P} \widehat{\ketbra{0}}_{\RegW} \ket{\gamma}}^2 } \geq 2\varepsilon^2 / 9.
    \end{align}
    In words, the measurement $\{\Pi, \Id-\Pi\}$ distinguishes whether $\widehat{P}$ is applied or not to $\widehat{\ketbra{0}}_{\RegW} \ket{\gamma}$. 
    
    Then consider the following three hybrids:
    \begin{itemize}
        \item $H_0$: In this hybrid, the adversary sends commitment and decommitment registers $(\RegC, \RegD)$ to the challenger. The challenger measures $\{\widehat{\ketbra{0}}_{\RegW}, \Id - \widehat{\ketbra{0}}_{\RegW}\}$. If the measurement rejects, then abort. Otherwise, send $\RegD$ back to the adversary.
        \item $H_1$: This hybrid is the same as $H_0$, except before sending $\RegD$ back to the adversary, the challenger initializes a new register $\RegE$ (of the same dimension as $\RegM$) to $\ket{0}$, then applies $\widehat{\SWAP}[\RegM, \RegE]$.
        \item $H_2$: This is the same as $H_0$, except before sending $\RegD$ back to the adversary, the challenger applies $\widehat{P}$ to $(\RegC, \RegD)$.
    \end{itemize}
    
    \Cref{eq:dist-advantage-P} implies that the adversary $A'$ that does the following achieves a distinguishing advantage $2\varepsilon^2 / 9$ between hybrids $H_0$ and $H_2$:
    \begin{enumerate}
        \item Prepare the state $\widehat{\ketbra{0}}_{\RegW} \ket{\gamma}$ from $\ket{\psi}_{\RegC \RegD \RegR}$ by applying $U_{+}$, then measuring $\{\widehat{P^+} \widehat{\ketbra{0}}_{\RegW}, \Id - \widehat{P^+} \widehat{\ketbra{0}}_{\RegW}\}$, initializing a qubit register $\RegB$ to $\ket{+}$, applying $\ctl_{\RegB} \mh U$, and measuring $\{\widehat{\ketbra{0}}_{\RegW}, \Id - \widehat{\ketbra{0}}_{\RegW}\}$. 
        \begin{itemize}
            \item If both measurements accept, then the resulting state is (normalized) $\ket{\gamma}$. Send $\RegC$ and $\RegD$ registers to the challenger, and receive $\RegD$ back from the challenger. Then output the result of measuring $\{\Pi, \Id - \Pi\}$.
            \item If either of the measurements reject, just output a random bit at the end of the hybrids. That is, prepare any valid commitment and decommitment registers $(\RegC', \RegD')$ by initializing new registers $(\RegM', \RegW')$ to $\ket{0}$ and applying $\Com$ to the registers. Send $(\RegC', \RegD')$ to the challenger, receive $\RegD'$, then output a random bit. 
        \end{itemize}
    \end{enumerate}
    
    Thus $A'$ achieves a distinguishing advantage $\varepsilon^2 / 9$ between hybrids $H_0$ and $H_1$ or $H_1$ and $H_2$. If $A'$ distinguishes $H_0$ and $H_1$, then it wins the swap binding game with the same advantage by definition of the swap binding game. If $A'$ distinguishes $H_1$ and $H_2$, then a modified $A'$ that applies $\widehat{P}$ to the commitment-decommitment registers right before sending them to the challenger will win the swap binding game with the same advantage.

\end{proof}

\begin{remark}
    We remark that a nearly identical proof to the above can be used to show that \emph{collapse-binding} for non-interactive QBCs is equivalent to a restricted version of Pauli binding where the adversary is only allowed to send Pauli operators from $\{\Id, Z\}^{\otimes n}$. A very similar definitional equivalence (for collapse-binding) was also proved in a concurrent and independent work of~\cite{EPRINT:DalSpo22}. \cite{EPRINT:DalSpo22} show that collapse binding is equivalent to ``chosen-bit binding,'' which corresponds to a version of our Pauli binding definition where the operators are from the set $\{Z_1,\dots,Z_n\}$. 
\end{remark}

\subsection{A statistical extraction-based definition}

In the classical setting, a defining feature of statistically binding commitments is that the commitment string information-theoretically fixes the committed message. Equivalently, there is an inefficient procedure to \emph{extract} the committed message from the commitment (up to negligible error). In the quantum setting, \cite{C:AnaQiaYue22} proposed defining statistical binding for quantum commitments to \emph{classical messages} (QBCs) in a similar fashion: a QBC is statistically binding if, given the commitment register $\RegC$, it is possible to extract the committed message bit (in an appropriate sense).

In this subsection, we present (a) a statistical extraction definition for commitments to \emph{quantum messages} and (b) a proof that this definition is equivalent to statistical swap binding. We emphasize that this alternative characterization of swap binding only applies in the \emph{statistical} binding setting. To state the statistical extraction definition, recall the following notation, where $\Com$ is the commitment unitary for a non-interactive quantum state commitment: 
\begin{itemize}
    \item $\widehat{\SWAP[\RegM, \RegM']} \coloneqq (\Com) \SWAP[\RegM, \RegM'] (\Com^\dagger)$ and
    \item $\Pi \coloneqq (\Com) (\Id_{\RegM} \otimes \ketbra{0}_{\RegW}) (\Com^\dagger)$.
\end{itemize}

\begin{definition}[Statistical extractability]
\label{def:statistical-binding-extraction}
A non-interactive quantum commitment scheme $\QSC$ with commitment unitary $\Com$ is \emph{statistically extractable} if there exists a CPTP map $\Ext : (\RegC,\RegM') \to (\RegC,\RegM')$ such that, for all states $\brho \in \bfD(\cC \otimes \cD \otimes \cH)$ and all unitaries $U$ acting only on $(\RegD,\RegH)$,
\begin{align*}
\norm{(\Pi_{(\RegC,\RegD)} \circ U_{(\RegD,\RegH)} \circ \Ext_{(\RegC,\RegM')})(\brho \otimes \ketbra{0}_{\RegM'}) - (\widehat{\SWAP[\RegM,\RegM']} \circ \Pi_{(\RegC,\RegD)} \circ U_{(\RegD,\RegH)})(\brho \otimes \ketbra{0}_{\RegM'})}_1 \\
\le \negl(\lambda).
\end{align*}
\end{definition}

The fact that statistical extractability is equivalent to statistical binding follows from the ``information-disturbance tradeoff'' \cite[Theorem 3]{KSW08}, restated in \Cref{importedtheorem:ksw}. To state the theorem, we require the notion of a \emph{complementary channel}. For any CPTP map $\Phi : \bfS(\cX) \rightarrow \bfS(\cY)$, Stinespring's dilation thoerem \cite{Stinespring1955PositiveFO} implies that there exists a unitary $U$ that maps registers $(\RegX,\RegW)$ to $(\RegY,\RegZ)$ (where $\RegW$ and $\RegZ$ are defined implicitly by $U$) such that $\Phi(\brho) = \Tr_{\RegZ}(U(\brho \otimes \ketbra{0}_{\RegW}))$\footnote{Here we use the notation that $U(\brho) = U \brho U^\dagger$ for a unitary $U$ because the proof of the main theorem in this subsection is stated in terms of quantum channels.}. We say that $\Phi^c$ is a \emph{complementary channel} to $\Phi$ if there exists a unitary $U$ such that $\Phi^c(\rho) = \Tr_{\RegY}(U(\brho \otimes \ketbra{0}_{\RegW}))$ and $\Phi(\rho) = \Tr_{\RegZ}(U(\brho \otimes \ketbra{0}_{\RegW}))$.

% For any choice of quantum channel $\Phi$ and corresponding dilation unitary $U$, we define the complementary channel to $\Phi$ with respect to $U$ to be $\Phi^c$ which has the action of $\Phi^c(\brho) = \Tr_{\RegY}(U(\brho \otimes \ketbra{0}_{\RegW})U^\dagger)$. We say that $\Phi^c$ is a complementary channel to $\Phi$ if there exists a dilation unitary $U$ such that $\Phi^c$ is the complementary channel to $\Phi$ with respect to $U$. 

\begin{theorem}[Information-disturbance tradeoff {\cite[Theorem 3]{KSW08}}]
\label{importedtheorem:ksw}
Let $\cM \otimes \cA \otimes \cB$ be a finite-dimensional Hilbert space. Suppose that $T : \cM \to \cA$ is a quantum channel with complementary channel $T^c : \cM \to \cB$. Then there exists a decoding channel $\Dec : \cA \to \cM$ such that
\[
\norm{\Dec \circ T - \Id_M}_{\diamond} \le \min_{\bsigma \in \bfD(\cB)} 2 \norm{T^c - S_{\bsigma}}_{\diamond}
\]
where $S_{\bsigma}$ is the completely depolarizing channel that always outputs $\bsigma \in \bfD(\cB)$.
\end{theorem}

\begin{theorem}
A non-interactive quantum commitment scheme $\QSC$ is statistically swap-binding if and only if it is statistically extractable.
\end{theorem}
\begin{proof}
Suppose that $\QSC$ is statistically extractable. Then invoking \Cref{def:statistical-binding-extraction} with $U = \Id$ and initial state $\Pi \brho \Pi$,
\[
    \norm{(\Pi_{(\RegC,\RegD)} \circ \Ext_{(\RegC,\RegM')} \circ \Pi_{(\RegC,\RegD)})(\brho \otimes \ketbra{0}_{\RegM'}) - (\widehat{\SWAP[\RegM,\RegM']} \circ \Pi_{(\RegC,\RegD)})(\brho \otimes \ketbra{0}_{\RegM'})}_1 \le \negl(\lambda).
\]
Since $\Pi_{(\RegC,\RegD)}$ is a projection, it follows that
\[
    \norm{(\Ext_{(\RegC,\RegM')} \circ \Pi_{(\RegC,\RegD)})(\brho \otimes \ketbra{0}_{\RegM'}) - (\widehat{\SWAP[\RegM,\RegM']} \circ \Pi_{(\RegC,\RegD)})(\brho \otimes \ketbra{0}_{\RegM'})}_1 \le \negl(\lambda).
\]
In particular, in the $b=1$ world of the swap-binding security experiment (\Cref{def:swap-binding}) the adversary receives the state
\begin{align*}
& \Tr_{(\RegC,\RegM')}[(\widehat{\SWAP[\RegM,\RegM']} \circ \Pi_{(\RegC,\RegD)})(\brho \otimes \ketbra{0}_{\RegM'})] \\
&\approx_s \Tr_{(\RegC,\RegM')}[(\Ext_{(\RegC,\RegM')} \circ \Pi_{(\RegC,\RegD)})(\brho \otimes \ketbra{0}_{\RegM'})] \\
&= \Tr_{(\RegC,\RegM')}[\Pi_{(\RegC,\RegD)}(\brho \otimes \ketbra{0}_{\RegM'})].
\end{align*}
Since $\Tr_{(\RegC,\RegM')}[\Pi_{(\RegC,\RegD)}(\brho \otimes \ketbra{0}_{\RegM'})]$ is the state the adversary receives in the $b=0$ world, swap-binding follows.

For the other direction, suppose that $\QSC$ is statistically swap-binding. Let $T : \RegM \to \RegC$ be defined by
\[
T(\brho) = \Tr_\RegD(\Com(\brho \otimes \ketbra{0}_\RegW)).
\]
The complementary channel is
\[
T^c(\brho) = \Tr_\RegC(\Com(\brho \otimes \ketbra{0}_\RegW)),
\]
and since $\QSC$ is statistically swap-binding we have
\[
T^c(\brho) \approx_s T^c(\ketbra{0}_\RegM)
\]
for all $\brho \in \bfD(\cM)$. Applying \Cref{importedtheorem:ksw} with $\bsigma = T^c(\ketbra{0}_\RegM)$ yields a decoding channel $\Dec : \RegC \to \RegM$ such that
\[
\norm{\Dec \circ T - \Id_M}_{\diamond} \le \negl(\lambda).
\]
Let $\widetilde{\Dec}$ be a Stinespring dilation of $\Dec$, i.e., $\widetilde{\Dec} : (\RegC, \RegE) \to (\RegM, \RegF)$ is a unitary such that $\Tr_{\RegF}(\widetilde{\Dec}(\btau \otimes \ketbra{0}_\RegE)) = \Dec(\btau)$. We can now define our extractor $\Ext : (\RegC,\RegM') \to (\RegC,\RegM')$ for $\btau \in \bfD(\cC)$ as
\[
\Ext(\btau) = \Tr_\RegE[(\widetilde{\Dec}^\dagger \circ \SWAP[\RegM,\RegM'] \circ \widetilde{\Dec})(\btau \otimes \ketbra{0}_{(\RegE,\RegM')})].
\]
Since $\norm{\Ext \circ \Pi - \Pi \circ \Ext}_{\diamond} \le \negl(\lambda)$, it follows that
\begin{align*}
& (\Pi_{(\RegC,\RegD)} \circ U_{(\RegD,\RegH)} \circ \Ext_{(\RegC,\RegM')})(\btau \otimes \ketbra{0}_{\RegM'}) \\
&\approx_s (\Ext_{(\RegC,\RegM')} \circ \Pi_{(\RegC,\RegD)} \circ U_{(\RegD,\RegH)})(\btau \otimes \ketbra{0}_{\RegM'}) \\
&= \Tr_\RegE[(\widetilde{\Dec}^\dagger \circ \SWAP[\RegM,\RegM'] \circ \widetilde{\Dec} \circ \Pi_{(\RegC,\RegD)} \circ U_{(\RegD,\RegH)})(\btau \otimes \ketbra{0}_{\RegM'} \otimes \ketbra{0}_\RegE)] \\
& \approx_s (\widehat{\SWAP[\RegM,\RegM']} \circ \Pi_{(\RegC,\RegD)} \circ U_{(\RegD,\RegH)})(\btau \otimes \ketbra{0}_{\RegM'}). \qedhere
\end{align*}
\end{proof}

\newpage

\section{Quantum encryption with short keys from PRUs}
\label{sec:pru-expansion}
Pseudo-random unitaries (PRUs) are families of efficient unitaries that are indistinguishable from Haar random unitaries under black-box access \cite{C:JiLiuSon18}. In \cref{theorem:pru-expansion}, we show that the Hilbert space of a PRU can be expanded at the cost of restricting the security to a single use --- effectively converting a pseudo-superpolynomial-design to a pseudo-1-design on a larger space.

The construction is based on the Schur transform $\USch$ \cite{Har05}, which maps between the standard representation of $(\mathbb{C}^p)^{\otimes \ell}$ and the ``Schur-Weyl basis,'' in which both $\ell$-fold products of unitaries $U^{\otimes \ell}$ and permutations act naturally. In particular, the Schur transform maps between the symmetric subspace on $\ell$ qu$p$its and an $\binom{\ell+p-1}{\ell}$-dimensional Hilbert space on $\left\lceil \log_2 \binom{\ell+p-1}{\ell} \right\rceil$ qubits \cite{Har05}.

Given a PRU family $\{U_k\}_{k \in \{0,1\}^{d(\lambda)}}$ acting on $p$ dimensions, we build a one-time quantum encryption scheme $\{\Expand(U_k, \ell)\}_{k \in \{0,1\}^{d(\lambda)}}$ with $\binom{\ell+p-1}{\ell}$-dimensional messages $\ket{\psi}$ as follows:
\begin{enumerate}
    \item Initialize the appropriate registers to select the symmetric subspace in the Schur basis and the remaining registers to $\ket{\psi}$. In the notation of \cite{Har05}, this is $\ket{\Lambda=0, p_\Lambda=0, q_\Lambda=\psi}$.
    \item Apply $\USch (U_k)^{\otimes \ell} \USch^\dagger$ to $\ket{\Lambda=0, p_\Lambda=0, q_\Lambda=\psi}$.
    \item Return the $q_\Lambda$ register as the output of $\Expand(U_k, \ell)$.
\end{enumerate}

\begin{theorem}[PRU Expansion]
\label{theorem:pru-expansion}
If $\{U_k\}_{k \in \{0,1\}^{d(\lambda)}}$ is a PRU family on $p$ dimensions, then $\{\Expand(U_k, \ell)\}_{k \in \{0,1\}^{d(\lambda)}}$ is a secure one-time quantum encryption scheme with messages of dimension $\binom{\ell+p-1}{\ell} \ge \left(1+\frac{p-1}{\ell}\right)^\ell$ for any $\ell = \poly(\lambda)$.
\end{theorem}
\begin{proof}
Since the symmetric subspace is invariant under $(U_k)^{\otimes \ell}$, the auxiliary registers $\Lambda$, $p_\Lambda$ are unaffected by $\USch (U_k)^{\otimes \ell} \USch^\dagger$. Letting $\rho$ be the state $\ket{\Lambda=0, p_\Lambda=0, q_\Lambda=\psi}$,
\begin{align*}
    \USch (U_k)^{\otimes \ell} \USch^\dagger \rho \USch (U_k^\dagger)^{\otimes \ell} \USch^\dagger &\approx_c \USch \left(\int_{U \from \mu(\mathbb{C}^p)} (U)^{\otimes \ell} \USch^\dagger \rho \USch (U^\dagger)^{\otimes \ell} \ dU\right) \USch^\dagger \\
    &= \binom{\ell+p-1}{\ell}^{-1} \ketbra{\Lambda=0, p_\Lambda=0}.
\end{align*}
The first line follows from the security of the PRU family $\{U_k\}_{k \in \{0,1\}^{d(\lambda)}}$. The equality follows from Schur's lemma, since the symmetric subspace on $(\mathbb{C}^p)^{\otimes \ell}$ is an irreducible representation of the unitary group under the action $U \mapsto U^{\otimes \ell}$ \cite{Har13}. Schur's lemma tells us that the integral is maximally mixed over the symmetric subspace, and in the Schur-Weyl basis the maximally mixed state over the symmetric subspace is the state where $q_\Lambda$ is random and $\Lambda=0, p_\Lambda=0$.
\end{proof}

Similarly, we note that a many-time secure PRS family $\{\ket{\phi_k}\}_{k \in \{0,1\}^{d(\lambda)}}$ on $p$ dimensions may be converted to a one-time secure PRS family on $\binom{\ell+p-1}{\ell}$ dimensions via $\Tr_{\Lambda,p_\Lambda}(\USch (\ketbra{\phi_k})^{\otimes \ell} \USch^\dagger)$.

\begin{theorem}[PRS Expansion]
If $\{\ket{\phi_k}\}_{k \in \{0,1\}^{d(\lambda)}}$ is an $\ell$-time secure PRS family on $p$ dimensions, then $\{\Expand(\ket{\phi_k}, \ell)\}_{k \in \{0,1\}^{d(\lambda)}}$ is a one-time secure PRS family on $\binom{\ell+p-1}{\ell} \ge \left(1+\frac{p-1}{\ell}\right)^\ell$ dimensions.
\end{theorem}

\end{document}